%% file: root.tex
\documentclass[journal,twoside,web]{ieeecolor}
\usepackage{generic}
\usepackage{cite}
\usepackage{amsmath,amssymb,amsfonts}
\usepackage{graphicx}
\usepackage{textcomp}
\def\BibTeX{{\rm B\kern-.05em{\sc i\kern-.025em b}\kern-.08em
    T\kern-.1667em\lower.7ex\hbox{E}\kern-.125emX}}
\input{defpack}

\begin{document}
\title{No-Regret Learning in \\Dynamic Stackelberg Games}
\author{Niklas Lauffer,  Mahsa Ghasemi, Abolfazl Hashemi, Yagiz Savas, and Ufuk Topcu
\thanks{This work was supported in part by AFRL grant FA9550-19-1-0169 and DARPA grant D19AP00004.}
\thanks{Niklas Lauffer is with the Department of Electrical Engineering and Computer Sciences at the University of California, Berkeley, Berkeley, CA 94720, USA. Mahsa Ghasemi and Abolfazl Hashemi are with The Elmore Family School of Electrical and Computer Engineering, Purdue University, West Lafayette, IN 47907, USA. Yagiz Savas and Ufuk Topcu are with the Department of Aerospace Engineering and Engineering Mechanics, The University of Texas at Austin, Austin, TX 78712, USA.}
}

\maketitle

\begin{abstract}
In a \textit{Stackelberg game}, a \textit{leader} commits to a randomized strategy, and a \textit{follower} chooses their best strategy in response.
We consider an extension of a standard Stackelberg game, called a \textit{discrete-time dynamic Stackelberg game}, that has an underlying state space that affects the leader's rewards and available strategies and evolves in a Markovian manner depending on both the leader and follower's selected strategies.
Although standard Stackelberg games have been utilized to improve scheduling in security domains, their deployment is often limited by requiring complete information of the follower's utility function.
% We assume that the follower acts according to a linearly parameterized utility function that is \textit{unknown} to the leader. 
In contrast, we consider scenarios where the follower's utility function is \textit{unknown} to the leader; however, it can be linearly parameterized.
Our objective then is to provide an algorithm that prescribes a randomized strategy to the leader at each step of the game based on observations of how the follower responded in previous steps. 
We design a \textit{no-regret} learning algorithm that, with high probability, achieves a regret bound (when compared to the best policy in hindsight) which is sublinear in the number of time steps; the degree of sublinearity depends on the number of features representing the follower's utility function. The regret of the proposed learning algorithm is independent of the size of the state space and polynomial in the rest of the parameters of the game.
We show that the proposed learning algorithm outperforms existing model-free reinforcement learning approaches.
\end{abstract}
%\section{Introduction}

%\section{Related work}

\section{Introduction} \label{sec:introduction}
Stackelberg games model strategic interactions between two agents, a leader and a follower \cite{stackelberg1934market,fudenberg1998theory}. The leader plays first by committing to a randomized strategy. The follower observes the leader's commitment and then plays a strategy to best respond to the leader's chosen strategy.

Stackelberg games have been successfully applied to a multitude of scenarios to model competition between firms \cite{stackelberg1934market}, improve security scheduling \cite{paruchuri2008playing}, and dictate resource allocation \cite{sinha2018stackelberg}. The Los Angeles Airport and other security agencies have deployed randomized patrol routes based on Stackelberg models \cite{tambe2011}. Researchers have also used Stackelberg models to improve national park wildlife ranger patrol patterns and resource distribution to protect against illegal poaching \cite{yang2014}.

% In most applications of Stackelberg games, the environment is assumed to remain constant or changes in the environment over time are handled by myopically recomputing strategies.

% \fromnik{Smooth transitions into the poaching example here.}
\emph{Repeated Stackelberg games} \cite{learning_security_games}  are used to model agents that repeatedly interact in a Stackelberg game.
A significant drawback of standard repeated Stackelberg game formulations is their inability to model dynamic scenarios. In scenarios involving repeated interaction between two strategic agents, decisions made early on can sometimes have long-lasting effects. 
% For example, wildlife park rangers may have a limited supply of resources that can be deployed over some time horizon. Spending too many resources early on will limit the available resources later to defend against poachers in the future. 
Moreover, repeated Stackelberg formulations cannot model dynamic changes in the leader's preferences or available strategies over time. 
% For instance, the objectives of a firm may change over time or in response to the decisions of competing firms.

Consider a scenario in which park rangers are responsible for protecting a geographical area containing different types of animals from illegal poaching. The geographical area is split into distinct regions, each of which contains a different density of the animals. Each month, the park rangers can deploy a mixed strategy to decide which of the regions to patrol. After observing the rangers' strategy, the poachers can attempt to lay snares in any one of the regions. If the poachers attempt to lay snares in a region being patrolled by the park rangers, they are caught and penalized. Otherwise, the snare has a chance of catching one of the animals in that region based on the density of each of the animals.

The park ranger's resources are limited. Over the course of each year, the park rangers are allocated a budget for anti-poaching patrols. The further a region is from the park headquarters, the more expensive it is to patrol. If the park rangers deplete their budget, they are unable to launch any more patrols until the following year.
Moreover, the damage that poaching incurs on an animal population can fluctuate throughout the year. For example, the mating season for rhinos typically occurs in the months of October and November during which poaching is especially damaging to the rhino population. Such a dynamic, time-varying setting requires a model with an underlying state space. 

In order to address the limitations of existing formulations, we introduce \textit{discrete-time dynamic Stackelberg games (DSGs)}, an extension of the standard repeated Stackelberg formulation that includes an underlying, persistent discrete-time state space. The state space allows DSGs to model scenarios with dynamic utility functions, resources, and other states. 
%In order to model scenarios with dynamic preferences or resources, we introduce an extension of the standard Stackelberg game model, called a \textit{dynamic Stackelberg game (DSG)}, that incorporates an underlying persistent state. 
The state can vary across steps of the DSG in a Markovian manner dependent on the actions of both the leader and the follower. In a DSG, the state is only relevant to the leader. The leader's reward and  the set of available actions are directly dependent on the current state, but the follower's are not.

Modeling a scenario as a standard Stackelberg game (and thus also a DSG) requires complete knowledge of the follower's payoff in different scenarios. In practice, the parameters of such models are often estimated based on historical data \cite{demiguel2009stochastic}. However, small modeling imperfections can lead to significant inefficiencies in the utility of the strategies derived from solving the Stackelberg game. 
Arguably, the most severe modeling limitation is that in order to compute the optimal strategy, the leader must precisely know the follower's utility function.

We study the online learning problem of sequentially synthesizing policies for the leader in a DSG under a setting in which the follower's utility function is unknown. In this setting, the leader must interact with the follower and incrementally learn its behavior over time by using past interactions with the follower to improve future decisions.

% \fromnik{In this work we bridge a gap between MDPs and Stackelberg games...} 

% \fromnik{Reference other dynamic/stochastic Stackelberg game formulations here.}

\subsection{Contributions}
%\paragraph{Our Setting and Results.}

% We consider a discrete-time dynamic Stackelberg game with finite state and action spaces.
% The leader has total information about how state affects its reward function and the transition function between states. 
% The follower's utility function, however, is unknown to the leader. Since the follower's utility function is independent of the state space, it is natural to desire a regret bound independent of the size of the state space. This precludes the use of existing online learning algorithms for repeated Stackelberg games since an independent game would have to be solved for each state in the state space. It also prevents abstracting away the reward structure induced by the game and treating it as a model-free reinforcement learning problem over a Markov decision process.
% \frommahsa{This paragraph is more suitable for the introduction.} \fromyagiz{I agree. I think this should be the first paragraph. First clearly explain what problem is solved in the paper, then elaborate on it explaining the motivation, limitations in the literature etc.}

We introduce a new modeling formalism, called a discrete-time dynamic Stackelberg game (DSG), that bridges the modeling formalism of repeated Stackelberg games and Markov decision processes. We propose an online learning algorithm for computing an adaptive policy for the leader in a scenario in which the follower's reward function is unknown. The proposed algorithm enables the leader to create an estimate of the follower's utility function and use that estimate to update its policy. 
% We adopt the commonly employed technique of \textit{linear function approximation} \cite{jinlinear2020} to model the follower's utility function. We assume that the leader has access to a known low-dimensional feature map and that the follower makes decisions based on a utility function that can be expressed as a linear function of the features. In many scenarios, domain specific knowledge can be used to construct such feature maps \cite{jinlinear2020, ng2000algorithms}.
% The leader keeps track of a version space representing the feasible parameters that encode the follower's utility function and shrinks the version space as it observes the follower's actions over time. It then computes a new policy by solving for a value function that is optimistic with respect to the current version space.
% \fromnik{Mention regret and bandit algorithms here, give a citation.}
We prove that the proposed algorithm achieves a sublinear regret bound with respect to the time horizon  --- establishing the first no-regret online learning algorithm in this setting. The regret bound's degree of sublinearity depends on the number of features representing the follower's utility function. It is independent of the size of the state space and polynomial in the rest of the parameters of the game.
Through a series of experiments, we evaluate the empirical performance of our algorithm and demonstrate the practical application of our approach.

\subsection{Related Work}

% \fromnik{maybe cite a few more papers here that study learning in Stackelberg games because those are the people that we want reviewing our submission.}
Generating policies for repeated Stackelberg games in a setting where the follower's utility function is unknown to the leader is an active area of research. In this setting, the agents play in some variation of a \textit{repeated Stackelberg game} where the leader makes decisions based on observations from previous steps in the game. The authors in \cite{letchford2009, marecki2012} study this problem with the primary objective of designing an algorithm that learns the follower's utility with low \textit{sample complexity}, i.e., in as few repetitions of the game as possible. Then, the learned representation can be used to play near optimally for the rest of time. Other variations of this problem include an online setting in \cite{blum2018}, where the follower's utility function can change (potentially adversarially) over time and the objective is to design a no-regret learning algorithm, that is, an algorithm which asymptotically converges to the optimal strategy. However, none of these works consider a dynamic scenario.
Since the follower's utility function in a dynamic Stackelberg game (DSG) is independent of the state space, it is natural to desire a regret bound independent of the size of the state space. This precludes the use of existing online learning algorithms for repeated Stackelberg games since an independent game would have to be learned and solved for each state in the state space.

DSGs are also closely related to \textit{stochastic games}, a well studied class of game played over a state-space where payoffs and transitions are also determined by the actions from a pair of agents \cite{neyman2003stochastic}. Online learning in stochastic games is also an ongoing area of research \cite{wei2017}. Closely related to our work, the authors in \cite{ouyang2017dynamic} give a general solution concept for stochastic dynamic games with asymmetric or unknown information. In contrast to DSGs, agents in a stochastic game choose actions simultaneously such that neither agent can observe its opponent's action before selecting its own action.

DSGs are similar to existing \textit{feedback Stackelberg games} \cite{Li2017} 
% \fromnik{maybe a few more references here about other Stackelberg games} 
in which two agents play in a Stackelberg game over a state space in which individual policies are determined for each step of the game. The authors in \cite{chen1972stackelberg} introduce a solution concept for feedback Stackelberg games in which the leader knows both utility functions and the follower only knows their own utility without considering any learning. Feedback Stackelberg games are typically studied in continuous settings modeled by differential equations with perfect information, that is, policies are computed with full knowledge of the follower's utility function. In contrast, we study a setting in which the follower's utility function is unknown to the leader.

% \frommahsa{I think it would be better to place the previous two paragraphs after the initial discussion related to repeated Stackelberg games. That way, relevant Stackelberg models will be close to each other, and the MDP discussion will come after all Stackelberg-related paragraphs.}

% DSGs can alternatively be viewed as a special case of an \textit{adversarial Markov decision process} with unknown reward and transition structure that can be represented by a Stackelberg game. \frommahsa{I do not understand ``that can be represented by a Stackelberg game.''} \fromnik{Leave this out or move it elsewhere? Don't confuse the reader.}
\textit{Markov decisions processes (MDPs)} are commonly used to model a single decision-making agent interacting with an environment \cite{puterman2014markov}. Online learning in MDPs is an active area of research. The authors in \cite{neu2013online} give an online learning algorithm in an MDP where the rewards are unknown and chosen by an adversary, but the transitions are known. 
By disregarding the reward and transition structure induced by the Stackelberg game, a DSG can be reduced to a Markov decision process (MDP) with an unknown reward and transition function. This makes it possible to treat the learning problem in a DSG as a \textit{reinforcement learning} problem in an MDP. However, reinforcement learning algorithms
% suffer from regret that scales with the size of the state space \cite{model-free-learning-regret}, this approach is unattractive. 
fail to incorporate the reward and transition structure inherent in the DSG, their learning rate scales with the size of the state space. Moreover, if the leader in the DSG has action space $\leadacts$, the resulting MDP after the transformation has a continuous action space in $\R^{|\leadacts|}$, making this approach intractable even for small $|\leadacts|$.
In Section \ref{sec:experimental}, we directly compare our learning algorithm against a model-free reinforcement learning approach.

\subsection{Organization}
After giving a formal construction of dynamic Stackelberg games and a formal description of the learning problem in Section \ref{sec:prob_def}, we introduce a novel algorithm in Section \ref{sec:alg} based on optimistically choosing policies that are consistent with previous observations. In Section \ref{sec:analysis}, we analyze the algorithm and show that it achieves a regret that, with high probability, is sublinear in the number of time-steps.
% , where the degree of sublinearity depends on the number of features representing the follower's utility function. 
In Section \ref{sec:experimental}, we experimentally demonstrate the impact of varying parameters of the game on regret and compare against a model-free reinforcement learning approach.

\section{Problem Definition} \label{sec:prob_def}

\begin{figure}[t]
    \fbox{
        \parbox{.945\linewidth}{
        \textbf{Parameters}: Discrete-time dynamic Stackelberg game $(\states, \leadacts, \followacts, r, u, P)$ and time horizon $T$.\\
        \textbf{For all episodes $t = 1, 2, \dots, T$, repeat}
        \begin{addmargin}[2em]{0em}
            \textbf{For all steps $h = 1, 2, \dots, H$, repeat}
                \begin{enumerate}%[noitemsep, nolistsep]
                    \item The learner observes the current state $s \in \states_h$.
                    \item Based on previous observations, the learner selects mixed strategy $\x \in \Delta(\leadacts(s))$.
                    \item The follower responds with action $b \in \followacts$ that maximizes its expected utility $\E_{a \samp \x} \left[ u(a,b) \right]$. The leader observes action $b$.
                    \item An action $a \samp \x$ is sampled. 
                    \item The learner obtains reward $r(s, a, b)$ and the environment transitions to the next state ${s' \samp P(\cdot \mid s, a, b)}$.
                \end{enumerate}
        \end{addmargin}
        }
    }
    \caption{The protocol for online learning in a discrete-time dynamic Stackelberg game.}
    \label{protocol:main}
\end{figure}

A discrete-time dynamic Stackelberg game (DSG) is played between a \textit{leader} agent and a \textit{follower} agent.
The game is played sequentially on a 6-tuple $(\states, \leadacts, \followacts, r, u, P)$ whose elements are defined as follows.
\begin{itemize}
    \item $\states$ is a set of states.
    \item $\leadacts = \{a_1, ..., a_n\}$ is a set of actions available to the leader. We denote the leader's available actions at a state $s$$\in$$\states$ by $\leadacts(s)$$\subseteq$$\leadacts$.
    \item $\followacts =\{b_1, ..., b_m\}$ is a set of actions available to the follower.
    \item $r : \states \times \leadacts \times \followacts \to \R$  is the reward function for the leader.
    \item $u : \leadacts \times \followacts \to \R$ is the utility function for the follower. (Notice the lack of dependency on the state space $\states$.)
    \item $P : \states \times \leadacts \times \followacts \times \states \to [0,1]$ is the transition function that satisfies $\sum_{s'\in \states}P(s,a,b,s')=1.$ 
    % for every $s \in \states,$ $a \in \leadacts,$ and $b \in \followacts.$
\end{itemize}

We outline the interactions between the leader and the follower in a DSG in Fig. \ref{protocol:main}. First, the leader chooses a mixed strategy $\x_s$$\in$$\Delta(\leadacts(s))$ in state $s$$\in$$\states$, where $\Delta$ is the probability simplex of appropriate dimension. Then, the follower chooses an action $b$$\in$$\followacts$ in response such that  
\begin{equation}\label{follower_response}
    b \in \argmax_{b'\in \followacts} \E_{a \samp \x_s} \left[ u(a, b') \right].
\end{equation}
Note that in the considered setting, the follower is myopic and aims to maximize its immediate expected utility. 

An action $a$$\samp$$\x_s$ is sampled from the leader's mixed strategy to determine the next state. We limit the scope of this paper to the class of \emph{episodic} DSGs. That is, we assume that the state space $\states = \states_1 \cup \dots \cup \states_H$ is divided into $H$ disjoint \textit{layers} with $\states_1$$=$$\{s_1\}$, and the game is separated into a series of $T$ episodes. We denote by $s_{t,h}$$\in$$\mathcal{S}_h$ the state occupied in the $h$th step of episode $t$. Then, at a state $s_{t,h}$$\in$$\states_h$, the next state $s_{t,h+1} \in \states_{h+1}$ is guaranteed to be contained within the next layer, i.e., $\sum_{s_{t,h+1}\in \states_{h+1}}P(s_{t,h}, a, b, s_{t,h+1})$$=$$1$. 

Suppose that, at a state $s_{t,h}$, the leader takes the action $a_h$$\samp$$\x_{s_{t,h}}$, and the follower takes the action $b_h$ in response. Then, the leader receives the reward $r_h = r(s_{t,h}, a_h, b_h)$ for step $h$ of the episode. Consequently, the leader's cumulative reward for the episode becomes $\sum_{h=1}^{H} r_h$. 

Our objective in this paper is to design an algorithm which the leader can use to learn a strategy that maximizes its \textit{expected} cumulative reward over all episodes. To represent the leader's expected reward compactly, we can think of the follower's response as a function $\varphi$ $:$ $\Delta(\leadacts) \to \followacts$ where
\begin{equation}
    \varphi(\x) = \argmax_{b \in \followacts} \E_{a \samp \x} \left[ u(a, b) \right].
\end{equation}
Then, the leader's reward function over actions induces an auxiliary reward function $R : \states \times \Delta(\leadacts) \to \R$ over mixed strategies given by
\begin{equation} \label{eq:aux_reward}
    R(s, \x) = \E_{a \samp \x} \left[ r(s, a, \varphi(\x)) \right].
\end{equation}
Let 
% $\pi$$=$$\{\x_{t,h} : t\in [T], h\in [H]\}$ 
$\pi : \states \to \Delta(\leadacts)$ be a policy 
% be a sequence of mixed strategies that describes the leader's action selection for all steps in all episodes. 
that describes the leader's action selection for each state in the DSG.
We denote by $\Pi$ the set of all 
% admissible sequences $\pi$. 
policies $\pi$.
The optimal policy for a leader \textit{that knows the follower's utility function $u$} is given by
 \begin{equation}\label{optimal_strategy}
   \pi^{\star}\in \argsup_{\pi \in \Pi} \mathbb{E}^{\pi}[ \sum_{t=1}^T \sum_{h=1}^H R(s_{t,h}, \pi(s_{t,h}) ) ],
  \end{equation}
where the expectation is taken over the trajectories (represented by $s_{t,h}$) induced by the leader's policy $\pi$, the follower's corresponding response as defined in \eqref{follower_response}, and the stochasticity in the transition function of the DSG. In this paper, we assume that the follower's utility function is \textit{unknown} to the leader. Hence, the leader needs to \textit{learn} a sequence of policies through interactions with the follower.

%\begin{remark}
%  The policy in \eqref{optimal_strategy} corresponds with the \textit{subgame perfect equilibrium} of the DSG.
%\end{remark}

We use \textit{regret} \cite{lattimore2020bandit} to evaluate the asymptotic performance of a learning algorithm.
% Regret is a metric used to evaluate the asymptotic performance of learning agents in online learning settings. 
Regret measures the difference between the cumulative reward of a learning agent and the cumulative reward of the best strategy in hindsight. In our case, the best strategy in hindsight is the best strategy given that the follower's utility function was known to the leader ahead of time, i.e., the policy in \eqref{optimal_strategy}.
\begin{definition}[Regret]
The regret $R_T$ of a sequence of policies $\{\pi_{t,h} \in \Pi : t\in [T], h\in [H]\}$ is given by
  \begin{equation}
    \mathbb{E}^{\pi^{\star}}[ \sum_{t=1}^T \sum_{h=1}^H R(s_{t,h}, \pi^{\star}(s_{t,h}) ) ]
    -\mathbb{E} [ \sum_{t=1}^T \sum_{h=1}^H R(s_{t,h}, \pi_{t,h}(s_{t,h})) ].
    \label{def:regret}
  \end{equation}

%\fromyagiz{mixed strategy, mixed strategy, or mixed strategy? :) }

\end{definition}

The main problem investigated in this paper is as follows.
\begin{problem} \label{prb:learning}
    Suppose the follower's utility function $u$ is fixed and unknown to the leader.
    Provide an online learning algorithm that computes a sequence of policies that minimizes the leader's regret $R_T$.
\end{problem}

\subsection{Assumptions}
Before proceeding with the learning algorithm that solves Problem \ref{prb:learning}, we list our assumptions on the interactions between the leader and the follower.

\begin{assumption}
    The range of the leader's reward function is $[0,1] \subset \R$, i.e., $r(s,a,b)$$\in$$[0,1]$ for all $s$$\in$$\states$, $a$$\in$$\leadacts$, and $b$$\in$$\followacts$.
\end{assumption}
We introduce the above assumption for notational simplicity.
Note that if the range of the reward function instead falls in some other closed interval $[a,b]$, rewards can be normalized, without loss of generality, to lie in the interval $[0,1]$.

%As in standard Stackelberg games, we also assume the following.
%\begin{assumption}
 %   The follower has a utility function $u : \leadacts \times \followacts \to \R$ and its policy maximizes the expected utility over this function.
%\end{assumption}
%Notice the lack of dependence on the state space $\states$. At each step of the game, the follower maximizes its immediate payoff.
%We can think of the strategic agent's policy $\varphi$ as a function $\Delta(\leadacts) \to \followacts$ defined as,
%\begin{equation}
%    \varphi(\x) = \argmax_{b \in \followacts} \E_{a \samp \x} \left[ u(a, b) \right].
%\end{equation}
%With this assumption, the leader's reward function over actions induces an auxiliary reward function $R : \states \times \Delta(\leadacts) \to \R$ over mixed strategies given by
%\begin{equation} \label{eq:aux_reward}
%    R(s, \x) = \E_{a \samp \x} \left[ r(s, a, \varphi(\x)) \right].
%\end{equation}
We adopt the solution concept of \emph{strong Stackelberg equilibrium}, that ties are broken in the favor of the leader.
\begin{assumption}[Strong Stackelberg equilibrium]
    Let $\tau(\x)$$\subset$$\followacts$ represent the set of the follower's best responses to a leader's mixed strategy $\x$$\in$$\Delta(\leadacts)$. In any state $s$$\in$$\states$, if the leader plays mixed strategy $\x$, the follower is guaranteed to play an action $b$$\in$$\tau(\x)$ such that,  for all $b'$$\in$$\tau(\x)$,
    \begin{equation}
        \E_{a\samp \x}[r(s,a,b)] \geq \E_{a\samp \x}[r(s,a,b')].
    \end{equation}
   
\end{assumption}
Strong Stackelberg equilibrium is often adopted as the solution concept for Stackelberg games since it ensures the existence of an optimal strategy \cite{Stengel2004LeadershipWC}.

\begin{assumption}[Linear function approximation]
    The follower's utility function is linearly parameterized. That is, there exists a feature mapping $\feat : \leadacts \times \followacts \to \R^p$ that is known to the leader such that, for some $\thetaV^* \in \R^p$, 
    \begin{equation}
        u(a,b) = \langle \feat(a,b) , \thetaV^* \rangle.
    \end{equation}
\end{assumption}
Linear function approximation is commonly used in reinforcement learning \cite{jin20a, silver2007} and online learning \cite{Neu2020}. 

Finally, for ease of notation, we define the following matrices, which we call \textit{feature matrices}, for each action $b \in \followacts$. $\featM_b \in \R^{n \times p}$ such that
\begin{equation}
    [\featM_b]_{i} = f(a_i,b).
\end{equation}
% \fromnik{address how we scale $f$ here to get rid of constant.}

\subsection{The Best Strategy in Hindsight}
Before discussing our algorithm for the general case, we first investigate the case in which the follower's utility function is known to gain some intuition.

If the follower's utility function is known, the leader's optimal policy $\pi^{\star} : \states \to \Delta(\leadacts)$ is the solution to a bilevel optimization problem.
By extending the linear program used in \cite{korzhyk2010} to solve a standard Stackelberg game, the optimal mixed strategy determined by policy $\pi^{\star}$ in state $s$ is given by the solution to the following optimization problem.
% \fromyagiz{Are we sure about the correctness of this formulation? Should it not be minimize $V$ subject to $V$$\geq$$\mathbb{E}[]$? Also, we probably can remove $\pi^{\star}$ from notation.}
\begin{subequations}
\label{eq:opt_policy}
\begin{align}
    \displaystyle&\max\limits_{\x_s} \quad  V(\pi^{\star},s)        \\
   & \textrm{s.t.} \; \;
     V(\pi^{\star}, s) = \displaystyle\E\limits_{a \samp \x_s} \left[r(s,a,b) + \displaystyle\sum\limits_{s'} P(s,a,b,s') V(\pi^{\star}, s') \right] &\\
    & \x_s^T (\featM_{b} - \featM_{b'}) \thetaV^* \geq 0, \ \forall b' \in B \\
    & \x_s \in \Delta(\leadacts(s))
    \label{eq:opt_policy:epsilon_bound}
\end{align}
\end{subequations}
maximized over all $b \in \followacts$ where $V(\pi^{\star}, s')$ is the real value computed from the solution of the problem for $s' \in \states$. Then we define $\pi^{\star}(s) = \x_s$. The optimization problem for $s \in \states_h$, i.e. states in layer $h$, relies on the solution $V(\pi^{\star}, s')$ to the optimization problem for states in the next layer $\states_{h+1}$.
Therefore, we can efficiently solve (\ref{eq:opt_policy}) for each layer of the state space through backwards induction, reducing (\ref{eq:opt_policy}) to a linear program (LP).

The best policy in hindsight is the best policy had the follower's utility function been known from the beginning. Therefore, the solutions to the set of LPs in \eqref{eq:opt_policy} give the best policy in hindsight and represent the subgame perfect equilibrium of the DSG.

\section{The Algorithm} \label{sec:alg}

First we consider the simpler case of learning \textit{pure strategies}. In this case, Problem \ref{prb:learning} has a simple solution. The leader only needs to learn the function $\varphi : \leadacts \to \followacts$ such that
\begin{equation}
    \varphi(a_t) = \argmax_{b_t} u(a_t, b_t).
\end{equation}
Even if the utility function $u$ is unknown, $\varphi$ can be exactly determined by testing each input in $\leadacts$ only once. This strategy does not work for the general case since the continuous space of mixed strategies cannot be enumerated. 

% \fromnik{Add more figures from slides?}
Recall that the follower's utility function follows the relationship $u(a,b) = \langle \feat(a,b) , \thetaV^* \rangle$ for a fixed vector $\thetaV^*$$\in$$\R^p$.
Every time a mixed strategy $\x$$\in$$\Delta(\leadacts)$ is played and a response $b \in \followacts$ from the follower is observed, we gain more information about the nature of the follower's utility function. In particular, we know that $\forall b' \in \followacts$,
\begin{subequations} \label{eq:sample_ineq}
\begin{align} 
    \E_{a \samp \x} [u(a, b)] & \geq  \E_{a \samp \x} [u(a, b')] \\
    \E_{a \samp \x} [\langle f(a, b), \thetaV^* \rangle] &\geq  \E_{a \samp \x} [ \langle f(a, b') , \thetaV^* \rangle ] \\
    %  \langle \x^T \featM_b, \thetaV^* \rangle &\geq   \langle \x^T \featM_{b'}, \thetaV^* \rangle  \\
      \x^T \featM_b\thetaV^* &\geq   \x^T \featM_{b'} \thetaV^* \\
      \x^T (\featM_b - \featM_{b'}) \thetaV^* &\geq 0.
\end{align}
\end{subequations}
If $\thetaV^*$ is unknown, then Equation \eqref{eq:sample_ineq} lets us interpret $(\x, b)$ as a sample for the halfspace parameterized by $\thetaV^*$.
Let $\{(\x_i , b_i)\}_{i \in I}$ represent the set of mixed strategies and actions chosen in response from the previous steps of the game.
We maintain a \textit{version space} $\Theta$$\subset$$\R^p$ of the possible values of $\thetaV^*$ given the set of previous plays $\{(\x_i, b_i)\}_{i \in I}$. Specifically, the version space is the convex region
\begin{equation}
    \Theta = \{\thetaV \in \R^p \mid ||\thetaV|| = 1 \land \forall i,b' \ \x_i^T (\featM_{b_i} - \featM_{b'}) \thetaV \geq 0 \}.
\end{equation}
Before each episode of the game, we optimistically solve for the optimal \textit{$\epsilon$-conservative policy} $\pi_t$. That is, for each state $s$$\in$$\states$ we solve for a triple $(\x_s, \thetaV_s, b_s)$ such that \textit{if} $\thetaV_s$ were the true parameterization of $u$, then the follower would respond to the mixed strategy $\x_s$ with $b_s$. We constrain $\thetaV_s$$\in$$\Theta$ to lie within the current version space. The policy is considered \textit{optimistic} since $\thetaV_s$ can take on any possible value that is consistent with previous observations. The mixed strategy $\x_s$ is computed using backwards induction in the following way.

For each state $s$, we compute $(\x_s, \thetaV_s)$ along with an \textit{optimistic value} $\estV_t(s)$ as the solution to the optimization problem
\begin{subequations}
\label{eq:cons_opt}
\begin{align}
    \displaystyle\max\limits_{\x_s,\thetaV_s} \quad & \estV_t(s)        \\
    \textrm{s.t.} \quad
    & \thetaV_s \in \Theta \label{eq:cons_opt:theta} \\
    & \x_s \in \Delta(\leadacts(s)) \\
    & \estV_t(s) = \displaystyle\E\limits_{a \samp \x_s} \left[r(s,a,b) + \displaystyle\sum\limits_{s'} P(s,a,b,s') \estV_t(s') \right] &\\
    & \x_s^T (\featM_{b} - \featM_{b'}) \thetaV_s \geq \epsilon, \ \forall b' \in \followacts, b' \neq b, \label{eq:cons_opt:epsilon_bound}
\end{align}
\end{subequations}
maximized over $b$$\in$$\followacts$. That is, we solve (\ref{eq:cons_opt}) for each $b$$\in$$\followacts$ and choose the solution for which $\estV_t(s)$ is maximized. Denote this distinguished action by $b_s$. The chosen policy is considered \textit{$\epsilon$-conservative} since constraint (\ref{eq:cons_opt:epsilon_bound}) enforces an $\epsilon$-size margin in the follower's decision boundary.

\begin{figure}[t]
    \fbox{
        \parbox{.945\linewidth}{
        \textbf{Parameters}:  Discrete-time dynamic Stackelberg game $(\states, \leadacts, \followacts, r, u, P)$ and time horizon $T$.\\
        $\Theta \gets \{\thetaV \in \R^p \mid ||\thetaV|| = 1\} $. \\
        \textbf{For all episodes $t = 1, 2, \dots, T$, repeat}
        \begin{addmargin}[2em]{0em}
            $\pi_t \gets \textproc{GetPolicy}(t, \Theta)$. \\
            \textbf{For all steps $h = 1, 2, \dots, H$, repeat}
                \begin{enumerate}%[noitemsep, nolistsep]
                    \item The learner observes the current state $s \in \states_h$.
                    \item The learner selects mixed strategy $\x \gets \pi_t(s)$.
                    \item The leader observes the follower's action $b$ and updates $\Theta \gets \textproc{Update}(\x,b,\Theta)$.
                    \item An action $a \samp \x$ is sampled. 
                    \item The learner obtains reward $r(s, a, b)$ and the environment transitions to the next state ${s' \samp P(\cdot \mid s, a, b)}$.
                \end{enumerate}
        \end{addmargin}
        }
    }
    \caption{The procedure for using the learning scheme to compute policies in a discrete-time dynamic Stackelberg game.}
     \label{procedure:main}
\end{figure}

The optimization problem for $s$$\in$$\states_h$, i.e. states in layer $h$, relies on the solution $\estV_t(s')$ to the optimization problem for states in the next layer $\states_{h+1}$.
Therefore, we can efficiently solve (\ref{eq:cons_opt}) for each layer of the state space by computing the final layer $\states_H$ first, and then iterating backwards.

This learning scheme solves Equation \eqref{eq:cons_opt} exactly $|\states|\cdot|\followacts|$ times to compute the triple $(\x_s, \thetaV_s, b_s)$ for each state in $\states$. Function \textproc{GetPolicy} in Algorithm \ref{alg:main} demonstrates how the policy is computed. These solutions give rise to a policy $\pi_t$ represented by the mixed strategy $\x_s$ at each state and an associated \textit{estimated value function} $\estV_t(s)$ at each state.

After computing a policy $\pi_t$ for episode $t$ in the game, it is used to play mixed strategies for the course of the episode. Beginning with the initial state $s_1$, the leader plays mixed strategy $\pi_t(s_1)$. After observing the follower's response $b$ to $\pi_t(s_1)$, the algorithm updates the version space $\Theta$ with the new information. Irregardless of which action the policy $\pi_t$ expected the follower to respond with, it is now known that $\x^T (\featM_{b} - \featM_{b'}) \thetaV \geq 0$. Therefore, the version space can be updated to 
\begin{equation}
    \Theta \cap \{\thetaV \in \R^p \mid ||\thetaV|| = 1 \land \forall b' \ \x^T (\featM_{b} - \featM_{b'}) \thetaV \geq 0 \}.
\end{equation}

\begin{algorithm}[t!]
  \caption{Subroutines used in the learning scheme. $P_{\epsilon}$ denotes the program in Equation \eqref{eq:cons_opt}.
    \label{alg:main}}
  \begin{algorithmic}[1]
    \Require{Dynamic Stackelberg game $(\states, \leadacts, \followacts, r, u, P)$.}
    \Statex
    \Function{GetPolicy}{$t, \Theta$}
      \For{$s \in \states$}
        \State $\pi_t(s) \gets 0$
        \State $\estV_t(s) \gets 0$
      \EndFor
      \For{$h \gets H \textrm{ to } 1$}
        \For{$s \in \states_h$}
            \State $\pi_t(s), \estV_t(s) \gets P_{\epsilon}(\Theta,\estV_t,s)$
        \EndFor
      \EndFor
      \State \Return{$\pi_t$}
    \EndFunction
    \Statex
    \Function{Update}{$\x, b, \Theta$}
        \State \Return{$\Theta \cap \{\thetaV \in \R^p \mid ||\thetaV|| = 1 \land \forall b' \ \x^T (\featM_{b} - \featM_{b'}) \thetaV \geq 0 \}$}
    \EndFunction
  \end{algorithmic}
\end{algorithm}

The function \textproc{Update} in Algorithm \ref{alg:main} encapsulates this update rule. Fig. \ref{procedure:main} outlines how the various components of the algorithm are used to compute and update policies. Fig. \ref{fig:learning_flow} gives a pictorial representation of the learning scheme and dynamic Stackelberg game.

\begin{remark}
    The procedure \textproc{GetPolicy} requires solving $|\states|\cdot|\followacts|$ copies of the nonconvex quadratic program in \eqref{eq:cons_opt} to compute the triple $(\x_s, \thetaV_s, b_s)$ for each state in $\states$.
\end{remark}

\begin{figure*}[t]
    \centering
    \includegraphics[width=.9\textwidth]{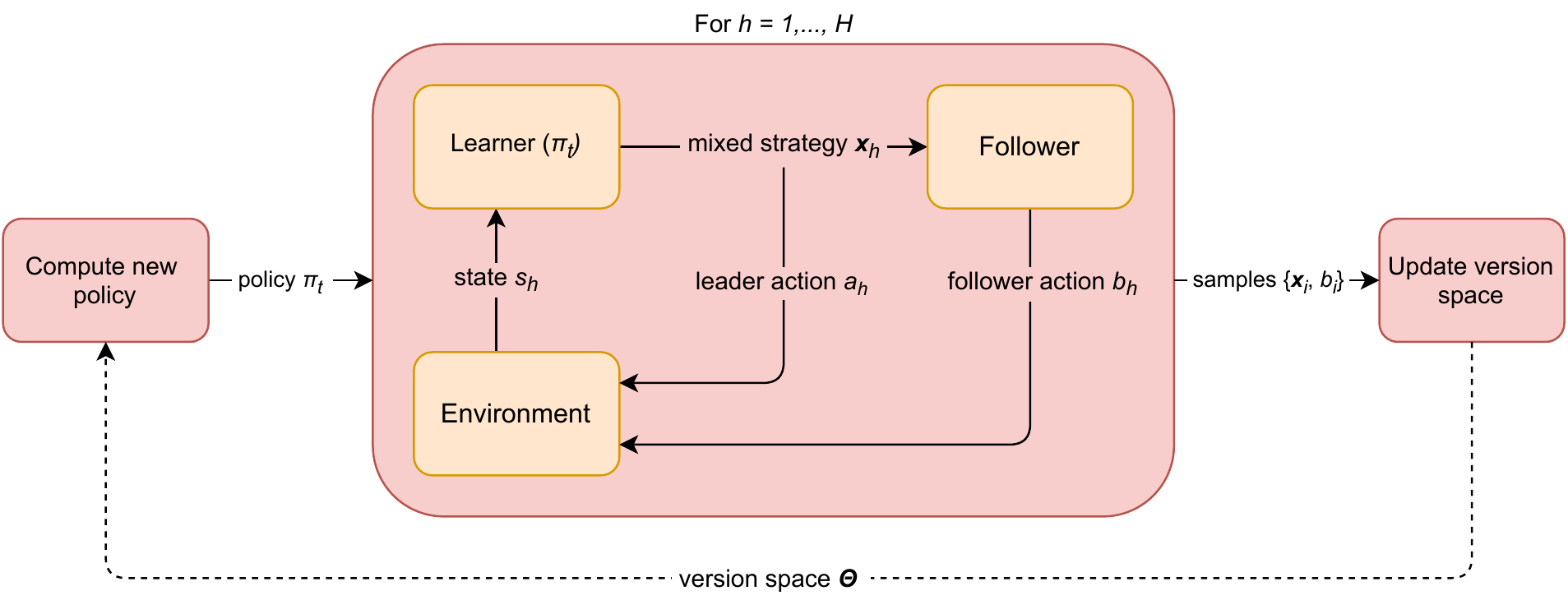}
    \caption{A flow chart for using the learning scheme to compute policies in a discrete-time dynamic Stackelberg game.}
    \label{fig:learning_flow}
\end{figure*}

\section{Regret Analysis} \label{sec:analysis}

% \frommahsa{Adding labels to nontrivial equations and inequalities in the proof steps and accompanying them with a description.}
% \fromnik{Use stackrel command.}
In this section, we establish a high-probability regret bound for the proposed scheme outlined in Fig. \ref{procedure:main}. 

\begin{theorem}\label{thm:main}
Fix confidence parameter $\delta>0$. The regret of the learning scheme in Fig. \ref{procedure:main} satisfies 
\begin{equation}
\begin{aligned}
    R_T \leq &\left(T\right)^{1-\frac{1}{p}} \left(d \sqrt{mn} (1 + \sqrt{nH}) H + H\right) \\ &+  H\sqrt{\frac{T}{2}\ln \left(\frac{1}{\delta} \right)}
    \end{aligned}
\end{equation}
with probability at least $1 - \delta$ for constant $d \in \R$.
\end{theorem}

\begin{remark}
    % \abolfazl{good to have a discussion here about the implications of this theorem. For instance, re-state the fact that the regret is independent of the state-state. Also, maybe the fact that we verify some the role of the main parameters in the experiment section.}
    % The constants $c$ and $d$ depend on the fixed feature map $f$.
    Most notably, the sublinearity in $T$ of the regret depends on the dimension $p$ of the follower's parameter space. However, the regret depends only polynomially in the remaining parameters of the game and has a complete lack of dependence on the size of the state space. In Section \ref{sec:experimental}, we experimentally show the dependence of the regret on $p$ and the independence on the size of the state space. 
\end{remark}
\begin{remark}
We further note that the result of Theorem \ref{thm:main} holds with high probability. Typically, high probability regret bounds are more challenging to derive than expected regret bounds in online learning \cite{abernethy2009beating,neu2015explore,ghasemi2021noregret}. As we will see, we establish the result of Theorem \ref{thm:main} by a careful application of the Azuma's inequality of Martingales \cite{chung2006concentration}. Given this result, an upper bound on the expected regret of the proposed
scheme can be obtained by straightforward integration of the
tail of the high-probability regret bound that we establish
in Theorem \ref{thm:main} (see e.g., \cite{ghasemi2021online}).
\end{remark}

% \fromnik{Remark here about justifying the regret bound because of the need to shrink the version space?}

Next, we give a high-level idea of our proof techniques, which is then followed by the formal proof of Theorem \ref{thm:main}.

%\paragraph{Proof overview.}
\subsection{Proof Outline}\label{sec:outline}

% \fromyagiz{I usually prefer the following structure. Tell me what you will show and why, show it, and tell me the implications of it.}

At the beginning of each episode of the game, we optimistically solve for the optimal $\epsilon$-conservative policy $\pi_t$ by calling \textproc{GetPolicy}. Each mixed strategy $\x_s = \pi_t(s)$ is associated with a choice $\thetaV_s \in \Theta$ and expected response $b_s \in \followacts$. Let $b'_s$ be the \textit{actual} action that the follower would take in response to mixed strategy $\x_s$. If $b_s'=b_s, \forall s \in \states$, then we will show that policy $\pi_t$ is at least as good as the optimal $\epsilon$-conservative mixed since $\thetaV^* \in \Theta$ and thus one of the possible values of $\thetaV$ consistent with previous observations. 

If instead $b_{s} \neq b'_{s}$ for some $s \in \states$, then we say $\pi_t$ has some probability of making a \textit{mistake}. Such a scenario is called a mistake since the computed policy $\pi_t$ with associated triple $(\x_s, \thetaV_s, b_s)$ expects the follower to play $b_s$ in response to mixed strategy $\x_s$. If $b_s$ is not actually played in response to $\x_s$, then the computed value $\estV_t(s)$ is overly optimistic. In Lemma \ref{lemma:actual_value_bound}, we show that playing policy $\pi_t$ induces regret proportional to the probability of making a mistake while executing $\pi_t$.

The resulting observation $(\x_s,b')$ after making a mistake shrinks the version space $\Theta$ by inducing new hyper-plane constraints. In Lemma \ref{lemma:mistake_shrink} we show that, since $\x_t$ is chosen to be $\epsilon$-conservative, whenever it leads to a mistake, $\Theta$ is shrunk by at least some fixed volume. In Lemma \ref{lemma:mistake_bound}, we show that since the initial volume of $\Theta$ is finite, the total number of mistakes is bounded. 

The learning scheme induces two sources of regret: regret from choosing an $\epsilon$-conservative policy for each episode and regret from the probability of making a bounded number of mistakes. In Lemma \ref{lemma:epsilon_diff} we upper bound the regret resulting from choosing an $\epsilon$-conservative policy.
Using the fixed upper bound on the number of mistakes induced by policy $\pi_t$, we derive a probabilistic upper bound on the cumulative \textit{probability} that policy $\pi_t$ leads to a mistake using Azuma's inequality of Martingales. Finally, at the end of Section \ref{sec:analysis}, we use this sequence of lemmas to prove Theorem \ref{thm:main}.

%\paragraph{Formal analysis.}
\subsection{Proof of Theorem \ref{thm:main}}
In this section, we establish the proof of Theorem \ref{thm:main} by providing a number of intermediate lemmas which we discussed in Section \ref{sec:outline}. Proofs of the intermediate lemmas are left to the appendix.

% Now, to choose our mixed strategy $\x$, instead of naively optimizing $\x \in \Delta(A)$ at each step of the game directly using $V$, we solve the conservative optimization problem
% \begin{equation}
%     \begin{array}{ll@{}ll}
%         \displaystyle\max\limits_{\x \in \Delta(A)} & 
%         \E_{a \samp \x} \left[r(s,a,b) + \displaystyle\sum\limits_{s'} P(s,a,b,s') V(s') \right]  &\\
%         \text{s.t.} & \x^T (\featM_b - \featM_{b'}) \thetaV\geq \epsilon & \forall b' \in \followacts \\
%         & \thetaV \in V
%     \end{array}
%     \label{opt:error_bound}
% \end{equation}
% for a fixed $\epsilon > 0$.

Before showing that the learning scheme makes a bounded number of mistakes in Lemmas \ref{lemma:mistake_shrink} and \ref{lemma:mistake_bound}, we show that the follower's reward function can be scaled without affecting its policy. This allows us to assume, without loss of generality, that the difference of feature matrices $\opnorm{\M_b - \M_{b'}} \leq 1$ for any $b, b' \in \followacts$. 

\begin{proposition} \label{prop:scale_f}
    The follower's policy is invariant under feature mappings up to scalar multiplication. That is, if $ \varphi(\x)$ and $ \varphi'(\x)$ denote the policies for feature mappings $f(a,b)$ and $f'(a,b) = c \cdot f(a,b)$, respectively, for $c\in \R$, then $\varphi(\x) = \varphi'(\x)$.
\end{proposition}

 Given Proposition \ref{prop:scale_f}, the feature map $f$ can be scaled arbitrarily, without loss of generality.

We proceed with the first two lemmas, by first showing that samples gathered when the learning schemes makes mistakes, are guaranteed to shrink the version space by a fixed amount. We use this to show that policies resulting from the optimization problem in Equation \eqref{eq:cons_opt} throughout the execution of the learning scheme, only make a bounded number of mistakes.

\begin{lemma} \label{lemma:mistake_shrink}
Let $(\x_s, \thetaV_s, b_s)$ for any $s \in \states$ be the solution from solving Equation \eqref{eq:cons_opt}.
If the solution $(\x_s, \thetaV_s, b_s)$ makes a mistake during execution in the sense that action $b_s$ is predicted by $\thetaV_s$ but $b^*$ is actually played in response to $\x$, then \textproc{Update($\x,b^*,\Theta$)} shrinks $\Theta$ by at least 
% $\Theta \cap B_{\epsilon}(\thetaV_s)$ 
$\Theta \cap B_{\epsilon}(\thetaV_s)$ 
% where $c$ is the operator norm of $\featM_{b_s} - \featM_{b^*}$.
\end{lemma}

Since the volume of the initial version space $\Theta$ is finite, Lemma \ref{lemma:mistake_shrink} bounds the total number of possible mistakes.

\begin{lemma} \label{lemma:mistake_bound}
Solutions $(\x_s, \thetaV_s, b_s)$ found by the learning scheme produce mistakes during execution at most $(\frac{2}{\epsilon})^{p-1}$ times.
% where $c$ is the maximum operator norm of all matrices $\featM_b - \featM_{b'}$ for $b,b' \in \followacts$.
\end{lemma}
 
Now that we have given an upper bound on the number of mistakes that mixed strategies resulting from the learning scheme can make, we proceed by bounding the regret induced by solving for an $\epsilon$-conservative policy.

Define the \textit{optimal $\epsilon$-conservative policy} $\pi^{\epsilon} : \states \to \Delta(\leadacts)$ in each state $s \in \states$ as the solution to the following optimization problem.
\begin{subequations}
\label{eq:optimal_cons_opt}
\begin{align}
    &\max\limits_{\x_s} \quad  V(\pi^{\epsilon},s)        \\
    &\textrm{s.t.} \quad
     V(\pi^{\epsilon},s) = \displaystyle\E\limits_{a \samp \x_s} \left[r(s,a,b) + \displaystyle\sum\limits_{s'} P(s,a,b,s') V(\pi^{\epsilon},s') \right] &\\
    & \qquad\x_s^T (\featM_{b} - \featM_{b'}) \thetaV^* > \epsilon, \ \forall b' \in B \\
    & \qquad\x_s \in \Delta(\leadacts(s))
        \label{eq:cons_opt:epsilon_bound2}
\end{align}
\end{subequations}
maximized over all $b \in \followacts$. Then $\pi^{\epsilon}(s) := \x_s$. Notice that the solutions to Equation \eqref{eq:optimal_cons_opt} can be computed using backwards induction in the same way that Equation \eqref{eq:cons_opt} is.

\begin{definition}
  Let $\actualV(\pi, s)$ represent the \textit{actual} expected value obtained during an episode from playing policy $\pi$ beginning in state $s$. 
\end{definition}

We say that a trace $\tau$ resulting from a policy $\pi_t$ \textit{makes a mistake} if $\tau$ results in playing a policy $\x_s$ at some state $s \in \states$ with the expectation that action $b_s$ will be played (i.e. $\x_s^T (\featM_{b_s} - \featM_{b'}) \thetaV_s \geq 0$ for all $b' \in \followacts$) but in reality, a different action is played (i.e. $\x_s^T (\featM_{b_s} - \featM_{b'}) \thetaV^* \leq 0$ for some $b' \in \followacts$).
First off, notice that if no traces obtained by following policy $\pi_t$ from state $s$ makes a mistake, then $\actualV(\pi_t, s) = \estV_t(s)$. Otherwise, we have the following lemma.

\begin{lemma} \label{lemma:actual_value_bound}
    % Let $\actualV(\pi, s)$ represent the actual value of policy $\pi$ from state $s$. Let $\estV_t(s)$ represent the estimated value of policy $\pi_t$ from state $s$. Let $\pi^{\star}$ represent the optimal policy and $\pi^{\epsilon}$ represent the optimal $\epsilon$-conservative policy. 
    Let $\lambda_t(s)$ be the probability that $\pi_t$ makes a mistake during the episode beginning from state $s$. Then,
    \begin{equation}
        \actualV(\pi^{\star}, s) \geq \actualV(\pi_t, s) \geq \actualV(\pi^{\epsilon}, s) - H \cdot \lambda_t(s).
    \end{equation}
\end{lemma}

The difference in optimality between $\actualV(\pi^{\star}, s)$ and $\actualV(\pi^{\epsilon}, s)$ differs only by the difference in feasibility region induced by relaxing (\ref{eq:cons_opt:epsilon_bound2}) to $\x_s^T (\featM_{b_s} - \featM_{b'}) \thetaV\geq 0$. Therefore, the suboptimality of $V_{\epsilon}^*$ is bounded by the maximum difference in value induced by two solutions that differ by the Hausdorff distance between the feasibility regions. The Hausdorff distance is realized at the critical points of the feasibility regions. Let $C \subset \followacts$, $|C| = q$ represent a critical boundary of the feasibility space induced by the relaxation of (\ref{eq:cons_opt:epsilon_bound2}). That is, there exists some ${\x^*}$ such that ${\x^*}^T (\featM_b - \featM_{b'}) \thetaV^* = 0$ for all $b' \in C$. 

The Hausdorff distance is upper bounded by the maximum distance between $\x^*$ subject to ${\x^*}^T (\featM_b - \featM_{b'}) \thetaV^* = 0$ and 
$\x$ subject to ${\x}^T (\featM_b - \featM_{b'}) \thetaV^* = \epsilon$ for all $b' \in C$ over all $C \subset \followacts$.
So we have $({\x}^T - {\x^*}^T)(\featM_b - \featM_{b'}) \thetaV^* = 0$ for all $b' \in C$.

Define a projection matrix $\projM \in \R^{n \times q}$ such that $[\projM]^{b'} = (\featM_b - \featM_{b'}) \thetaV^*$. Then we have
\begin{equation}
    (\x^T - {\x^*}^T) \projM = \bar{\epsilon}
\end{equation}
where $\bar{\epsilon} = (\epsilon, \dots, \epsilon)$. 
If $q \geq n$, then $\projM$ has a right pseudoinverse. Rearranging, this gives 
\begin{subequations}
\begin{align}
    ||(\x^T - {\x^*}^T)|| &= ||\bar{\epsilon} \projM^{\dagger} ||\\
    &\leq \sqrt{q} \cdot \epsilon \cdot d \\
    &\leq \sqrt{m} \cdot \epsilon \cdot d
\end{align}
\end{subequations}
 since $m \leq q$ where $d = \frac{1}{\sigma_{\min}(\projM)}$. 

If instead $q < n$, then we make use of the following lemma.
\begin{lemma}\label{lem:aux}
Let $\projM \in \R^{n \times q}$ with $q < n$ with rank $q$. If $\x^T \projM = \bar{\epsilon}$, then there exists a full rank matrix $\projM' \in \R^{n \times n}$ with the first $q$ columns identical to $\projM$ such that $\x^T \projM' = (\bar{\epsilon}, \bar{0})$. Moreover, $\projM'$ can be constructed to have minimum singular value equal to the minimum singular value of $\projM$. 
\end{lemma}

% Since $\projM'$ has first $q$ columns identical to $\projM$, we know that $||\x^{T}\projM|| < ||\x^{T} \projM'||$ for all $\x \in \R^n$. Let $\followacts \in \R^{n \times (n-q)}$ be the remaining columns of $\projM'$. Then we have
% \begin{align}
%     \opnorm{(\projM')^{-1}} = \frac{1}{\inf_b\{\frac{||\projM'b||}{||b||}\}}
%     \leq \frac{1}{\inf_b\{\max(\frac{||\projM b||}{||b||},\frac{||Ab||}{||b||})\}} 
%     = \frac{1}{\min(\text{the nonzero singular values between $M$ and $A$})}
% \end{align}
% because $M'$ has singular values lower bounded by the singular values from $M$ and $A$.

\noindent
Therefore, again we have that
\begin{subequations} \label{eq:pre_lem5}
\begin{align}
    ||\x^T - {\x^*}^T|| &= ||(\bar{\epsilon}, \bar{0}) (\M')^{-1} ||\\
    &\leq \sqrt{q} \cdot \epsilon \cdot d \\
    &\leq \sqrt{m} \cdot \epsilon \cdot d
\end{align}
\end{subequations}
where $d = \frac{1}{\sigma_{\min}(\M)}$ since $\opnorm{\A^{-1}} = \frac{1}{\sigma_{\min}(\A)}$ for any matrix $\A$.
This result gives rise to the following lemma.

% The difference in value between the two optimization problems is upper bounded by the difference in value of the Hausdorff distance between the two feasibility regions. In fact, since the optimization problem is linear, the solution $\x^*$ to (\ref{opt:zero_bound}) and $\x$ to (\ref{opt:error_bound}) satisfies
% \begin{equation}
%     ||\x^T(\featM_b - \featM_{b'}) - {\x^*}^T(\featM_b - \featM_{b'})|| \leq \epsilon
% \end{equation}
% for all $b,b' \in \followacts$. \fromnik{I'm not 100\% sure about this claim. Needs to be more convincing.} Therefore, we have that $||\x^T - {\x^*}^T|| \leq \epsilon$ which in turn bounds
% \begin{align}
%     \E\limits_{a \samp \x} \left[r(s,a,b_j)] \right] - \E\limits_{a \samp \x^*} \left[r(s,a,b_j)] \right]& \\
%      = (\x^T  - {\x^*}^T) \r_{s,b} &\leq ||\x^T - {\x^*}^T|| \cdot ||\r_{s,b}|| \\
%      &\leq \epsilon \cdot ||\r_{s,b}|| \\
%      &\leq \epsilon \cdot \sqrt{n}
% \end{align}
% since $\r_{s,b} \in [0,1]^n$. \fromnik{ Note that this analysis ignores the the future states and therefore the future discounted reward. I think a similar analysis should bound the difference in reward while taking future states into account.}

\begin{lemma} \label{lemma:epsilon_diff}
For all $s \in \states$,
    \[V(\pi^{\star}, s) - V(\pi^{\epsilon}, s) \leq \epsilon d \sqrt{mn} (1 + \sqrt{nH}) H,\]
where $d$ is defined as above.
\end{lemma}

% The \textit{global regret} against policy $\pi^{\star}$ that chooses mixed policies $\x_s^*$ is
% \begin{align}
%     R_{T} &= \sum_{t=1}^T V^*(s_1)
%     - V_t(s_1)\\
%     &= 
% \end{align}
% Whenever we make a mistake within an episode, we can incur arbitrary (up to $H$) regret for that episode. However, whenever we don't make a mistake, we incur regret from the above Lemma. Therefore,
% \begin{align}
%     R_T &\leq H\left(\frac{1}{\epsilon}\right)^{p-1} +  \left( \epsilon \cdot d \cdot \sqrt{mnH} \right) \left(T-\left(\frac{1}{\epsilon}\right)^{p-1}\right) \label{eq:global_regret_complex_bound} \\
%     &\leq H\left(\frac{1}{\epsilon}\right)^{p-1} + \left( \epsilon \cdot d \cdot \sqrt{mnH} \right) T \label{eq:global_regret_simple_bound}.
% \end{align}
% \fromnik{Note that we lose some optimality by simplifying to (\ref{eq:global_regret_simple_bound}). How much better can we do if we optimize (\ref{eq:global_regret_complex_bound}) instead?}
% Choosing $\epsilon = T^{-\frac{1}{p}}$ gives
% \begin{align}
%     R_T &\leq H\left(T\right)^{\frac{p-1}{p}} +  d \cdot \sqrt{mnH} \cdot \left(T\right)^{\frac{p-1}{p}} \\
%      &\leq \left(T\right)^{\frac{p-1}{p}} \left(H + d \cdot \sqrt{mnH}\right)
% \end{align}
% So $R_T = O(T^{\frac{p-1}{p}})$ which is sublinear for any $p$.

We are now ready to prove Theorem \ref{thm:main}.
% \begin{theorem}\label{thm:main}
% Fix the confidence parameter $\delta>0$. The regret of the proposed scheme satisfies 
% \begin{equation}
%     R_T \leq \left(T\right)^{1-\frac{1}{p}} \left(d \sqrt{mn} (1 + \sqrt{nH}) H + c^{p-1} H\right) +  H\sqrt{\frac{T}{2}\ln \left(\frac{1}{\delta} \right)}
% \end{equation}
% with probability at least $1 - \delta$.
% \end{theorem}
\begin{proof}[Proof of Theorem \ref{thm:main}]

Recall that the regret $R_T$ is defined as
\begin{equation}
    R_T = \sum_{t=1}^T \Big( \actualV(\pi^{\star},s_1) - \actualV(\pi_t,s_1) \Big)
\end{equation}

By Lemma \ref{lemma:actual_value_bound} we have,
\begin{subequations}
\begin{align}
    R_T
    &\leq \sum_{t=1}^T \Big( \actualV(\pi^{\star},s_1) -  \left( \actualV(\pi^{\epsilon},s_1) - H \cdot \lambda_t(s_1) \right) \Big) \\
    &\leq \sum_{t=1}^T \Big( \actualV(\pi^{\star},s_1) - \actualV(\pi^{\epsilon},s_1) \Big) + \sum_{t=1}^T  H \cdot \lambda_t(s_1) \\
    &=\sum_{t=1}^T \Big( \actualV(\pi^{\star},s_1) - \actualV(\pi^{\epsilon},s) \Big) + H \sum_{t=1}^T \lambda_t(s_1) \label{eq:sources_of_regret}
\end{align}
\end{subequations}
Notice that Equation \eqref{eq:sources_of_regret} clearly shows the two sources of regret: from choosing an $\epsilon$-conservative policy and from the probability of making mistakes.
By Lemma \ref{lemma:epsilon_diff} we continue to get,
\begin{equation}
         R_T \leq T \Big( \epsilon d \sqrt{mn} (1 + \sqrt{nH}) H \Big) + H \sum_{t=1}^T \lambda_t(s_1).
\end{equation}
The final step is to use the fixed upper bound on the number of mistakes to give a \textit{probabilistic} upper bound on the cumulative probability of making a mistake.

% Let $X_t$ be the indicator random variable for the event that $\pi_t$ makes a mistake during execution. By construction, $\E[X_t] = \lambda_t(s_1)$. From Lemma \ref{lemma:mistake_bound}, the actual number of mistakes $X = \sum_{t=1}^T X_t$ is upper bounded by $\left(\frac{2}{\epsilon}\right)^{p-1}$. \fromnik{Add justification for why the $X_t$ are independent. Or maybe they're not and we need to use a different inequality.} Using the Hoeffding inequality, we obtain the following bound on the expected number of mistakes $\lambda = \sum_{t=1}^T \lambda_t$.
% \fromnik{Get rid of the constant $2$ here by arguing that we can just choose $\epsilon' = 2\epsilon$ instead.}

Let $X_t$ be the indicator random variable for the event that $\pi_t$ makes a mistake during execution on step $t$ of the game. Let $\lambda_t = \lambda_t(s_1)$. By construction, $\E[X_t] = \lambda_t$. Notice that $\lambda_t$ is itself a random variable that depends on the outcomes of $X_1, X_2, \dots, X_{t-1}$.

Define the martingale $M_t = \sum_{i=1}^t (X_i - \lambda_i)$ with filtration $\mathcal{F}_t = \sigma(\{X_j, \lambda_j\}_{j=1}^t)$. The sequence $\{M_t\}$ is indeed a martingale since,
\begin{equation}
    \E[M_{t+1} \mid \mathcal{F}_t] = \E[X_{t+1} - \lambda_{t+1} \mid \mathcal{F}_t] + M_t = M_t
\end{equation}
where the final equality comes from the fact that, by construction, the expectation of $X_{t+1}$ is $\lambda_{t+1}$ after $X_1, \dots, X_t$ have been observed.
Since $|M_{t+1} - M_t| = |X_{t+1} - \lambda_{t+1}| \leq 1$, by Azuma's inequality \cite{chung2006concentration}, we have
\begin{equation} \label{eq:azuma}
    \pr\left[|M_t| \geq \alpha\right] \leq \exp
    \left(-\frac{\alpha^2}{2T}\right).
\end{equation}

From Lemma \ref{lemma:mistake_bound}, the actual number of mistakes $\sum_{t=1}^T X_t$ is upper bounded by $\left(\frac{2}{\epsilon}\right)^{p-1}$. Therefore,
\begin{equation}
    M_T = \sum_{i=1}^T X_i - \sum_{i=1}^T \lambda_i \leq \left(\frac{2}{\epsilon}\right)^{p-1} - \sum_{i=1}^T \lambda_i.
\end{equation}
Combining this with Equation \eqref{eq:azuma}, we obtain
\begin{equation}
    \pr \left[\sum_{i=1}^T \lambda_i \geq \left(\frac{2}{\epsilon}\right)^{p-1} + \alpha \right] \leq \exp
    \left(-\frac{\alpha^2}{2T}\right).
\end{equation}

% Then, by the Chernoff bound,
% \begin{align}
%     \pr \left[\sum_{t=1}^T \lambda_i  \geq \left(\frac{2}{\epsilon}\right)^{p-1} + a \right]
%     &\leq \pr \left[\sum_{t=1}^T \lambda_i \geq \sum_{t=1}^T X_t + a \right] \\
%     &= \pr \left[ \sum_{t=1}^T \lambda_i - \sum_{t=1}^T X_t \geq a \right] \\
%     &\leq \pr \left[ \left| \sum_{t=1}^T \lambda_i - \sum_{t=1}^T X_t \right| \geq a \right] \\
%     &\leq \frac{\sum_{t=1}^T \lambda_t(s_1)(1-\lambda_t(s_1))}{a^2} \\
%     &\leq \frac{T}{4 a^2}
% \end{align}
% \begin{subequations}
% \begin{align}
%      \pr \left[\lambda \geq \left( \frac{2}{\epsilon}\right)^{p-1} + \alpha \right] 
%      &= \pr \left[\lambda - \left( \frac{2}{\epsilon}\right)^{p-1} \geq \alpha \right]\\
%     &\leq \pr \left[\lambda - X \geq \alpha \right] \\
%     &\leq e^{-\frac{2\alpha^2}{T}}.
% \end{align}
% \end{subequations}
Returning to the regret, this gives
\begin{equation}
         R_T \leq T \Big( \epsilon d \sqrt{mn} (1 + \sqrt{nH}) H \Big) + H \left(\left(\frac{2}{\epsilon}\right)^{p-1} + \alpha \right)
\end{equation}
with probability at least $1 - e^{-\frac{\alpha^2}{2T}}$.
Choosing $\epsilon = 2T^{-\frac{1}{p}}$ gives
\begin{subequations}
\begin{align}
    R_T &\leq 2d \sqrt{mn} (1 + \sqrt{nH}) H \left(T\right)^{1-\frac{1}{p}} + H 
    % c^{p-1}
    \left(T\right)^{1-\frac{1}{p}} + \alpha H \\
     &\leq 2\left(T\right)^{1-\frac{1}{p}} \left(d \sqrt{mn} (1 + \sqrt{nH}) H + 
    %  c^{p-1}
     H\right) + \alpha H
\end{align}
\end{subequations}
with probability at least $1 - e^{-\frac{\alpha^2}{2T}}$. Taking $\alpha = \beta \sqrt{T}$ gives
\begin{equation}
    R_T \leq \left(T\right)^{1-\frac{1}{p}} \left(2 d \sqrt{mn} (1 + \sqrt{nH}) H + H\right) + \beta \sqrt{T} H
\end{equation}
with probability at least $1 - e^{-\frac{\beta^2}{2}}$ for any choice of $\beta$.
Letting $\delta = e^{-\frac{\beta^2}{2}}$ we can get the alternate representation
\begin{equation}
\begin{aligned}
    R_T &\leq \left(T\right)^{1-\frac{1}{p}} \left(2 d \sqrt{mn} (1 + \sqrt{nH}) H + H\right) \\&\qquad+  H\sqrt{2T\ln \left(\frac{1}{\delta}\right)}
    \end{aligned}
\end{equation}
with probability at least $1 - \delta$.
\end{proof}
\subsection{Anytime sublinear regret for DSGs}
% \fromnik{should we keep the figure (Algorithm 2) for the anytime version?} \abolfazl{let's keep it if we have space}
While we showed in the previous section that the learning scheme outlined in Fig. \ref{procedure:main} enjoys a high probability sublinear regret, the algorithm relies on knowing the number of episodes $T$. In what follows, using the \textit{doubling trick} \cite{auer1995gambling}, we show how to adapt the learning scheme into \textit{an any-time algorithm} that does not require knowing $T$ in advance. That is, even if the horizon $T$ is unknown, the adapted version of our learning scheme will achieve the same regret bound as the standard learning scheme, with probability $1-\delta$, for any confidence parameter $\delta>0$.

The adapted learning scheme, outlined in Algorithm \ref{alg:anytime}, operates by learning over increasingly large time segments $\{T_i\}$, starting from an initial segment $T_0$. To make our analysis simpler, we first consider the scheme in which progress made in a segment $T_j$ is discarded once segment $T_{j+1}$ begins.

Let $T_i = 2^i$ and $\Ttrue = \sum_{i=0}^n T_i$. $CR(t)$ is the cumulative regret of Algorithm \ref{alg:anytime} at time $t$ and $R(t)$ is the regret of the learning scheme outline in Fig. \ref{procedure:main} [cf. Theorem \ref{thm:main}]. Then, it holds with probability at least $(1-\delta)^n$,
\begin{subequations}
\begin{align}
    CR&(\Ttrue) = \sum_{i=0}^n R_{T_i} \\
    \leq& \sum_{i=0}^n \left(T_i \right)^{1-\frac{1}{p}} \left(d \sqrt{mn} (1 + \sqrt{nH}) H + c^{p-1} H\right)\\&\qquad +  H\sqrt{\frac{T_i}{2}\ln \left(\frac{1}{\delta}\right)}\\
    =& \left(d \sqrt{mn} (1 + \sqrt{nH}) H + c^{p-1} H\right) \sum_{i=0}^n \left(T_i \right)^{1-\frac{1}{p}} \\&\qquad +   H\sqrt{\frac{1}{2}\ln \left(\frac{1}{\delta}\right)} \sum_{i=0}^n \sqrt{T_i} \\
    % =& \left(d \sqrt{mn} (1 + \sqrt{nH}) H + c^{p-1} H\right) \sum_{i=0}^n \left(2^i \right)^{1-\frac{1}{p}} \\&\qquad +   H\sqrt{\frac{1}{2}\ln \left(\frac{1}{\delta}\right)} \sum_{i=0}^n \sqrt{2^i} \\
    % =& \left(d \sqrt{mn} (1 + \sqrt{nH}) H + c^{p-1} H\right) \sum_{i=0}^n \left(2^{1-\frac{1}{p}}\right)^i \\&\qquad +   H\sqrt{\frac{1}{2}\ln \left(\frac{1}{\delta}\right)} \sum_{i=0}^n \sqrt{2}^i \\
    % =& \left(d \sqrt{mn} (1 + \sqrt{nH}) H + c^{p-1} H\right) \left(\frac{\left(2^{1-\frac{1}{p}}\right)^{n+1}-1}{\left(2^{1-\frac{1}{p}}\right)-1} \right)\\&\qquad +   H\sqrt{\frac{1}{2}\ln \left(\frac{1}{\delta}\right)} \left( \frac{\sqrt{2}^{n+1}-1}{\sqrt{2}-1} \right) \\
    % =& \mathcal{O}(\Ttrue^{1-\frac{1}{p}})
\end{align}
\end{subequations}
Then, using a simple geometric sum formula, one can establish that $CR(\Ttrue) = \Tilde{\mathcal{O}}(\Ttrue^{1-\frac{1}{p}})$, with probability at least $(1-\delta)^n\geq 1-n \delta$. Therefore, by re-scaling $\delta$ to $\delta/n$ we establish the intended result.
% \abolfazl{OK, while this calculation is correct, it seems the analysis of the probability needs thinking. I think we have the following: we have $n$ terms, each holding with prob $1-\delta$. So to ensure the whole thing holds, we need to scale $\delta$ by $\delta/n$. Someone please check.}

\begin{algorithm}[t]
  \caption{An any-time version of the learning scheme.
    \label{alg:anytime}}
  \begin{algorithmic}[1]
    \Require{Dynamic Stackelberg game $(\states, \leadacts, \followacts, r, u, P)$.}
    \State Initialize $T_0$.
    \For{$i \in \{0,\dots,n\}$}
        \State $T_i \gets 2^i T_0$
        \State $\textproc{Learn}(T_i)$
    \EndFor
  \end{algorithmic}
\end{algorithm}

Now we consider the case in which progress in the learning scheme is carried over between time segments. In this case, instead of discarding the halfspaces accumulated during previous segments, we continue to use those halfspaces to restrict the possible value of $\thetaV^{*}$. The cumulative regret incurred during segment $T_i$ is then upper bounded by $R_{T_i}$ since the volume of the version space $\Theta$ when starting segment $T_i$ is at most the volume of the initial version space $\Theta_0$. In fact, it is likely smaller, since mistakes made in previous segments would have incurred new halfspaces that shrink the size of $\Theta$ even before segment $T_i$ begins.
% \begin{align}
%     \pr \left[(1-\delta)^{-1}\left(\frac{2}{\epsilon}\right)^{p-1} \leq \sum_{t=1}^T \lambda_i\right] &= \pr \left[\left(\frac{2}{\epsilon}\right)^{p-1} \leq (1-\delta) \sum_{t=1}^T \lambda_i\right]\\
%     &\leq \pr \left[X \leq (1-\delta) \sum_{t=1}^T \lambda_i\right] \\
%     &\leq e^{-\sum_{t=1}^T\lambda_i \cdot \delta^2/2} \\
%     &\leq e^{-T\delta^2/8}
% \end{align}
% for $\delta \in [0,1]$.
% Returning to the regret, this gives
% \begin{equation}
%          R_T \leq T \Big( \epsilon \cdot d \cdot \sqrt{mnH} \Big) + H (1-\delta)^{-1} \left(\frac{2}{\epsilon}\right)^{p-1}
% \end{equation}
% with probability at least $1- e^{-T\delta^2/8}$.
% Choosing $\epsilon = T^{-\frac{1}{p}}$ gives
% \begin{align}
%     R_T &\leq d \cdot \sqrt{mnH} \cdot \left(T\right)^{1-\frac{1}{p}} + (1-\delta)^{-1} H \cdot c^{p-1} \cdot \left(T\right)^{1-\frac{1}{p}} \\
%      &\leq \left(T\right)^{\frac{p-1}{p}} \left(d \cdot \sqrt{mnH} + (1-\delta)^{-1} \cdot c^{p-1} \cdot H\right)
% \end{align}
% with probability at least $1- e^{-T\delta^2/8}$. Defining $\alpha = (1-\delta)^{-1} \in [1,\infty]$ simplifies this to
% \begin{equation}
%     R_T \leq \left(T\right)^{\frac{p-1}{p}} \left(d \cdot \sqrt{mnH} + \alpha \cdot c^{p-1} \cdot H\right)
% \end{equation}
% with probability at least $1- e^{-T(1-\frac{1}{\alpha})^2/8}$ for any $\alpha \in [1,\infty]$.

\section{Experimental Results} \label{sec:experimental}

In this section, we present experimental results on the regret that our algorithm incurs over discrete-time dynamic Stackelberg games (DSGs) with varying parameters. First we discuss results related to the poaching example from the introduction, and then we give results averaged over randomly generated DSG instances to show performance as parameters of the DSG vary. Our implementation uses Gurobi \cite{gurobi} to compute solutions to the quadratic program in Equation \eqref{eq:cons_opt}.

To demonstrate the performance of our algorithm, we report the \textit{average regret} over time. At time step $t$, the average regret is calculated as 
\begin{equation}
    \frac{1}{tH}\sum_{i = 1}^t \sum_{h = 1}^H R^{\pi^{\star}}_{i,h} - R^{\pi_i}_{i,h}
\end{equation}
where $R^{\pi^{\star}}_{i,h}$ and $R^{\pi_i}_{i,h}$ is the reward that the optimal policy and learning policy, respectively, receives at step $h$ of episode $i$.
To compare the performance of multiple different policies, we also report the \textit{average cumulative reward} over time. At time step $t$, the average cumulative reward is calculated as 
\begin{equation}
    \frac{1}{tH}\sum_{i = 1}^t \sum_{h = 1}^H R^{\pi_i}_{i,h}
\end{equation}
where $R^{\pi_i}_{i,h}$ is the reward that policy $\pi_i$ receives at step $h$ of episode $i$.

To the best of the authors' knowledge, no other algorithms exist for directly solving Problem \ref{prb:learning}. However, by disregarding the reward and transition structure induced by the Stackelberg game, a discrete-time dynamic Stackelberg game can be reduced to a Markov decision process (MDP) with an unknown reward and transition functions in the following way.

Let $(\states, \leadacts, \followacts, r, u, P)$ be a discrete-time dynamic Stackelberg game.
Consider the auxilliary reward function $R : \Delta(\leadacts) \times \states \to \R$ defined in Equation \eqref{eq:aux_reward} induced by the follower's policy.
An auxiliary transition function $P' : \states \times \Delta(\leadacts) \times \states$ can be defined in a similar way as
\begin{equation} \label{eq:aux_trans}
    P'(s, \x, s') = \E_{a \samp \x} \left[ p(s, a, \varphi(\x), s) \right].
\end{equation}
Therefore, the game can be reduced to the MDP with parameters $(\states, \Delta(\leadacts), R, P')$. Notice that the action space is the continuous space of mixed policies and that if the follower's utility function is unknown, then $R$ and $P'$ are also unknown.

By disregarding the reward and transition structure given by $R$ and $P'$, respectively, learning in this setting can be done with reinforcement learning. Q-learning is a popular model-free reinforcement learning algorithm suitable for this scenario. Since the action space of the resulting MDP is continuous, we discretize it and use a tabular implementation of Q-learning. 

\subsection{Poaching Example}

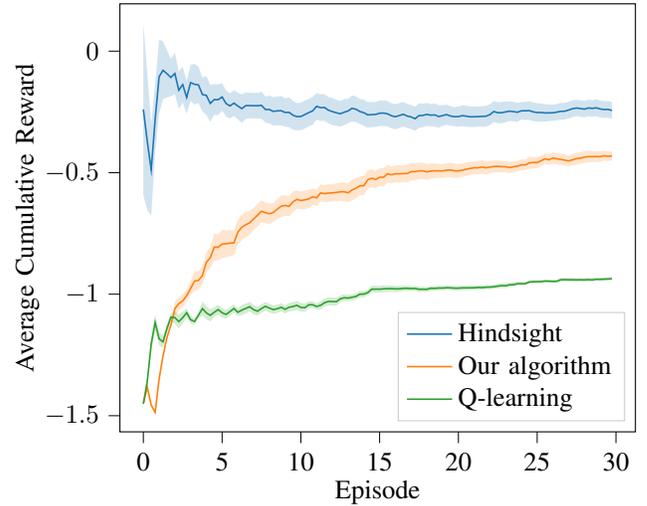
\begin{figure}[t]
    \centering
    \input{poaching_example_T=30}
    \caption{The average regret over time of the learning scheme, a random policy, and the best policy in hindsight on an instance of the poaching example. The best policy in hindsight is the optimal policy had the follower's reward function been known ahead of time.}
    \label{fig:poaching_comparison}
\end{figure}

Recall the motivating example from Section \ref{sec:introduction}: park rangers are responsible for protecting a geographical area containing different types of animals from illegal poachers. The geographical area is split into distinct regions, each of which contains a different density of the animals.  Each month, the park rangers can deploy a (mixed) policy to decide which of the regions to patrol. After observing the rangers' policy, the poachers can attempt to lay snares in any one of the regions.

The park rangers may not know exactly which types of animals the poachers are attempting to poach. Therefore, we model the poachers reward function as an unknown linear combination of the probability of poaching each type of animal with a penalty if the poacher is caught.
Let $D: [N] \to [0,1]^M$ represent a density function such that $[D(n)]_m$ represents the density of animal $m$ in subregion $n$. The poacher's utility function can therefore be described as
\begin{equation}
    u(a,b) = \langle \feat(a,b) , \thetaV^* \rangle
\end{equation}
with unknown weight vector $\thetaV^*$ and known feature function
\begin{equation}
    f(a,b) = \begin{cases}
    (0,\dots,0,-1) & \text{if $a = b$;} \\
    (D(b), 0)  & \text{otherwise.}
    \end{cases}
\end{equation}
The last index in $f(a,b)$ describes whether or not the poacher is caught by the part rangers. Therefore, the last index in $\thetaV^*$ describes the severity of the ranger getting caught and the first $M$ indices describe the payoff for successfully poaching different types of animals.

The park ranger's reward function correlated directly with the fluctuating severity of the various animals being poached throughout the year. Let $C: [12] \to [0,1]^M$ represent the severity of poaching each of the $M$ types of animals per month.
Let $s = (s_1, s_2)$ represent the park ranger's state where $s_1$ is the current month and $s_2$ is the remaining budget. Then, the park rangers' reward function is
\begin{equation}
    r(s,a,b) = \begin{cases}
    c & \text{if $a = b$;} \\
    -C(s_1) \cdot D(b) & \text{otherwise.}
    \end{cases}
\end{equation}
where $c$ is a constant describing the value of catching and arresting a poacher.
The set of actions $\leadacts(s)$ available to the park rangers from state $s$ are the regions that keep the park rangers under budget.

We explore the performance of our learning scheme on an instance of the poaching example with a budget of six, a horizon of four, and three different types of animals (giving $p=4$). The park is split into four different regions, two of which cost one unit of the budget to patrol, and two that cost two units of budget to patrol. The poaching density for each type of animal is randomly chosen in each region, and the poaching severity for each type of animal is randomly chosen for each region.
The poacher's preferences $\thetaV^*$ is randomly chosen and unknown to the leader. 

Fig. \ref{fig:poaching_comparison} shows a comparison of the performance of our learning scheme against the best policy in hindsight and a discretized implementation of Q-learning over a time horizon of 30. Since policies in the poaching example are stochastic, the results in Fig. \ref{fig:poaching_comparison} are averaged over 10 samples.

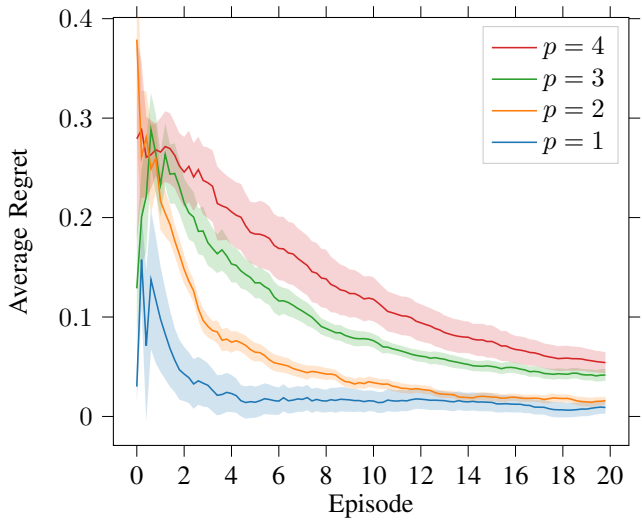
\begin{figure}[t]
    \centering
    \input{avg_regret_ps_error}
    \caption{The average regret for different numbers $p$ of features representing the follower's utility function.}
    \label{fig:varying_p}
\end{figure}

\subsection{Randomly Generated Instances}

Each of the randomly generated DSG instances has four actions available for both the leader and follower ($n = 4$ and $m = 4$). In order to report the average performance in different scenarios, numerical values in the game are generated randomly and results are averaged over ten different games. In particular, rewards and transitions between layers in in the state space are generated uniformly at random. The state spaces resemble a tree, with an increasing number of states in each layer.

Fig. \ref{fig:varying_p} shows the average regret that our algorithm incurs over a horizon of $T=20$ with varying dimensions $p$ of the follower parameter space. As outlined above, the numerical values are generated randomly and the state space has $(1,2,4,8,16)$ states in each layer. In agreement with the theoretical analysis, the asymptotic behavior of the learning agent approaches the optimal policy more slowly with a larger value of $p$.

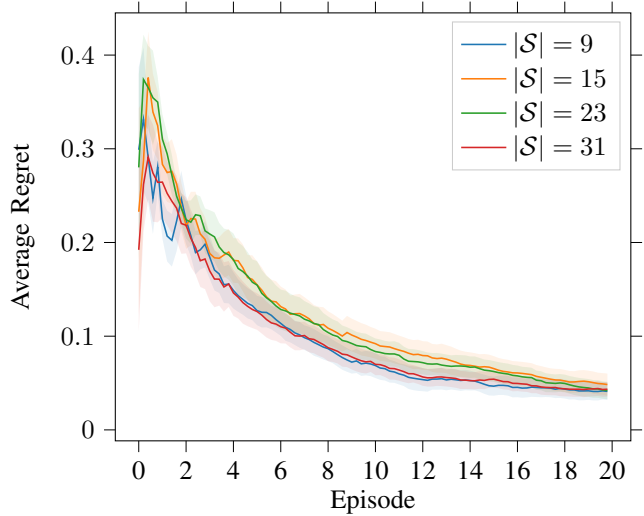
\begin{figure}[t]
    \centering
    \input{avg_regret_s_size_error}
    \caption{The average regret for different size state spaces $\states$. Table \ref{fig:chart} shows the number of states in each layer of the different state spaces.}
    \label{fig:varying_s}
\end{figure}

\definecolor{color0}{rgb}{0.12156862745098,0.466666666666667,0.705882352941177}
\definecolor{color1}{rgb}{1,0.498039215686275,0.0549019607843137}
\definecolor{color2}{rgb}{0.172549019607843,0.627450980392157,0.172549019607843}
\definecolor{color3}{rgb}{0.83921568627451,0.152941176470588,0.156862745098039}
\definecolor{color4}{rgb}{0.580392156862745,0.403921568627451,0.741176470588235}

\begin{table}[t]\label{fig:chart}
\centering
\begin{tabular}{|c|c|c|c|c|c|}
% \multicolumn{5}{|c|}{Layers} \\
\hline
& $\states_1$ & $\states_2$ & $\states_3$ & $\states_4$ & $\states_5$ \\
\hline
% \tikz\draw[color0,fill=color0] (0,0) circle (.5ex); & 1 & 1 & 1 & 1 & 1 \\
%  \hline
\tikz\draw[color0,fill=color0] (0,0) circle (.5ex); & 1 & 2 & 2 & 2 & 2 \\
 \hline
\tikz\draw[color1,fill=color1] (0,0) circle (.5ex); & 1 & 2 & 4 & 4 & 4 \\
 \hline
\tikz\draw[color2,fill=color2] (0,0) circle (.5ex); & 1 & 2 & 4 & 8 & 8 \\
 \hline
\tikz\draw[color3,fill=color3] (0,0) circle (.5ex); & 1 & 2 & 4 & 8 & 16 \\
 \hline
\end{tabular}
\vspace{5mm}
\caption{The number of states in each layer of a DSG instance having four actions available for both the leader and follower ($n = 4$ and $m = 4$).}
\end{table}

\begin{figure}[t]
    \centering
    \input{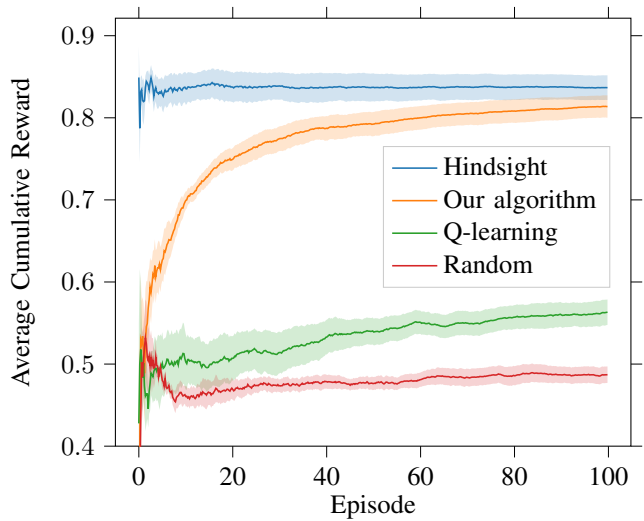}
    \caption{The average cumulative reward for different types of policies. Results are averaged over randomly generated DSGs with 9 states.}
    \label{fig:q_comparison}
\end{figure}

Fig. \ref{fig:varying_s} shows the average regret that our algorithm incurs over a horizon of $T=20$ with varying size state spaces but fixed $p=4$.
The number of states in each layer is varied according to the  chart in Table \ref{fig:chart}.
In agreement with the theoretical analysis, the plot indicates that the regret of our algorithm is independent of the size of the state space.

Fig. \ref{fig:q_comparison} compares the average cumulative reward between the best policy in hindsight, our algorithm, Q-learning, and a policy that randomly chooses mixed policies. The Q-learning implementation discretizes the action space into mixed strategies of the form $\left(\frac{z_1}{10},\frac{z_2}{10},\frac{z_3}{10},\frac{z_4}{10}\right)$ for integers $z_1, z_2, z_3, z_4 \in \mathbb{Z}$ such that $\sum_i z_i = 10$. The state space is partitioned into layers of size $(1,2,2,2,2)$. 
As the figure demonstrates, Q-learning significantly under-performs in comparison to our algorithm that takes advantage of the knowledge of the game's reward structure. Larger state spaces lead to even greater discrepancies in performance since the amount of time it takes Q-learning to learn the structure of the induced MDP scales with the size of the state space while our algorithm provably does not.

\section{Conclusions}

We introduce a new class of game, called a discrete-time dynamic Stackelberg game (DSG), by unifying the sequential decision making aspects of a Markov decision process and the asymmetric interaction between two strategic agents from a Stackelberg game. DSGs extend standard repeated Stackelberg games to scenarios with dynamic reward structures and action spaces. We study DSGs in an online learning setting and give a novel no-regret learning algorithm for playing in a DSG against a follower with unknown utility function. 
Experimental results show that even in practice, our algorithm outperforms model-free reinforcement learning approaches for solving DSGs. Moreover, our algorithm achieves regret independent of the size of the state space, providing a scalable solution in environments with large state spaces. 

\subsection{Future Work}

\textit{a) Tightness of the analysis:} The most limiting aspect of our result seems to be the nature of the regret's dependence on the dimension of $\thetaV^*$, the parameter of the follower's utility function. Lemma \ref{lemma:mistake_bound} achieves an upper bound on the number of mistakes that is linear in the volume of the initial version space of $\thetaV$. A closely related problem to this subroutine of the main algorithm is the problem of exactly learning an unknown halfspace through query synthesis for which there exist algorithms that achieve sample complexity that is logarithmic in the volume of the initial version space \cite{alabdu2015}. Bridging the gap between the sample efficiency of our subroutine and the efficiency of this closely related problem could be a fruitful method for improving the regret bound of our result.

\textit{b) Function class of the follower's utility:}
Our algorithm relies on the fact that the follower's utility function is linearly parameterized. Without this assumption, estimating the follower's utility function is challenging. 
For any prior observation $(\x, b)$, we know that $\forall b' \in \followacts$,
\begin{subequations}
\begin{align}
    \E_{a \samp \x} [u(a, b)] & \geq  \E_{a \samp \x} [u(a, b')] \\
    \langle \x, \thetaV_b \rangle & \geq \langle \x, \thetaV_{b'} \rangle
\end{align}
\end{subequations}
for some vectors $\thetaV_b$.
Estimating the decision boundaries induced by the vectors $\thetaV_b$ from queries
requires exactly learning this multi-class classification problem. Solving the posed learning problem for discrete-time dynamic Stackelberg games with a generalized follower utility function would likely require incorporating some form of online learning algorithm for the multi-class classification problem.

\textit{c) Computational complexity:} In order to compute a policy for each episode, our algorithm requires solving $|\states| \cdot |\followacts|$ copies of the nonconvex quadratic program in Equation \eqref{eq:cons_opt}. In general, solving a nonconvex quadratic program is NP-hard which could pose computational limitations when scaling to larger problem instances. Designing more computationally efficient algorithms for Problem \ref{prb:learning} could be a useful direction for future research.

\bibliographystyle{IEEEtran}
\bibliography{references}

\begin{appendices}

\section*{Proof of Proposition \ref{prop:scale_f}}
%\begin{proof}
    Let $c \in \R$ and define an augmented feature function $f'(a,b) = c \cdot f(a,b)$. A corresponding augmented utility function for the follower can be defined as 
    \begin{equation}
            u'(a,b) = \langle \feat'(a,b) , \thetaV^* \rangle
            = c \cdot u(a,b).
    \end{equation}
    The follower's corresponding augmented policy is therefore 
    \begin{equation}
     \begin{aligned}
    \varphi'(\x) &= \argmax_{b \in \followacts} \E_{a \samp \x} \left[ c \cdot u(a, b) \right] \\
    &= \argmax_{b \in \followacts} \E_{a \samp \x} c \cdot \left[ u(a, b) \right]
    = \varphi(\x),
        \end{aligned}
\end{equation}
identical to the original utility function.
%\end{proof}

\section*{Proof of Lemma \ref{lemma:mistake_shrink}}
    From Equation \eqref{eq:cons_opt}, we see that
    $\x_s^T(\featM_{b_s} - \featM_{b^*}) \thetaV_s \geq \epsilon$.
    Since $b^*$ was the optimal choice for the follower, we have $\x_s^T(\featM_{b_s} - \featM_{b^*}) \thetaV_s^* \leq 0$. Combining these, we have the following.
    \begin{align*}
        \x_s^T(\featM_{b_s} - \featM_{b^*}) \thetaV_s - \x_s^T(\featM_{b_s} - \featM_{b^*}) \thetaV^* &\geq \epsilon \\
        \x_s^T (\featM_{b_s} - \featM_{b^*}) (\thetaV_s - \thetaV^*) &\stackrel{(a)}{\geq} \epsilon \\
        ||\x_s^T (\featM_{b_s} - \featM_{b^*})|| \cdot ||\thetaV_s - \thetaV^*|| &\geq \epsilon \\
        ||\x_s^T|| \cdot ||\thetaV_s - \thetaV^*|| & \stackrel{(b)}{\geq} \epsilon \\
        ||\thetaV_s - \thetaV^*|| &\stackrel{(c)}{\geq} \epsilon
    \end{align*}
where step $(a)$ is by Cauchy-Schwartz inequality, step $(b)$ follows from Proposition \ref{prop:scale_f} and step $(c)$ follows from the fact that $\x_s^T$ is a probability vector.
Therefore, $\thetaV^* \not\in B_{\epsilon}(\thetaV_s)$ and the proof is complete.

\section*{Proof of Lemma \ref{lemma:mistake_bound}}

Since we shrink $\Theta$ by at least $\Theta \cap B_{\epsilon}(\thetaV_s)$ with mistake $(\x_s, \thetaV_s, b_s)$, we know that any future mistakes cannot fall within the ball $B_{\epsilon}(\thetaV_s)$. Let $\thetaV_1$ and $\thetaV_2$ be two estimates for the halfspace $\thetaV^*$ that induced a mistake during a previous observation. It must be that $B_{\frac{\epsilon}{2}}(\thetaV_1) \cap B_{\frac{\epsilon}{2}}(\thetaV_2) = \emptyset$. So, the number of mistakes is upper-bounded by the number of $\frac{\epsilon}{2}$-balls that can fit in the initial space $\Theta_0 = \{\thetaV \in \R^p \mid ||\thetaV|| = 1\} = \{\thetaV \in \R^{p-1} \mid ||\thetaV|| \leq 1\}$. Let $V_{p-1}(r)$ denote the $(p-1)$-dimensional volume of a $(p-1)$-sphere with radius $r$. An upper bound on the number of mistakes is therefore given by the ratio between the volumes
\begin{equation}
    \frac
    {\text{V}_{p-1}(1)}
    {\text{V}_{p-1}(\frac{\epsilon}{2})}
    = \frac{\alpha}
    {\alpha \left(\frac{\epsilon}{2}\right)^{p-1}}
    = \left(\frac{2}{\epsilon} \right)^{p-1}
\end{equation}
where $\alpha = \frac{\pi^{(\frac{p-1}{2})}}{\Gamma(\frac{p-1}{2} + 1)}$ and $\Gamma$ is the gamma function.

\section*{Proof of Lemma \ref{lemma:actual_value_bound}}

    The inequality $\actualV(\pi^{\star}, s) \geq V(\pi_t, s)$ is obvious since $\pi^{\star}$ is the optimal policy.
    Let $\tau \samp \pi_t(s)$ represent a \textit{trace}, i.e. a sequence of state-action pairs sampled from following policy $\pi_t$ beginning in state $s$. Let $\traceV(\tau)$ represent the actual reward obtained from $\tau$ and let $\estV_t(\tau)$ represent the \textit{estimated reward} obtained from $\tau$, i.e. the value that the $\estV_t(s)$ expects outcome $\tau$ to achieve. Notice that $\estV_t(\tau) = \actualV(\tau)$ if $\tau$ does not make a mistake. Then $\actualV(\pi_t, s) = \E_{\tau \samp \pi_t(s)} [V(\tau)]$. Let $\mistake(\tau)$ be a Boolean value that is true if and only if $\tau$ does not make a mistake under the learning policy. Then we have the following.
    \begin{subequations}
    \begin{align}
        \actualV(\pi_t, s) 
        &= \E_{\tau \samp \pi_t(s)} \left[ \actualV(\tau) \mid \text{$\mistake(\tau)$} \right] \cdot \pr \left[ \text{$\mistake(\tau)$} \right]\\&\qquad +  \E_{\tau \samp \pi_t(s)} \left[ \actualV(\tau) \mid \text{$\lnot \mistake(\tau)$} \right] \cdot \pr \left[ \text{$\lnot \mistake(\tau)$} \right] \\
        &= \E_{\tau \samp \pi_t(s)} \left[ \actualV(\tau) \mid \text{$\mistake(\tau)$} \right] (1-\lambda_t(s)) \\
        &\qquad+  \E_{\tau \samp \pi_t(s)} \left[ \actualV(\tau) \mid \text{$\lnot \mistake(\tau)$} \right] \lambda_t(s) \\
        &\stackrel{(a)}{\geq} \E_{\tau \samp \pi_t(s)} \left[ \actualV(\tau) \mid \text{$\mistake(\tau)$} \right] (1-\lambda_t(s)) \\
        &\stackrel{(b)}{=} \E_{\tau \samp \pi_t(s)} \left[ \estV_t(\tau) \mid \text{$\mistake(\tau)$} \right] (1-\lambda_t(s)) \\
        &= \E_{\tau \samp \pi_t(s)} \left[ \estV_t(\tau) \right] - \E_{\tau \samp \pi_t(s)} \left[ \estV_t(\tau) \mid \lnot \text{$\mistake(\tau)$} \right] \lambda_t(s) \\
        &\geq \E_{\tau \samp \pi_t(s)} \left[ \estV_t(\tau) \right] - H \cdot \lambda_t(s) \\
        &= \estV_t(s) - H \cdot \lambda_t(s) \\
        &\stackrel{(c)}{\geq} \actualV(\pi^{\epsilon}) - H \cdot \lambda_t(s).
    \end{align}
    \end{subequations}
    The inequality (a) holds because rewards are normalized to be positive, the equality (b) comes from the fact that $\estV_t(\tau) = \actualV(\tau)$ on traces $\tau$ that the learning policy doesn't make a mistake on, and (c) comes from the fact that $\thetaV^* \in \Theta$ so $\estV_t(s) \geq \actualV(\pi^{\epsilon}, s)$.

\section*{Proof of Lemma \ref{lem:aux}}
% \setcounter{lemma}{3}
% \begin{lemma}
% Let $\M \in \R^{n \times q}$ with $q < n$ with rank $q$ and $\x \in \R^n$. If $\x^T \M = \ep$, then there exists a full rank matrix $\M' \in \R^{n \times n}$ with the first $q$ columns identical to $\M$ such that $\x^T \M' = (\ep, \0)$. Moreover, $\M'$ can be constructed to have minimum singular value equal to the minimum singular value of $\M$. 
% \end{lemma}

Let $\A = \begin{pmatrix} \v_1 & \dots & \v_q \end{pmatrix}$. We just need to show that there exists some $\v$ such that $\x^T \v = 0$ and $\v \not\in \spn\{v_1, ..., \v_q\}$. We will also show $v$ can be chosen such that the resulting matrix $\A' = \begin{pmatrix} \A & \v \end{pmatrix}$ has $\sigma_{\min}(\A') = \sigma_{\min}(\A)$. Then, through repeated application of this result, the lemma holds. 

Let $\stcomp{\x}$ denote the complement space of $\x$.
For sake of contradiction, suppose no such vector exists. Then $\stcomp{\x} \subset \spn\{\v_1,\dots,\v_q\}$. But since $\dim(\stcomp{\x}) = n-1$, it must be that $\x \in \ker\{\v_1, \dots, \v_q\}$ since the kernel space is nontrivial. But this is a contradiction, since $\x^T \M \neq \0$ by definition.

So, there exists $\v$ such that that $\x^T \v = 0$ and $\v \not\in \spn\{\v_1, ..., \v_q\}$. Let $\A' = \begin{pmatrix} \A & \v \end{pmatrix}$. Then,
\begin{equation}
\begin{aligned}
    (\A')^T \A' &= \begin{pmatrix} \A^T \\ \v^T
                \end{pmatrix}
                \begin{pmatrix} \A & \v
                \end{pmatrix}
                % = \begin{pmatrix} \A^T \A & \A^T \v \\ \v^T \A & \v^T \v
                % \end{pmatrix}
               = \begin{pmatrix} \A^T \A & \0 \\ \0 & \v^T\ \v
                \end{pmatrix}
\end{aligned}
\end{equation}
Since the singular values of $\A'$ are the square roots of the eigenvalues of $(\A')^T \A'$ we have that $\sigma_{\min}(\A') = \min( \sigma_{\min}(\A), ||\v||)$. Since $||\v||$ can be arbitrarily chosen, we can construct $\A'$ such that $\sigma_{\min}(\A') = \sigma_{\min}(\A)$ holds.

% Let $A = [v_{q+1}, \dots, v_n]$.
% Notice that we can choose $v_{q+1}, \dots, v_n$ with arbitrary norm. Therefore, we can make the operator norm of $A$ arbitrarily small.
% Therefore $\forall \x \in \R^n$, by the triangle inequality we have that 
% \begin{equation}
%     ||\x^T M|| \leq ||\x^{T}M'|| \leq ||\x^{T}M|| + ||\x^{T}A|| \leq ||\x|| \cdot \opnorm{M} + ||\x|| \cdot \opnorm{A}
% \end{equation}
% and
% \begin{equation}
%     ||\x^{T}M'|| \geq ||\x^{T}M|| + ||\x^{T}A|| \leq ||\x|| \cdot \opnorm{M} + ||\x|| \cdot \opnorm{A}
% \end{equation}
% so the operator norm of $M'$ can be arbitrarily close to the operator norm of $M$.

\section*{Proof of Lemma \ref{lemma:epsilon_diff}}

Define vectors $\v_{s,b}^*, \v_{s,b}^{\epsilon} \in \R^{n}$ such that
\begin{equation*}
\begin{aligned}
    [\v_{s,b}^*]_a &= \sum\limits_{s'} P(s,a,b,s') V(\pi^{\star}, s'), \; \text{and}\\
    [\v_{\epsilon, s,b}^*]_a &= \sum\limits_{s'} P(s,a,b,s') V(\pi^{\epsilon}, s'),
\end{aligned}
\end{equation*}
and vector $\r_{s,b} \in \R^n$ such that 
\begin{equation*}
    [\r_{s,b}]_a = r(s,a,b).
\end{equation*}
We have that $||\x^T - {\x^*}^T|| \leq \sqrt{m} \cdot \epsilon \cdot d$ from Equation \eqref{eq:pre_lem5} which in turn gives the following.
\begin{subequations}
\begin{align}
    % &\E_{a \samp \x} \left[r(s,a,b) + \displaystyle\sum\limits_{s'} P(s,a,b,s') V(\pi^{\epsilon}, s') \right] - \E_{a \samp \x^*} \left[r(s,a,b) + \displaystyle\sum\limits_{s'} P(s,a,b,s') V(\pi^{\star}, s') \right] \\
    V&(\pi^{\star}, s) - V(\pi^{\epsilon}, s) \\
    &= \x^T\left(\r_{s,b} + \v_{s,b}^{\epsilon} \right)  - {\x^*}^T \left(\r_{s,b} + \v^*_{s,b} \right)  \\
    &= (\x^T  - {\x^*}^T) \r_{s,b} + \x^T \v_{s,b}^{\epsilon} - {\x^*}^T \v^*_{s,b} \\
    &= (\x^T  - {\x^*}^T) \r_{s,b} + (\x^T \v^*_{s,b} - \x^T \v^*_{s,b}) \\&\qquad+ \x^T \v_{s,b}^{\epsilon} - {\x^*}^T \v^*_{s,b} \\
    &= (\x^T  - {\x^*}^T) \r_{s,b} + \x^T (\v_{s,b}^{\epsilon} - \v^*_{s,b}) + (\x^T - {\x^*}^T) \v^*_{s,b} \\
    & \leq ||\x^T - {\x^*}^T|| \cdot ||\r_{s,b}||  + ||\x|| \cdot ||\v_{s,b}^{\epsilon} - \v^*_{s,b}||\\&\qquad + ||\x^T - {\x^*}^T|| \cdot ||\v^*_{s,b}||\\
    &\leq \epsilon \cdot d \cdot \sqrt{m} \cdot ||\r_{s,b}||  + \epsilon \cdot d \cdot \sqrt{mn} \cdot \sqrt{nH} + ||\v_{s,b}^{\epsilon} - \v^*_{s,b}|| \\
    &\leq \epsilon \cdot d \cdot \sqrt{mn} + \epsilon \cdot d \cdot \sqrt{mn} \cdot \sqrt{nH} + ||\v_{s,b}^{\epsilon} - \v^*_{s,b}|| \\
    &\leq \epsilon d \sqrt{mn} (1 + \sqrt{nH}) + ||\v_{s,b}^{\epsilon} - \v^*_{s,b}|| 
\end{align}
\end{subequations}
since $\r_{s,b} \in [0,1]^n$ and $\v^*_{s,b} \in [0,H-1]^n$. $||\v_{s,b}^{\epsilon} - \v^*_{s,b}||$ is upper bounded by the maximum difference $\max_{s' \in \states_{l+1}}\{V(\pi^{\epsilon}, s') - V(\pi^{\star}, s')\}$ over valuations in the next layer and $V(\pi^{\epsilon}, s') = V(\pi^{\star}, s') = 0$ for $s' \in \states_H$. Therefore, for any $s \in \states$, recursive application of this computation gives the desired bound 
\begin{equation}
V(\pi^{\star}, s) - V(\pi^{\epsilon}, s)\leq (\epsilon d  \sqrt{mn} (1 + \sqrt{nH}))H. 
\end{equation}

\end{appendices}

% \section{ToDo}
% \begin{itemize}
%     \item Double check the $\leadacts(s) = \leadacts$ assumption throughout the paper. 
%     \item Expand the algorithm section to have more details.
% \end{itemize}

\begin{IEEEbiography}[{\includegraphics[width=1in,height=1.25in,clip,keepaspectratio]{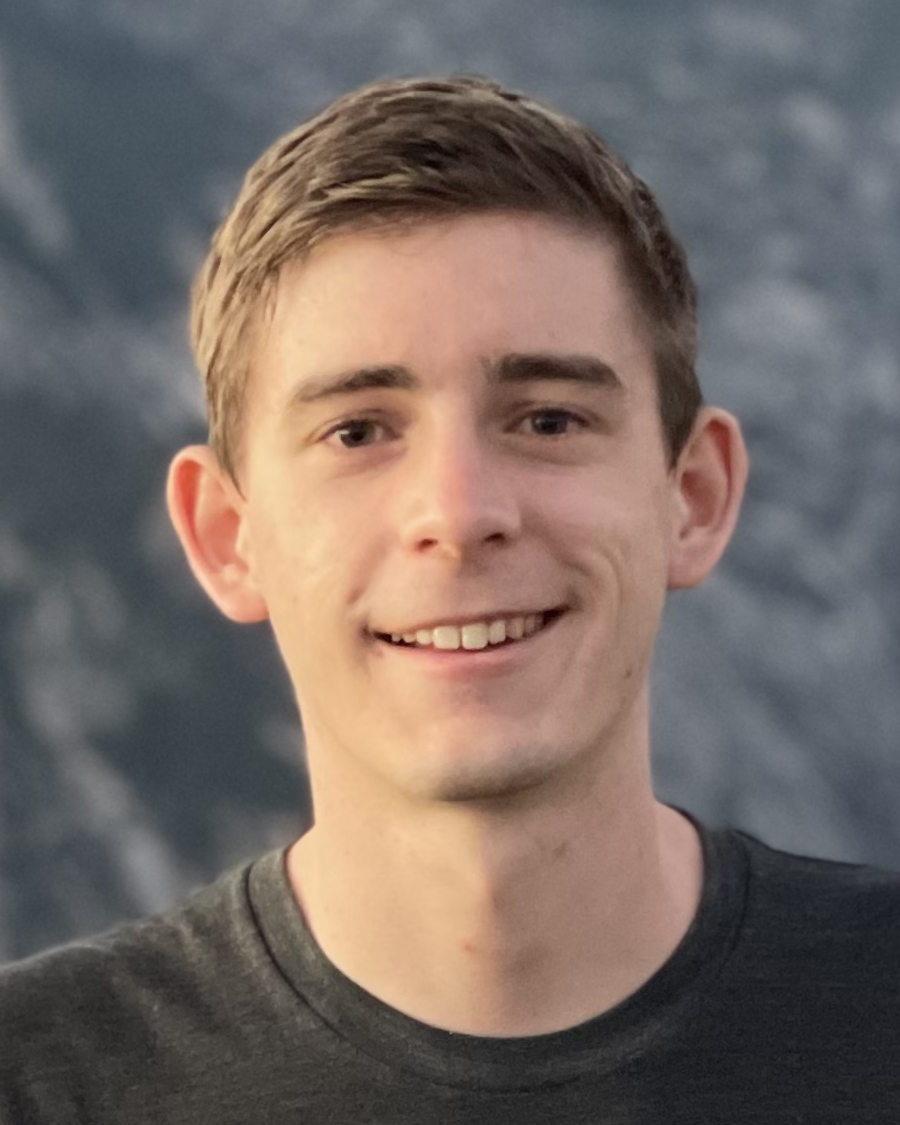}}]{Niklas Lauffer} is a PhD student in the Department of Electrical Engineering and Computer Sciences at the University of California, Berkeley. He received his B.S. degree in Computer Science and Math from the University of Texas at Austin in 2021. His research focuses on developing autonomous learning and decision making systems that are provably safe and beneficial by employing ideas from game theory, formal methods, and learning theory.
\end{IEEEbiography}

\vskip -0.1\baselineskip plus -1fil

\begin{IEEEbiography}[{\includegraphics[width=1in,height=1.25in,clip,keepaspectratio]{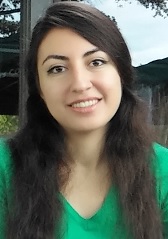}}]{Mahsa Ghasemi}
is an Assistant Professor at the School of Electrical and Computer Engineering at Purdue University. Her research focuses on theoretical and foundational advancements preparing autonomous systems to co-exist with humans in our complex world. Her contributions enable efficient and reliable integration of autonomy in various applications such as robotics, shared autonomy, and networked systems. 
She received her B.Sc. degree in Mechanical Engineering from Sharif University of Technology, and her M.S.E. and Ph.D. degrees in Mechanical Engineering and Electrical and Computer Engineering, respectively, from The University of Texas at Austin. 
\end{IEEEbiography}

\vskip -0.1\baselineskip plus -1fil

\begin{IEEEbiography}[{\includegraphics[width=1in,height=1.25in,clip,keepaspectratio]{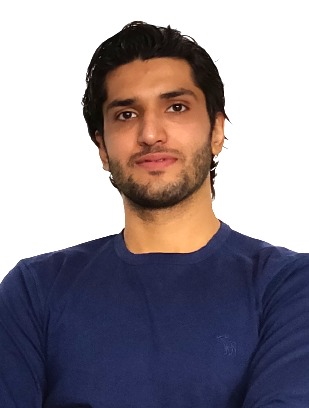}}]{Abolfazl Hashemi}
is an Assistant Professor at the School of Electrical and Computer Engineering at Purdue University. His research goal is to enhance the performance and capabilities of the networked systems characterized by limited communication budgets and data scarcity. Abolfazl received his Ph.D. and M.S.E. degrees in the Electrical and Computer Engineering department at UT Austin in 2020 and 2016. Before that, He received his B.Sc. degree in Electrical Engineering from the Sharif University of Technology in 2014. He was the recipient of the Iranian national elite foundation fellowship and a best student paper award finalist at the 2018 American Control Conference.
\end{IEEEbiography}

\vskip -0.1\baselineskip plus -1fil

\begin{IEEEbiography}[{\includegraphics[width=1in,height=1.25in,clip,keepaspectratio]{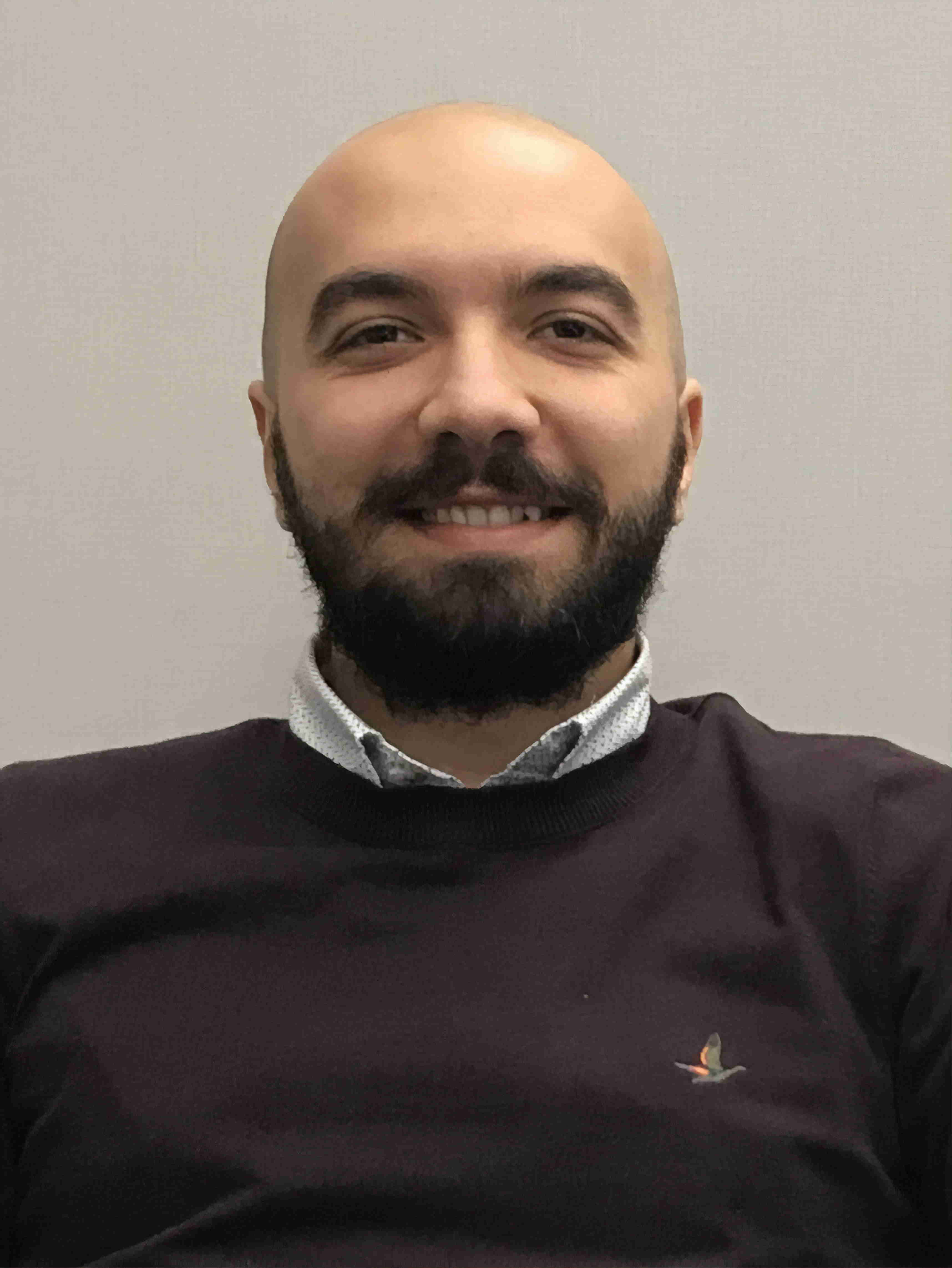}}]{Yagiz Savas} is a PhD candidate in the Department of Aerospace Engineering, University of Texas at Austin. He received his B.Sc. degree in Mechanical Engineering from Bogazici University, Turkey in 2017. His research focuses on developing socially intelligent autonomous systems that co-exist, cooperate, and compete with each other, as well as with humans, by drawing novel connections between controls, formal methods, and information theory.
\end{IEEEbiography}

\vskip -0.1\baselineskip plus -1fil

\begin{IEEEbiography}[{\includegraphics[width=1in,height=1.25in,clip,keepaspectratio]{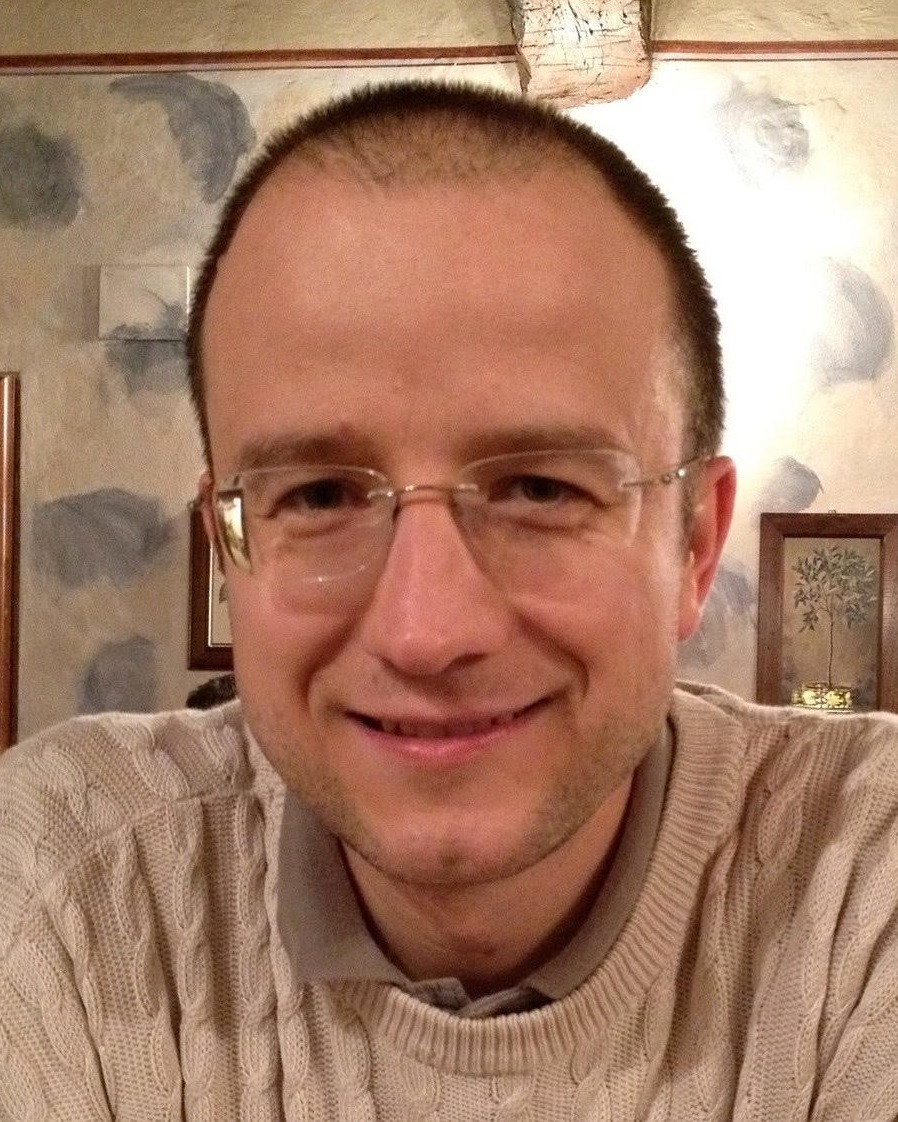}}]{Ufuk Topcu} 
is  an  associate  professor  in  the  Department of Aerospace Engineering and Engineering Mechanics and the Oden Institute at The University of  Texas  at  Austin.  He  received  his  Ph.D.  degree from  the  University  of  California  at  Berkeley  in 2008.  His  research  focuses  on  the  theoretical,  algorithmic, and computational aspects of design and verification  of  autonomous  systems  through  novel connections  between  formal  methods,  learning  theory, and controls.
\end{IEEEbiography}

\end{document}

%% file: defpack.tex
\usepackage{float}
\usepackage[caption=false,font=footnotesize]{subfig}
\usepackage{tabularx}
\usepackage{verbatim}
\usepackage{hyperref}
\usepackage{xmpmulti}
\usepackage{transparent}
\usepackage{tikz, pgfplots} 
\usetikzlibrary{arrows,decorations.pathmorphing,positioning,fit,trees,shapes,shadows,automata,calc} 
\usepackage{float}
\usepackage{algorithm, algpseudocode}
\usepackage{booktabs}
\usepackage{amsthm}
\usepackage{hyperref}
\usepackage{wrapfig}
\usepackage{setspace}
\usepackage{epsfig}
\usepackage[all]{xy}
\usepackage{mathtools}

\usepackage{mathrsfs}  
\usepackage{bm}
\usepackage{multirow}
\usepackage{soul}
\usepackage{amsthm}

\usepackage{scrextend}

\usepackage{etoolbox}
\makeatletter
\@ifundefined{color@begingroup}%
  {\let\color@begingroup\relax
   \let\color@endgroup\relax}{}%
\def\fix@ieeecolor@hbox#1{%
  \hbox{\color@begingroup#1\color@endgroup}}
\patchcmd\@makecaption{\hbox}{\fix@ieeecolor@hbox}{}{\FAILED}
\patchcmd\@makecaption{\hbox}{\fix@ieeecolor@hbox}{}{\FAILED}
\newtheorem{theorem}{Theorem} 
\newtheorem{assumption}{Assumption} 
\newtheorem{lemma}{Lemma}            

\newtheorem{problem}{Problem}
                       
\newtheorem{definition}{Definition}
 
\newtheorem{proposition}{Proposition}

\newtheorem{remark}{Remark}
\DeclareMathOperator*{\argmax}{arg\,max}

\let\oldIEEEkeywords\IEEEkeywords
\def\IEEEkeywords{\oldIEEEkeywords\normalfont\bfseries\ignorespaces}

\newcommand{\argsup}{\mathrm{arg}\sup}

\def\R{\mathbb{R}}

\newcommand{\pr}{\ensuremath{\mathrm{Pr}}}

\newcommand{\states}{\mathcal{S}}
\newcommand{\leadacts}{\mathcal{A}}
\newcommand{\followacts}{\mathcal{B}}

\newcommand{\feat}{f}
\newcommand{\featM}{\mathbf{M}}
\newcommand{\projM}{\mathbf{M}}

\newcommand{\samp}{\sim}
\newcommand{\E}{\mathop{\mathbb{E}}}

\newcommand{\estV}{\widetilde V}
\newcommand{\actualV}{V}
\newcommand{\traceV}{V}
\newcommand{\mistake}{\mu}
\newcommand{\Ttrue}{T_{\text{true}}}

\newcommand{\opnorm}[1]{\ensuremath{||{#1}}||_{op}}
\newcommand{\spn}{\ensuremath{\text{span}}}
\newcommand{\stcomp}[1]{{#1}^{\perp}}  

\def\x{{\mathbf x}}

\def\thetaV{{\boldsymbol{\theta}}}
\def\0{{\mathbf{0}}}

\def\v{{\mathbf v}}
\def\r{{\mathbf r}}

\def\A{{\mathbf A}}
\def\M{{\mathbf M}}

\def\R{{\mathbb{R}}}

%% file: poaching_example_T=30.tex
% This file was created with tikzplotlib v0.9.14.
\begin{tikzpicture}

\definecolor{color0}{rgb}{0.12156862745098,0.466666666666667,0.705882352941177}
\definecolor{color1}{rgb}{1,0.498039215686275,0.0549019607843137}
\definecolor{color2}{rgb}{0.172549019607843,0.627450980392157,0.172549019607843}

\begin{axis}[
legend cell align={left},
legend style={at={(.98,.28)}, fill opacity=0.8, draw opacity=1, text opacity=1, draw=white!80!black},
tick align=outside,
tick pos=left,
x grid style={white!69.0196078431373!black},
xlabel={Episode},
xmin=-5.95, xmax=124.95,
xtick style={color=black},
xtick={0,20,40,60,80,100,120},
xticklabels={0,5,10,15,20,25,30},
y grid style={white!69.0196078431373!black},
ylabel={Average Cumulative Reward},
ymin=-1.56629298903831, ymax=0.194512656009218,
ytick style={color=black}
]
\path [fill=color0, fill opacity=0.2]
(axis cs:0,0.114476035779784)
--(axis cs:0,-0.593761122361325)
--(axis cs:1,-0.653657921991399)
--(axis cs:2,-0.676397194275266)
--(axis cs:3,-0.422584671583451)
--(axis cs:4,-0.259092717916728)
--(axis cs:5,-0.197270457355834)
--(axis cs:6,-0.190931733366078)
--(axis cs:7,-0.185302279557557)
--(axis cs:8,-0.177281409573333)
--(axis cs:9,-0.237598797394927)
--(axis cs:10,-0.210924235881456)
--(axis cs:11,-0.257578713086698)
--(axis cs:12,-0.203530608703565)
--(axis cs:13,-0.219441633087776)
--(axis cs:14,-0.213899270684863)
--(axis cs:15,-0.248704188913582)
--(axis cs:16,-0.262794240363228)
--(axis cs:17,-0.291553187442453)
--(axis cs:18,-0.262550746579748)
--(axis cs:19,-0.263200079088527)
--(axis cs:20,-0.244876209821679)
--(axis cs:21,-0.269220713365661)
--(axis cs:22,-0.277177546657068)
--(axis cs:23,-0.267021751905564)
--(axis cs:24,-0.281007355798162)
--(axis cs:25,-0.294786303666016)
--(axis cs:26,-0.278530963412321)
--(axis cs:27,-0.276621174275593)
--(axis cs:28,-0.276741876987341)
--(axis cs:29,-0.277227728472859)
--(axis cs:30,-0.278458536503243)
--(axis cs:31,-0.293843518560778)
--(axis cs:32,-0.296902645454858)
--(axis cs:33,-0.306642014839765)
--(axis cs:34,-0.3050850949904)
--(axis cs:35,-0.305243333583683)
--(axis cs:36,-0.310717441460965)
--(axis cs:37,-0.31297545155802)
--(axis cs:38,-0.323460074308917)
--(axis cs:39,-0.325520854934264)
--(axis cs:40,-0.325441783667436)
--(axis cs:41,-0.318636339010601)
--(axis cs:42,-0.308977817519643)
--(axis cs:43,-0.302714356581839)
--(axis cs:44,-0.28174672912732)
--(axis cs:45,-0.287657272926668)
--(axis cs:46,-0.286591403467762)
--(axis cs:47,-0.296678623444357)
--(axis cs:48,-0.304422079125234)
--(axis cs:49,-0.307400176569316)
--(axis cs:50,-0.310248321865219)
--(axis cs:51,-0.305185134927456)
--(axis cs:52,-0.295997107741509)
--(axis cs:53,-0.288601044646989)
--(axis cs:54,-0.291997784843618)
--(axis cs:55,-0.300547430870417)
--(axis cs:56,-0.307111733990193)
--(axis cs:57,-0.31527282974535)
--(axis cs:58,-0.312744384120527)
--(axis cs:59,-0.311654207491706)
--(axis cs:60,-0.307922955867896)
--(axis cs:61,-0.307820651542253)
--(axis cs:62,-0.310409312027264)
--(axis cs:63,-0.317602572188468)
--(axis cs:64,-0.321364132314291)
--(axis cs:65,-0.316323109791168)
--(axis cs:66,-0.309882067812364)
--(axis cs:67,-0.316659948614305)
--(axis cs:68,-0.3206303851782)
--(axis cs:69,-0.327199312358358)
--(axis cs:70,-0.316049350979508)
--(axis cs:71,-0.312884178677492)
--(axis cs:72,-0.307358716179734)
--(axis cs:73,-0.309100615283913)
--(axis cs:74,-0.309618673501181)
--(axis cs:75,-0.315686558880696)
--(axis cs:76,-0.307520348995565)
--(axis cs:77,-0.302025752651512)
--(axis cs:78,-0.307340309438746)
--(axis cs:79,-0.308493347220624)
--(axis cs:80,-0.310910051266394)
--(axis cs:81,-0.308179122350703)
--(axis cs:82,-0.311308438540437)
--(axis cs:83,-0.312571855066243)
--(axis cs:84,-0.312353757748328)
--(axis cs:85,-0.310854860353609)
--(axis cs:86,-0.310478116796934)
--(axis cs:87,-0.3104599990361)
--(axis cs:88,-0.303979061301037)
--(axis cs:89,-0.291701278538381)
--(axis cs:90,-0.296428470409861)
--(axis cs:91,-0.29723824054819)
--(axis cs:92,-0.292499793495567)
--(axis cs:93,-0.287324449046968)
--(axis cs:94,-0.288985035970643)
--(axis cs:95,-0.290192004243568)
--(axis cs:96,-0.281613857694457)
--(axis cs:97,-0.282593456931427)
--(axis cs:98,-0.283311861604986)
--(axis cs:99,-0.288186522612379)
--(axis cs:100,-0.291015951476997)
--(axis cs:101,-0.289797920777185)
--(axis cs:102,-0.290661433897642)
--(axis cs:103,-0.287850529753507)
--(axis cs:104,-0.287446876368359)
--(axis cs:105,-0.284144146424855)
--(axis cs:106,-0.28247056172787)
--(axis cs:107,-0.28177293156421)
--(axis cs:108,-0.276824003536456)
--(axis cs:109,-0.271356411915631)
--(axis cs:110,-0.267292978713426)
--(axis cs:111,-0.269683962250628)
--(axis cs:112,-0.272424958261357)
--(axis cs:113,-0.267977605769662)
--(axis cs:114,-0.26912805563966)
--(axis cs:115,-0.268245847903323)
--(axis cs:116,-0.272786936204265)
--(axis cs:117,-0.274887485456313)
--(axis cs:118,-0.274023534122202)
--(axis cs:119,-0.278163154357387)
--(axis cs:119,-0.209870708825716)
--(axis cs:119,-0.209870708825716)
--(axis cs:118,-0.205157202493627)
--(axis cs:117,-0.205345620265168)
--(axis cs:116,-0.202962524967651)
--(axis cs:115,-0.197819502087084)
--(axis cs:114,-0.202041435243863)
--(axis cs:113,-0.199435684771458)
--(axis cs:112,-0.204996004361126)
--(axis cs:111,-0.202721154302375)
--(axis cs:110,-0.200472541302038)
--(axis cs:109,-0.20238310342379)
--(axis cs:108,-0.20734847666528)
--(axis cs:107,-0.212813517760355)
--(axis cs:106,-0.212689174602055)
--(axis cs:105,-0.217473455017804)
--(axis cs:104,-0.219970536488944)
--(axis cs:103,-0.219837386241064)
--(axis cs:102,-0.217749231380847)
--(axis cs:101,-0.214779747789158)
--(axis cs:100,-0.214114679545816)
--(axis cs:99,-0.21160345749662)
--(axis cs:98,-0.205955230174925)
--(axis cs:97,-0.208043695150842)
--(axis cs:96,-0.209455342755023)
--(axis cs:95,-0.212664752305009)
--(axis cs:94,-0.208771261886657)
--(axis cs:93,-0.206159295277778)
--(axis cs:92,-0.207444664177778)
--(axis cs:91,-0.209016760678574)
--(axis cs:90,-0.210628841382571)
--(axis cs:89,-0.204948320299678)
--(axis cs:88,-0.218201833813501)
--(axis cs:87,-0.224933590171374)
--(axis cs:86,-0.226089597116937)
--(axis cs:85,-0.227824092301717)
--(axis cs:84,-0.227383435991615)
--(axis cs:83,-0.227885006616062)
--(axis cs:82,-0.223763562051054)
--(axis cs:81,-0.222450142517302)
--(axis cs:80,-0.227595656341895)
--(axis cs:79,-0.230752618150499)
--(axis cs:78,-0.226232815845684)
--(axis cs:77,-0.219878419397001)
--(axis cs:76,-0.222556118756111)
--(axis cs:75,-0.222425461726148)
--(axis cs:74,-0.215114095051239)
--(axis cs:73,-0.217925823355042)
--(axis cs:72,-0.213122611782759)
--(axis cs:71,-0.215945855912829)
--(axis cs:70,-0.214126955369114)
--(axis cs:69,-0.228809859240344)
--(axis cs:68,-0.220814997957027)
--(axis cs:67,-0.217465576590495)
--(axis cs:66,-0.209207182773273)
--(axis cs:65,-0.213500431949097)
--(axis cs:64,-0.218262273100362)
--(axis cs:63,-0.214218062648223)
--(axis cs:62,-0.205383778526062)
--(axis cs:61,-0.201785612569716)
--(axis cs:60,-0.202681305025207)
--(axis cs:59,-0.208803883570302)
--(axis cs:58,-0.209191899336583)
--(axis cs:57,-0.209541645113627)
--(axis cs:56,-0.199525616294756)
--(axis cs:55,-0.196579948483831)
--(axis cs:54,-0.186139984595458)
--(axis cs:53,-0.181736010709155)
--(axis cs:52,-0.195732225161975)
--(axis cs:51,-0.197551022632656)
--(axis cs:50,-0.204716208445441)
--(axis cs:49,-0.19987611774921)
--(axis cs:48,-0.192700271852432)
--(axis cs:47,-0.189150532295254)
--(axis cs:46,-0.176775480592082)
--(axis cs:45,-0.178065387261676)
--(axis cs:44,-0.171996000886973)
--(axis cs:43,-0.188623590498382)
--(axis cs:42,-0.196598473741929)
--(axis cs:41,-0.203486386938623)
--(axis cs:40,-0.212034242843842)
--(axis cs:39,-0.213409945075258)
--(axis cs:38,-0.203971506925981)
--(axis cs:37,-0.190342448191323)
--(axis cs:36,-0.194075848539364)
--(axis cs:35,-0.184705444405427)
--(axis cs:34,-0.186483798376943)
--(axis cs:33,-0.188943575916358)
--(axis cs:32,-0.183251614172036)
--(axis cs:31,-0.186337902723783)
--(axis cs:30,-0.167484997574732)
--(axis cs:29,-0.171747217336089)
--(axis cs:28,-0.1722415683163)
--(axis cs:27,-0.171280069553838)
--(axis cs:26,-0.169913113901497)
--(axis cs:25,-0.178189236837557)
--(axis cs:24,-0.170694425604726)
--(axis cs:23,-0.160201115288934)
--(axis cs:22,-0.173877221403187)
--(axis cs:21,-0.161455801882906)
--(axis cs:20,-0.131979635887364)
--(axis cs:19,-0.136529836368026)
--(axis cs:18,-0.135415230756047)
--(axis cs:17,-0.139042786724208)
--(axis cs:16,-0.101312639602733)
--(axis cs:15,-0.108204053138906)
--(axis cs:14,-0.0640324591918758)
--(axis cs:13,-0.0529851359123877)
--(axis cs:12,-0.0553344917225472)
--(axis cs:11,-0.122895714115468)
--(axis cs:10,-0.0639973279128414)
--(axis cs:9,-0.0837523193407806)
--(axis cs:8,-0.00634087840205931)
--(axis cs:7,-0.0313509880699034)
--(axis cs:6,0.00522010461340602)
--(axis cs:5,0.0397737541691826)
--(axis cs:4,0.0486241546737428)
--(axis cs:3,-0.140612737691914)
--(axis cs:2,-0.299554833154144)
--(axis cs:1,-0.0883943803097167)
--(axis cs:0,0.114476035779784)
--cycle;

\path [fill=color1, fill opacity=0.2]
(axis cs:0,-1.4505048292877)
--(axis cs:0,-1.4505048292877)
--(axis cs:1,-1.38206652621817)
--(axis cs:2,-1.45631284049175)
--(axis cs:3,-1.48625636880888)
--(axis cs:4,-1.3489157308698)
--(axis cs:5,-1.25417232368971)
--(axis cs:6,-1.17812995442589)
--(axis cs:7,-1.1272109554157)
--(axis cs:8,-1.08920827098118)
--(axis cs:9,-1.06621308748322)
--(axis cs:10,-1.05485052436151)
--(axis cs:11,-1.03117781086008)
--(axis cs:12,-1.00377486542286)
--(axis cs:13,-0.986671324320705)
--(axis cs:14,-0.981669175558073)
--(axis cs:15,-0.961296398662628)
--(axis cs:16,-0.910065623892332)
--(axis cs:17,-0.893856163081603)
--(axis cs:18,-0.861651903890867)
--(axis cs:19,-0.860525022104991)
--(axis cs:20,-0.851913759027452)
--(axis cs:21,-0.848665646698443)
--(axis cs:22,-0.843153042009076)
--(axis cs:23,-0.840137413689713)
--(axis cs:24,-0.795006620423426)
--(axis cs:25,-0.777365070466801)
--(axis cs:26,-0.766804327449065)
--(axis cs:27,-0.756115301129628)
--(axis cs:28,-0.7355869709512)
--(axis cs:29,-0.721690599559572)
--(axis cs:30,-0.703499139697776)
--(axis cs:31,-0.708005464546621)
--(axis cs:32,-0.710779637715313)
--(axis cs:33,-0.703036410494389)
--(axis cs:34,-0.691502029076081)
--(axis cs:35,-0.683010941469397)
--(axis cs:36,-0.676291974758351)
--(axis cs:37,-0.679033114574954)
--(axis cs:38,-0.662606526083188)
--(axis cs:39,-0.646884874018422)
--(axis cs:40,-0.650608491215862)
--(axis cs:41,-0.650484744097409)
--(axis cs:42,-0.645618382034031)
--(axis cs:43,-0.633540865685237)
--(axis cs:44,-0.637230028205864)
--(axis cs:45,-0.618296602595789)
--(axis cs:46,-0.620500414005304)
--(axis cs:47,-0.62086270291462)
--(axis cs:48,-0.618393040352008)
--(axis cs:49,-0.618546289359496)
--(axis cs:50,-0.615513482368099)
--(axis cs:51,-0.61745169712996)
--(axis cs:52,-0.620887574148516)
--(axis cs:53,-0.61409709936669)
--(axis cs:54,-0.605203183722166)
--(axis cs:55,-0.592955046184088)
--(axis cs:56,-0.592212021802734)
--(axis cs:57,-0.566952854661664)
--(axis cs:58,-0.559373713690379)
--(axis cs:59,-0.564182749279854)
--(axis cs:60,-0.554564774364341)
--(axis cs:61,-0.556345857300202)
--(axis cs:62,-0.540703923423504)
--(axis cs:63,-0.544298830281643)
--(axis cs:64,-0.539003533688526)
--(axis cs:65,-0.540999837944656)
--(axis cs:66,-0.538523879103462)
--(axis cs:67,-0.540601664948137)
--(axis cs:68,-0.5318941640401)
--(axis cs:69,-0.531162536571982)
--(axis cs:70,-0.526785154293169)
--(axis cs:71,-0.526093933716792)
--(axis cs:72,-0.529839950777021)
--(axis cs:73,-0.523888611669428)
--(axis cs:74,-0.523263012237713)
--(axis cs:75,-0.521418615121604)
--(axis cs:76,-0.520637079672361)
--(axis cs:77,-0.518222136872099)
--(axis cs:78,-0.518729017072107)
--(axis cs:79,-0.521879628888011)
--(axis cs:80,-0.521679861784292)
--(axis cs:81,-0.516924375621345)
--(axis cs:82,-0.513783528482786)
--(axis cs:83,-0.50989029736063)
--(axis cs:84,-0.505803011330314)
--(axis cs:85,-0.507064596725264)
--(axis cs:86,-0.504452438188873)
--(axis cs:87,-0.507478864599731)
--(axis cs:88,-0.503151634145075)
--(axis cs:89,-0.498597496655307)
--(axis cs:90,-0.497668696457295)
--(axis cs:91,-0.498432851101374)
--(axis cs:92,-0.49883751816965)
--(axis cs:93,-0.501833451889007)
--(axis cs:94,-0.498205363153298)
--(axis cs:95,-0.494813715073054)
--(axis cs:96,-0.490147396349551)
--(axis cs:97,-0.490524700229437)
--(axis cs:98,-0.48374950703896)
--(axis cs:99,-0.481525139936575)
--(axis cs:100,-0.482554963407878)
--(axis cs:101,-0.472519477524497)
--(axis cs:102,-0.464523532588661)
--(axis cs:103,-0.468217988740828)
--(axis cs:104,-0.46336310044602)
--(axis cs:105,-0.466354536175673)
--(axis cs:106,-0.468742584798732)
--(axis cs:107,-0.471539208664895)
--(axis cs:108,-0.474548511146075)
--(axis cs:109,-0.470104195573329)
--(axis cs:110,-0.467275656637407)
--(axis cs:111,-0.463880306262352)
--(axis cs:112,-0.460577127770787)
--(axis cs:113,-0.457616235024012)
--(axis cs:114,-0.456156499441561)
--(axis cs:115,-0.455790900583923)
--(axis cs:116,-0.456281483895514)
--(axis cs:117,-0.450078157648902)
--(axis cs:118,-0.452362170936247)
--(axis cs:119,-0.44970198588997)
--(axis cs:119,-0.412136375149478)
--(axis cs:119,-0.412136375149478)
--(axis cs:118,-0.413077817172852)
--(axis cs:117,-0.410460885633275)
--(axis cs:116,-0.41253022297963)
--(axis cs:115,-0.411551978571142)
--(axis cs:114,-0.412160071390354)
--(axis cs:113,-0.410427881780838)
--(axis cs:112,-0.411177747054788)
--(axis cs:111,-0.416377965206677)
--(axis cs:110,-0.417009763766018)
--(axis cs:109,-0.423941764703035)
--(axis cs:108,-0.428684683601586)
--(axis cs:107,-0.42525071605055)
--(axis cs:106,-0.422021489449674)
--(axis cs:105,-0.419192675776152)
--(axis cs:104,-0.415752079471265)
--(axis cs:103,-0.425509223532045)
--(axis cs:102,-0.421350373755095)
--(axis cs:101,-0.429233085875456)
--(axis cs:100,-0.434694691259925)
--(axis cs:99,-0.434505054053864)
--(axis cs:98,-0.435530420850027)
--(axis cs:97,-0.445622870937891)
--(axis cs:96,-0.445292216602794)
--(axis cs:95,-0.451209909985473)
--(axis cs:94,-0.45321593438215)
--(axis cs:93,-0.45593274049098)
--(axis cs:92,-0.452443250735084)
--(axis cs:91,-0.451892320482713)
--(axis cs:90,-0.450479145907902)
--(axis cs:89,-0.450092251403755)
--(axis cs:88,-0.451155656471248)
--(axis cs:87,-0.452226405607896)
--(axis cs:86,-0.4485648934615)
--(axis cs:85,-0.452906411082105)
--(axis cs:84,-0.453640366945133)
--(axis cs:83,-0.455269251054031)
--(axis cs:82,-0.457784632084449)
--(axis cs:81,-0.462162894438798)
--(axis cs:80,-0.463955047900509)
--(axis cs:79,-0.463085770637798)
--(axis cs:78,-0.459190932768094)
--(axis cs:77,-0.460704782454292)
--(axis cs:76,-0.465175533924986)
--(axis cs:75,-0.463691589149337)
--(axis cs:74,-0.461816656993909)
--(axis cs:73,-0.46234921079617)
--(axis cs:72,-0.46376101363382)
--(axis cs:71,-0.459097233557714)
--(axis cs:70,-0.458374897982107)
--(axis cs:69,-0.460477128904542)
--(axis cs:68,-0.459554489080711)
--(axis cs:67,-0.466612785266233)
--(axis cs:66,-0.468986304215303)
--(axis cs:65,-0.468148737092312)
--(axis cs:64,-0.470211642225615)
--(axis cs:63,-0.47319650327832)
--(axis cs:62,-0.46847298805505)
--(axis cs:61,-0.480942764256852)
--(axis cs:60,-0.483256127578554)
--(axis cs:59,-0.492278926849443)
--(axis cs:58,-0.486136165034806)
--(axis cs:57,-0.490462346418285)
--(axis cs:56,-0.514905827888593)
--(axis cs:55,-0.515375698637365)
--(axis cs:54,-0.522549011994855)
--(axis cs:53,-0.527965300305298)
--(axis cs:52,-0.544748882066245)
--(axis cs:51,-0.539848799430722)
--(axis cs:50,-0.541960315487831)
--(axis cs:49,-0.544089457855204)
--(axis cs:48,-0.547463446233998)
--(axis cs:47,-0.54763254959184)
--(axis cs:46,-0.55184792191151)
--(axis cs:45,-0.548151665021695)
--(axis cs:44,-0.569882726573183)
--(axis cs:43,-0.564662943560904)
--(axis cs:42,-0.570375173014323)
--(axis cs:41,-0.573047316311205)
--(axis cs:40,-0.578258021446116)
--(axis cs:39,-0.572725642504433)
--(axis cs:38,-0.5771672154851)
--(axis cs:37,-0.598349777182494)
--(axis cs:36,-0.593428006625554)
--(axis cs:35,-0.593112491552832)
--(axis cs:34,-0.608703652325327)
--(axis cs:33,-0.624969403237977)
--(axis cs:32,-0.627019700129528)
--(axis cs:31,-0.62162802891128)
--(axis cs:30,-0.614102604420965)
--(axis cs:29,-0.627997829958166)
--(axis cs:28,-0.645169043706363)
--(axis cs:27,-0.656836103816365)
--(axis cs:26,-0.661908228338636)
--(axis cs:25,-0.671706570931296)
--(axis cs:24,-0.694078748212082)
--(axis cs:23,-0.737682186349934)
--(axis cs:22,-0.736243239567568)
--(axis cs:21,-0.736896307782321)
--(axis cs:20,-0.734822070639133)
--(axis cs:19,-0.753587812924711)
--(axis cs:18,-0.750985565556622)
--(axis cs:17,-0.806084003363389)
--(axis cs:16,-0.827227144927686)
--(axis cs:15,-0.884805684294962)
--(axis cs:14,-0.904844065173877)
--(axis cs:13,-0.905597960583288)
--(axis cs:12,-0.948434591203851)
--(axis cs:11,-0.981080854481005)
--(axis cs:10,-1.0001992992207)
--(axis cs:9,-1.01153942250324)
--(axis cs:8,-1.03018850780276)
--(axis cs:7,-1.1272109554157)
--(axis cs:6,-1.17812995442589)
--(axis cs:5,-1.25417232368971)
--(axis cs:4,-1.3489157308698)
--(axis cs:3,-1.48625636880888)
--(axis cs:2,-1.45631284049175)
--(axis cs:1,-1.38206652621817)
--(axis cs:0,-1.4505048292877)
--cycle;

\path [fill=color2, fill opacity=0.2]
(axis cs:0,-1.4505048292877)
--(axis cs:0,-1.4505048292877)
--(axis cs:1,-1.38085276683566)
--(axis cs:2,-1.23454406072237)
--(axis cs:3,-1.15033991884709)
--(axis cs:4,-1.21037290093521)
--(axis cs:5,-1.22148797781121)
--(axis cs:6,-1.1601395230631)
--(axis cs:7,-1.11146932797326)
--(axis cs:8,-1.12273399021976)
--(axis cs:9,-1.1413323742966)
--(axis cs:10,-1.12237155864699)
--(axis cs:11,-1.09423493716615)
--(axis cs:12,-1.12164031348319)
--(axis cs:13,-1.12575912603753)
--(axis cs:14,-1.09883356689123)
--(axis cs:15,-1.08876498846794)
--(axis cs:16,-1.1041458533306)
--(axis cs:17,-1.11211692490192)
--(axis cs:18,-1.09432980466461)
--(axis cs:19,-1.08302303730539)
--(axis cs:20,-1.09579224923246)
--(axis cs:21,-1.10150908848465)
--(axis cs:22,-1.08880560273105)
--(axis cs:23,-1.07555445104827)
--(axis cs:24,-1.09055246617785)
--(axis cs:25,-1.08666509974172)
--(axis cs:26,-1.07897196168256)
--(axis cs:27,-1.07125945815043)
--(axis cs:28,-1.08298843619428)
--(axis cs:29,-1.08864322738494)
--(axis cs:30,-1.07681201807713)
--(axis cs:31,-1.06724845383548)
--(axis cs:32,-1.07518697137041)
--(axis cs:33,-1.08155827677106)
--(axis cs:34,-1.07679490661118)
--(axis cs:35,-1.07121485547423)
--(axis cs:36,-1.07894810897821)
--(axis cs:37,-1.07159323890009)
--(axis cs:38,-1.06871256786985)
--(axis cs:39,-1.06412337173092)
--(axis cs:40,-1.07206018759193)
--(axis cs:41,-1.07236831345582)
--(axis cs:42,-1.05857319962191)
--(axis cs:43,-1.05663241745615)
--(axis cs:44,-1.06291716245631)
--(axis cs:45,-1.06260261781126)
--(axis cs:46,-1.05200319386039)
--(axis cs:47,-1.04614433487047)
--(axis cs:48,-1.04737054861657)
--(axis cs:49,-1.04716619607201)
--(axis cs:50,-1.03353957165447)
--(axis cs:51,-1.02848646378312)
--(axis cs:52,-1.02904612562397)
--(axis cs:53,-1.02581843109953)
--(axis cs:54,-1.02230053808971)
--(axis cs:55,-1.01280960596796)
--(axis cs:56,-1.01234701022611)
--(axis cs:57,-1.00571703532156)
--(axis cs:58,-0.993436616826611)
--(axis cs:59,-0.993598388073977)
--(axis cs:60,-0.993550241453709)
--(axis cs:61,-0.992308142506541)
--(axis cs:62,-0.990946521089478)
--(axis cs:63,-0.990321225366209)
--(axis cs:64,-0.991372881259121)
--(axis cs:65,-0.990477797918624)
--(axis cs:66,-0.987573359520385)
--(axis cs:67,-0.987003840203869)
--(axis cs:68,-0.987723314800088)
--(axis cs:69,-0.98733460655479)
--(axis cs:70,-0.991176130038974)
--(axis cs:71,-0.990307463083897)
--(axis cs:72,-0.990435802271201)
--(axis cs:73,-0.984746555576465)
--(axis cs:74,-0.985530756946814)
--(axis cs:75,-0.984155418332513)
--(axis cs:76,-0.98408117565284)
--(axis cs:77,-0.982483282713845)
--(axis cs:78,-0.983186660398809)
--(axis cs:79,-0.984887611267424)
--(axis cs:80,-0.984795905867583)
--(axis cs:81,-0.983665977196366)
--(axis cs:82,-0.983007390342662)
--(axis cs:83,-0.982812136840437)
--(axis cs:84,-0.982670516448215)
--(axis cs:85,-0.980832984055765)
--(axis cs:86,-0.981753821595034)
--(axis cs:87,-0.979891850922018)
--(axis cs:88,-0.979466266895794)
--(axis cs:89,-0.97815348600059)
--(axis cs:90,-0.976384895949875)
--(axis cs:91,-0.974408177164926)
--(axis cs:92,-0.975110055987742)
--(axis cs:93,-0.970410931902711)
--(axis cs:94,-0.96991704717476)
--(axis cs:95,-0.963486442607413)
--(axis cs:96,-0.963526873343555)
--(axis cs:97,-0.963263369798657)
--(axis cs:98,-0.956686290046097)
--(axis cs:99,-0.955386565363202)
--(axis cs:100,-0.955332384877917)
--(axis cs:101,-0.954918392078269)
--(axis cs:102,-0.952920059670776)
--(axis cs:103,-0.953040657438372)
--(axis cs:104,-0.954013999688021)
--(axis cs:105,-0.953571390587136)
--(axis cs:106,-0.94834622873134)
--(axis cs:107,-0.948215208939862)
--(axis cs:108,-0.948687328341855)
--(axis cs:109,-0.948654943808855)
--(axis cs:110,-0.949055240615378)
--(axis cs:111,-0.948633819368375)
--(axis cs:112,-0.949032072935236)
--(axis cs:113,-0.947910204013295)
--(axis cs:114,-0.948515571632109)
--(axis cs:115,-0.946047352508685)
--(axis cs:116,-0.946288421011596)
--(axis cs:117,-0.945326426405838)
--(axis cs:118,-0.944081519905453)
--(axis cs:119,-0.94294523657606)
--(axis cs:119,-0.929642270775601)
--(axis cs:119,-0.929642270775601)
--(axis cs:118,-0.930439505473281)
--(axis cs:117,-0.9307598030776)
--(axis cs:116,-0.931695096072749)
--(axis cs:115,-0.931234920424789)
--(axis cs:114,-0.933703730380906)
--(axis cs:113,-0.932159987416107)
--(axis cs:112,-0.933110979205578)
--(axis cs:111,-0.932820082893469)
--(axis cs:110,-0.933097212644758)
--(axis cs:109,-0.932666992519068)
--(axis cs:108,-0.932576417630654)
--(axis cs:107,-0.931955123129483)
--(axis cs:106,-0.932540081577781)
--(axis cs:105,-0.940044157289167)
--(axis cs:104,-0.940788317059613)
--(axis cs:103,-0.939889391502794)
--(axis cs:102,-0.939313946457233)
--(axis cs:101,-0.941471030257232)
--(axis cs:100,-0.941708937428244)
--(axis cs:99,-0.941832383958741)
--(axis cs:98,-0.943109166025646)
--(axis cs:97,-0.9497091682273)
--(axis cs:96,-0.950039737520699)
--(axis cs:95,-0.949603488801412)
--(axis cs:94,-0.951850700873112)
--(axis cs:93,-0.952267422457445)
--(axis cs:92,-0.955418368547453)
--(axis cs:91,-0.955262198172143)
--(axis cs:90,-0.956997787326171)
--(axis cs:89,-0.960831958555912)
--(axis cs:88,-0.962036099819558)
--(axis cs:87,-0.962317603815917)
--(axis cs:86,-0.963900582483223)
--(axis cs:85,-0.962881490414091)
--(axis cs:84,-0.964592618914611)
--(axis cs:83,-0.964783854542699)
--(axis cs:82,-0.964989419846357)
--(axis cs:81,-0.964388511367718)
--(axis cs:80,-0.965673187269988)
--(axis cs:79,-0.965470747552048)
--(axis cs:78,-0.964956865292539)
--(axis cs:77,-0.964444149292749)
--(axis cs:76,-0.96537165783291)
--(axis cs:75,-0.965167480875609)
--(axis cs:74,-0.966406231154594)
--(axis cs:73,-0.966368544516417)
--(axis cs:72,-0.969444138963185)
--(axis cs:71,-0.969024248896602)
--(axis cs:70,-0.968789436610337)
--(axis cs:69,-0.96540675043981)
--(axis cs:68,-0.965786902923739)
--(axis cs:67,-0.964506260595009)
--(axis cs:66,-0.964491697551688)
--(axis cs:65,-0.965448976380361)
--(axis cs:64,-0.965106891145615)
--(axis cs:63,-0.964079822097659)
--(axis cs:62,-0.964096226342759)
--(axis cs:61,-0.963649883224286)
--(axis cs:60,-0.965078319459492)
--(axis cs:59,-0.965018518211352)
--(axis cs:58,-0.965078500434803)
--(axis cs:57,-0.981218312580083)
--(axis cs:56,-0.986942938219808)
--(axis cs:55,-0.987013299738883)
--(axis cs:54,-0.995930050771441)
--(axis cs:53,-0.999780549708577)
--(axis cs:52,-1.00279646913011)
--(axis cs:51,-1.00054264381762)
--(axis cs:50,-1.00504783365043)
--(axis cs:49,-1.01420543074064)
--(axis cs:48,-1.01343909960854)
--(axis cs:47,-1.01305469523851)
--(axis cs:46,-1.01820951934265)
--(axis cs:45,-1.02922423533191)
--(axis cs:44,-1.03839653250785)
--(axis cs:43,-1.03060481111257)
--(axis cs:42,-1.02893954987218)
--(axis cs:41,-1.0392289711225)
--(axis cs:40,-1.03677681950709)
--(axis cs:39,-1.02716346210884)
--(axis cs:38,-1.03254014060755)
--(axis cs:37,-1.03583946487769)
--(axis cs:36,-1.04566002658268)
--(axis cs:35,-1.03683577784215)
--(axis cs:34,-1.04283952414362)
--(axis cs:33,-1.04552384644526)
--(axis cs:32,-1.03935479830785)
--(axis cs:31,-1.03246812354116)
--(axis cs:30,-1.04090974164429)
--(axis cs:29,-1.05154420840433)
--(axis cs:28,-1.04408057120063)
--(axis cs:27,-1.03138901602755)
--(axis cs:26,-1.03849246624379)
--(axis cs:25,-1.04429098935893)
--(axis cs:24,-1.05856409637447)
--(axis cs:23,-1.04223323250308)
--(axis cs:22,-1.05403563555346)
--(axis cs:21,-1.06508232318892)
--(axis cs:20,-1.05708584234709)
--(axis cs:19,-1.04346627848142)
--(axis cs:18,-1.05446854232298)
--(axis cs:17,-1.06217983352811)
--(axis cs:16,-1.04977884034676)
--(axis cs:15,-1.02983385772898)
--(axis cs:14,-1.07370360713307)
--(axis cs:13,-1.09883416915379)
--(axis cs:12,-1.08838421424949)
--(axis cs:11,-1.05820749632964)
--(axis cs:10,-1.07096403035097)
--(axis cs:9,-1.08469548489487)
--(axis cs:8,-1.06922505990017)
--(axis cs:7,-1.07812101682194)
--(axis cs:6,-1.12202716746159)
--(axis cs:5,-1.1709835905134)
--(axis cs:4,-1.15646860920924)
--(axis cs:3,-1.08295955418963)
--(axis cs:2,-1.17564932900261)
--(axis cs:1,-1.33597563286963)
--(axis cs:0,-1.4505048292877)
--cycle;

\addplot [semithick, color0]
table {%
0 -0.23964254329077
1 -0.371026151150558
2 -0.487976013714705
3 -0.281598704637683
4 -0.105234281621492
5 -0.0787483515933256
6 -0.0928558143763361
7 -0.10832663381373
8 -0.091811143987696
9 -0.160675558367854
10 -0.137460781897149
11 -0.190237213601083
12 -0.129432550213056
13 -0.136213384500082
14 -0.13896586493837
15 -0.178454121026244
16 -0.182053439982981
17 -0.21529798708333
18 -0.198982988667897
19 -0.199864957728277
20 -0.188427922854522
21 -0.215338257624283
22 -0.225527384030127
23 -0.213611433597249
24 -0.225850890701444
25 -0.236487770251786
26 -0.224222038656909
27 -0.223950621914715
28 -0.224491722651821
29 -0.224487472904474
30 -0.222971767038988
31 -0.240090710642281
32 -0.240077129813447
33 -0.247792795378062
34 -0.245784446683671
35 -0.244974388994555
36 -0.252396645000165
37 -0.251658949874671
38 -0.263715790617449
39 -0.269465400004761
40 -0.268738013255639
41 -0.261061362974612
42 -0.252788145630786
43 -0.24566897354011
44 -0.226871365007146
45 -0.232861330094172
46 -0.231683442029922
47 -0.242914577869806
48 -0.248561175488833
49 -0.253638147159263
50 -0.25748226515533
51 -0.251368078780056
52 -0.245864666451742
53 -0.235168527678072
54 -0.239068884719538
55 -0.248563689677124
56 -0.253318675142475
57 -0.262407237429489
58 -0.260968141728555
59 -0.260229045531004
60 -0.255302130446552
61 -0.254803132055985
62 -0.257896545276663
63 -0.265910317418346
64 -0.269813202707327
65 -0.264911770870132
66 -0.259544625292818
67 -0.2670627626024
68 -0.270722691567613
69 -0.278004585799351
70 -0.265088153174311
71 -0.264415017295161
72 -0.260240663981246
73 -0.263513219319477
74 -0.26236638427621
75 -0.269056010303422
76 -0.265038233875838
77 -0.260952086024256
78 -0.266786562642215
79 -0.269622982685562
80 -0.269252853804144
81 -0.265314632434002
82 -0.267536000295746
83 -0.270228430841152
84 -0.269868596869971
85 -0.269339476327663
86 -0.268283856956936
87 -0.267696794603737
88 -0.261090447557269
89 -0.24832479941903
90 -0.253528655896216
91 -0.253127500613382
92 -0.249972228836673
93 -0.246741872162373
94 -0.24887814892865
95 -0.251428378274288
96 -0.24553460022474
97 -0.245318576041135
98 -0.244633545889955
99 -0.249894990054499
100 -0.252565315511407
101 -0.252288834283171
102 -0.254205332639244
103 -0.253843957997285
104 -0.253708706428652
105 -0.250808800721329
106 -0.247579868164963
107 -0.247293224662283
108 -0.242086240100868
109 -0.236869757669711
110 -0.233882760007732
111 -0.236202558276502
112 -0.238710481311242
113 -0.23370664527056
114 -0.235584745441761
115 -0.233032674995203
116 -0.237874730585958
117 -0.240116552860741
118 -0.239590368307914
119 -0.244016931591551
};
\addlegendentry{Hindsight}
\addplot [semithick, color1]
table {%
0 -1.4505048292877
1 -1.38206652621817
2 -1.45631284049175
3 -1.48625636880888
4 -1.3489157308698
5 -1.25417232368971
6 -1.17812995442589
7 -1.1272109554157
8 -1.05969838939197
9 -1.03887625499323
10 -1.0275249117911
11 -1.00612933267054
12 -0.976104728313357
13 -0.946134642451997
14 -0.943256620365975
15 -0.923051041478795
16 -0.868646384410009
17 -0.849970083222496
18 -0.806318734723744
19 -0.807056417514851
20 -0.793367914833292
21 -0.792780977240382
22 -0.789698140788322
23 -0.788909800019824
24 -0.744542684317754
25 -0.724535820699049
26 -0.714356277893851
27 -0.706475702472996
28 -0.690378007328781
29 -0.674844214758869
30 -0.65880087205937
31 -0.664816746728951
32 -0.668899668922421
33 -0.664002906866183
34 -0.650102840700704
35 -0.638061716511114
36 -0.634859990691952
37 -0.638691445878724
38 -0.619886870784144
39 -0.609805258261428
40 -0.614433256330989
41 -0.611766030204307
42 -0.607996777524177
43 -0.59910190462307
44 -0.603556377389524
45 -0.583224133808742
46 -0.586174167958407
47 -0.58424762625323
48 -0.582928243293003
49 -0.58131787360735
50 -0.578736898927965
51 -0.578650248280341
52 -0.582818228107381
53 -0.571031199835994
54 -0.56387609785851
55 -0.554165372410727
56 -0.553558924845664
57 -0.528707600539975
58 -0.522754939362592
59 -0.528230838064648
60 -0.518910450971448
61 -0.518644310778527
62 -0.504588455739277
63 -0.508747666779981
64 -0.504607587957071
65 -0.504574287518484
66 -0.503755091659382
67 -0.503607225107185
68 -0.495724326560405
69 -0.495819832738262
70 -0.492580026137638
71 -0.492595583637253
72 -0.496800482205421
73 -0.493118911232799
74 -0.492539834615811
75 -0.49255510213547
76 -0.492906306798673
77 -0.489463459663196
78 -0.488959974920101
79 -0.492482699762904
80 -0.492817454842401
81 -0.489543635030071
82 -0.485784080283617
83 -0.48257977420733
84 -0.479721689137724
85 -0.479985503903684
86 -0.476508665825187
87 -0.479852635103814
88 -0.477153645308162
89 -0.474344874029531
90 -0.474073921182599
91 -0.475162585792044
92 -0.475640384452367
93 -0.478883096189994
94 -0.475710648767724
95 -0.473011812529263
96 -0.467719806476172
97 -0.468073785583664
98 -0.459639963944493
99 -0.458015096995219
100 -0.458624827333901
101 -0.450876281699977
102 -0.442936953171878
103 -0.446863606136436
104 -0.439557589958642
105 -0.442773605975912
106 -0.445382037124203
107 -0.448394962357723
108 -0.451616597373831
109 -0.447022980138182
110 -0.442142710201713
111 -0.440129135734514
112 -0.435877437412788
113 -0.434022058402425
114 -0.434158285415958
115 -0.433671439577533
116 -0.434405853437572
117 -0.430269521641088
118 -0.43271999405455
119 -0.430919180519724
};
\addlegendentry{Our algorithm}
\addplot [semithick, color2]
table {%
0 -1.4505048292877
1 -1.35841419985265
2 -1.20509669486249
3 -1.11664973651836
4 -1.18342075507223
5 -1.19623578416231
6 -1.14108334526235
7 -1.0947951723976
8 -1.09597952505997
9 -1.11301392959573
10 -1.09666779449898
11 -1.0762212167479
12 -1.10501226386634
13 -1.11229664759566
14 -1.08626858701215
15 -1.05929942309846
16 -1.07696234683868
17 -1.08714837921501
18 -1.0743991734938
19 -1.06324465789341
20 -1.07643904578978
21 -1.08329570583679
22 -1.07142061914226
23 -1.05889384177568
24 -1.07455828127616
25 -1.06547804455032
26 -1.05873221396317
27 -1.05132423708899
28 -1.06353450369746
29 -1.07009371789463
30 -1.05886087986071
31 -1.04985828868832
32 -1.05727088483913
33 -1.06354106160816
34 -1.0598172153774
35 -1.05402531665819
36 -1.06230406778045
37 -1.05371635188889
38 -1.0506263542387
39 -1.04564341691988
40 -1.05441850354951
41 -1.05579864228916
42 -1.04375637474704
43 -1.04361861428436
44 -1.05065684748208
45 -1.04591342657159
46 -1.03510635660152
47 -1.02959951505449
48 -1.03040482411255
49 -1.03068581340632
50 -1.01929370265245
51 -1.01451455380037
52 -1.01592129737704
53 -1.01279949040405
54 -1.00911529443057
55 -0.999911452853419
56 -0.999644974222957
57 -0.993467673950823
58 -0.979257558630707
59 -0.979308453142665
60 -0.9793142804566
61 -0.977979012865414
62 -0.977521373716118
63 -0.977200523731934
64 -0.978239886202368
65 -0.977963387149492
66 -0.976032528536037
67 -0.975755050399439
68 -0.976755108861913
69 -0.9763706784973
70 -0.979982783324656
71 -0.97966585599025
72 -0.979939970617193
73 -0.975557550046441
74 -0.975968494050704
75 -0.974661449604061
76 -0.974726416742875
77 -0.973463716003297
78 -0.974071762845674
79 -0.975179179409736
80 -0.975234546568785
81 -0.974027244282042
82 -0.973998405094509
83 -0.973797995691568
84 -0.973631567681413
85 -0.971857237234928
86 -0.972827202039128
87 -0.971104727368968
88 -0.970751183357676
89 -0.969492722278251
90 -0.966691341638023
91 -0.964835187668534
92 -0.965264212267598
93 -0.961339177180078
94 -0.960883874023936
95 -0.956544965704412
96 -0.956783305432127
97 -0.956486269012979
98 -0.949897728035871
99 -0.948609474660971
100 -0.948520661153081
101 -0.948194711167751
102 -0.946117003064004
103 -0.946465024470583
104 -0.947401158373817
105 -0.946807773938151
106 -0.940443155154561
107 -0.940085166034672
108 -0.940631872986254
109 -0.940660968163961
110 -0.941076226630068
111 -0.940726951130922
112 -0.941071526070407
113 -0.940035095714701
114 -0.941109651006508
115 -0.938641136466737
116 -0.938991758542173
117 -0.938043114741719
118 -0.937260512689367
119 -0.936293753675831
};
\addlegendentry{Q-learning}
\end{axis}

\end{tikzpicture}

%% file: avg_regret_ps_error.tex
% This file was created by tikzplotlib v0.9.2.
\begin{tikzpicture}

\definecolor{color0}{rgb}{0.12156862745098,0.466666666666667,0.705882352941177}
\definecolor{color1}{rgb}{1,0.498039215686275,0.0549019607843137}
\definecolor{color2}{rgb}{0.172549019607843,0.627450980392157,0.172549019607843}
\definecolor{color3}{rgb}{0.83921568627451,0.152941176470588,0.156862745098039}

\begin{axis}[
legend cell align={left},
legend style={fill opacity=0.8, draw opacity=1, text opacity=1, draw=white!80!black},
tick align=outside,
tick pos=both,
x grid style={white!69.0196078431373!black},
xlabel={Episode},
xmin=-4.95, xmax=103.95,
xtick style={color=black},
xtick={0,10,20,30,40,50,60,70,80,90,100},
xticklabels={0,2,4,6,8,10,12,14,16,18,20},
y grid style={white!69.0196078431373!black},
ylabel={Average Regret},
ymin=-0.0291288382169883, ymax=0.401187341946905,
ytick style={color=black}
]
\path [fill=color0, fill opacity=0.2]
(axis cs:0,0.0465863543005446)
--(axis cs:0,0.0135106210481085)
--(axis cs:1,0.109489113549858)
--(axis cs:2,-0.00502355730044766)
--(axis cs:3,0.0674328584800052)
--(axis cs:4,0.0536572601986685)
--(axis cs:5,0.0447143834988904)
--(axis cs:6,0.0375269883726301)
--(axis cs:7,0.0297736351080319)
--(axis cs:8,0.0244969308744626)
--(axis cs:9,0.0168680096796691)
--(axis cs:10,0.0153345542542447)
--(axis cs:11,0.0134854032487755)
--(axis cs:12,0.00764893275379543)
--(axis cs:13,0.0117613867074614)
--(axis cs:14,0.00987718647216097)
--(axis cs:15,0.00925986231765092)
--(axis cs:16,0.00485457161382915)
--(axis cs:17,0.00101590647640121)
--(axis cs:18,0.00225893461436119)
--(axis cs:19,0.00455389933283675)
--(axis cs:20,0.00433704698365404)
--(axis cs:21,0.00367308289098653)
--(axis cs:22,6.48831039339118e-05)
--(axis cs:23,-0.00227815197413748)
--(axis cs:24,-0.00144627343836824)
--(axis cs:25,-0.00139064753689254)
--(axis cs:26,6.86093691984746e-05)
--(axis cs:27,0.0023493698771486)
--(axis cs:28,0.00195510867711626)
--(axis cs:29,0.0016353602404487)
--(axis cs:30,0.00158260668430518)
--(axis cs:31,0.00513850537954941)
--(axis cs:32,0.00284193601531759)
--(axis cs:33,0.00349282654310761)
--(axis cs:34,0.00467913342462723)
--(axis cs:35,0.00454915749616535)
--(axis cs:36,0.00611274753475525)
--(axis cs:37,0.00357165258428708)
--(axis cs:38,0.00529234778655207)
--(axis cs:39,0.00409863238684575)
--(axis cs:40,0.00399866574326415)
--(axis cs:41,0.00483097770240027)
--(axis cs:42,0.00362645368495266)
--(axis cs:43,0.00498960636431942)
--(axis cs:44,0.00471947329216687)
--(axis cs:45,0.00461687604668499)
--(axis cs:46,0.00445774340825576)
--(axis cs:47,0.00486579519989771)
--(axis cs:48,0.0040993283764251)
--(axis cs:49,0.0052350372285076)
--(axis cs:50,0.00513238943971333)
--(axis cs:51,0.00335751871698999)
--(axis cs:52,0.00365126032823924)
--(axis cs:53,0.00523090334091826)
--(axis cs:54,0.00473043977835022)
--(axis cs:55,0.00464596763945111)
--(axis cs:56,0.0053420937121688)
--(axis cs:57,0.00521514937528066)
--(axis cs:58,0.00636878121145973)
--(axis cs:59,0.00758267507379124)
--(axis cs:60,0.00745836892504056)
--(axis cs:61,0.00706509426653367)
--(axis cs:62,0.00648203713405045)
--(axis cs:63,0.00661572509510782)
--(axis cs:64,0.00650142415136333)
--(axis cs:65,0.00640291772482752)
--(axis cs:66,0.00563830295769643)
--(axis cs:67,0.00507657283531777)
--(axis cs:68,0.00654494137728513)
--(axis cs:69,0.00566342735004137)
--(axis cs:70,0.00558366076764642)
--(axis cs:71,0.0060140516628633)
--(axis cs:72,0.00527563303941839)
--(axis cs:73,0.0052971059617532)
--(axis cs:74,0.00512017679808628)
--(axis cs:75,0.00505280605074302)
--(axis cs:76,0.00387517506937161)
--(axis cs:77,0.00389536170170848)
--(axis cs:78,0.00440315698952483)
--(axis cs:79,0.00398907533210382)
--(axis cs:80,0.00393982748849759)
--(axis cs:81,0.00411394215282219)
--(axis cs:82,0.0032597772981297)
--(axis cs:83,0.00331692488985083)
--(axis cs:84,0.00189043512472079)
--(axis cs:85,0.00186845332094496)
--(axis cs:86,0.00181564622384924)
--(axis cs:87,0.000358705876608262)
--(axis cs:88,-0.000723380807399795)
--(axis cs:89,-0.000967350486580395)
--(axis cs:90,-0.000956720261453139)
--(axis cs:91,-0.00100298942839196)
--(axis cs:92,-0.000752309612292321)
--(axis cs:93,-0.000281318682565281)
--(axis cs:94,0.000203247553819547)
--(axis cs:95,0.000201130391800594)
--(axis cs:96,0.00100376337127068)
--(axis cs:97,0.0018942523878412)
--(axis cs:98,0.00268030289104435)
--(axis cs:99,0.00271849991318786)
--(axis cs:99,0.0153860336032124)
--(axis cs:99,0.0153860336032124)
--(axis cs:98,0.0161590314882708)
--(axis cs:97,0.0159860681240359)
--(axis cs:96,0.0155280705675506)
--(axis cs:95,0.0145681927117376)
--(axis cs:94,0.0147215421087033)
--(axis cs:93,0.0140373559336229)
--(axis cs:92,0.0136994551041454)
--(axis cs:91,0.0136303412101866)
--(axis cs:90,0.014213986423141)
--(axis cs:89,0.0143719196056204)
--(axis cs:88,0.0145999151013931)
--(axis cs:87,0.0154110368909843)
--(axis cs:86,0.0166494343708994)
--(axis cs:85,0.0167629952981528)
--(axis cs:84,0.0169602070075429)
--(axis cs:83,0.0188039245943909)
--(axis cs:82,0.0193160812671686)
--(axis cs:81,0.0205053080950389)
--(axis cs:80,0.0206482612935703)
--(axis cs:79,0.02090636455974)
--(axis cs:78,0.0213043255854707)
--(axis cs:77,0.0213660194890899)
--(axis cs:76,0.0215545966046894)
--(axis cs:75,0.023228609243326)
--(axis cs:74,0.023538324033237)
--(axis cs:73,0.023853284255579)
--(axis cs:72,0.0236989853549716)
--(axis cs:71,0.0243231665326176)
--(axis cs:70,0.0237988296506155)
--(axis cs:69,0.0241388129313386)
--(axis cs:68,0.0252903709854331)
--(axis cs:67,0.0244427254824314)
--(axis cs:66,0.0244598208369062)
--(axis cs:65,0.0253727294698692)
--(axis cs:64,0.0257630791540211)
--(axis cs:63,0.0261056843828787)
--(axis cs:62,0.0260698970724668)
--(axis cs:61,0.026729208823386)
--(axis cs:60,0.0272750038015958)
--(axis cs:59,0.027729587198289)
--(axis cs:58,0.0273946299955283)
--(axis cs:57,0.0260994023314987)
--(axis cs:56,0.0267215000876241)
--(axis cs:55,0.0262424667134281)
--(axis cs:54,0.0267196024718541)
--(axis cs:53,0.0276902053458446)
--(axis cs:52,0.0248124999234512)
--(axis cs:51,0.0244378035110525)
--(axis cs:50,0.0258490536672973)
--(axis cs:49,0.0263660347406433)
--(axis cs:48,0.0257620035843966)
--(axis cs:47,0.0270739133500715)
--(axis cs:46,0.0274297858285863)
--(axis cs:45,0.0286496848499623)
--(axis cs:44,0.0292863445132947)
--(axis cs:43,0.0296056462737772)
--(axis cs:42,0.0286570036357867)
--(axis cs:41,0.028779750134738)
--(axis cs:40,0.0269533760058814)
--(axis cs:39,0.0276272104060284)
--(axis cs:38,0.0289960646155382)
--(axis cs:37,0.0277646817752874)
--(axis cs:36,0.0318422527897887)
--(axis cs:35,0.0297865887894499)
--(axis cs:34,0.0306376341834342)
--(axis cs:33,0.0293928703659718)
--(axis cs:32,0.0287039897197144)
--(axis cs:31,0.0325108334552082)
--(axis cs:30,0.029242913360356)
--(axis cs:29,0.0302176771390346)
--(axis cs:28,0.0321570747486639)
--(axis cs:27,0.0333298465206761)
--(axis cs:26,0.0314963472665773)
--(axis cs:25,0.0299698087004119)
--(axis cs:24,0.0311686010484284)
--(axis cs:23,0.0303300725964232)
--(axis cs:22,0.031505443168554)
--(axis cs:21,0.038224705438824)
--(axis cs:20,0.0418019477877266)
--(axis cs:19,0.0438920451771129)
--(axis cs:18,0.0419251821642931)
--(axis cs:17,0.0415497386975989)
--(axis cs:16,0.0477472515419649)
--(axis cs:15,0.0535204008030684)
--(axis cs:14,0.057088427523273)
--(axis cs:13,0.0599258286526922)
--(axis cs:12,0.0577110508772232)
--(axis cs:11,0.0644903935348443)
--(axis cs:10,0.0700259604941129)
--(axis cs:9,0.0770285565435241)
--(axis cs:8,0.089140880077187)
--(axis cs:7,0.107235574575553)
--(axis cs:6,0.128310393092415)
--(axis cs:5,0.151385734013478)
--(axis cs:4,0.181662880816173)
--(axis cs:3,0.207571967638077)
--(axis cs:2,0.147238160287012)
--(axis cs:1,0.205815857089827)
--(axis cs:0,0.0465863543005446)
--cycle;

\path [fill=color1, fill opacity=0.2]
(axis cs:0,0.477082061030364)
--(axis cs:0,0.28069875126689)
--(axis cs:1,0.194816013336662)
--(axis cs:2,0.234730655317451)
--(axis cs:3,0.203399332206198)
--(axis cs:4,0.226017558518939)
--(axis cs:5,0.18754746535843)
--(axis cs:6,0.180881813879676)
--(axis cs:7,0.174724141366062)
--(axis cs:8,0.161441931628846)
--(axis cs:9,0.146534637068616)
--(axis cs:10,0.133188629530907)
--(axis cs:11,0.124076102771918)
--(axis cs:12,0.115138057405035)
--(axis cs:13,0.0964616511942302)
--(axis cs:14,0.0847161682177444)
--(axis cs:15,0.0795145310351269)
--(axis cs:16,0.0768286939527947)
--(axis cs:17,0.0735238641218892)
--(axis cs:18,0.0654947429147431)
--(axis cs:19,0.0662322078590701)
--(axis cs:20,0.0639632617127656)
--(axis cs:21,0.0664325590980717)
--(axis cs:22,0.0649356820126964)
--(axis cs:23,0.0615817625231519)
--(axis cs:24,0.0573412762140834)
--(axis cs:25,0.0551523730575152)
--(axis cs:26,0.0556390852080904)
--(axis cs:27,0.052661444104553)
--(axis cs:28,0.0483645808687417)
--(axis cs:29,0.0455839645562322)
--(axis cs:30,0.044816512237831)
--(axis cs:31,0.04401612518734)
--(axis cs:32,0.042718211198508)
--(axis cs:33,0.0410304364152454)
--(axis cs:34,0.0395663965003763)
--(axis cs:35,0.0384673299309214)
--(axis cs:36,0.0368582993818574)
--(axis cs:37,0.0393244817910662)
--(axis cs:38,0.0382971284734444)
--(axis cs:39,0.0374609601047955)
--(axis cs:40,0.0365918022945698)
--(axis cs:41,0.0371581791831951)
--(axis cs:42,0.0363853442009916)
--(axis cs:43,0.0325827692741067)
--(axis cs:44,0.0297753807621752)
--(axis cs:45,0.0291598808323956)
--(axis cs:46,0.0272647485280542)
--(axis cs:47,0.0283013796717259)
--(axis cs:48,0.0276830916288344)
--(axis cs:49,0.0294897018816167)
--(axis cs:50,0.0287138240718445)
--(axis cs:51,0.0280437740803227)
--(axis cs:52,0.0275108384023275)
--(axis cs:53,0.0274787018597451)
--(axis cs:54,0.0254334896424812)
--(axis cs:55,0.0246847386091875)
--(axis cs:56,0.0233212920371842)
--(axis cs:57,0.0234569606645335)
--(axis cs:58,0.0226878200420912)
--(axis cs:59,0.0238655598595671)
--(axis cs:60,0.0231916236555275)
--(axis cs:61,0.0224471443797941)
--(axis cs:62,0.0219887989415834)
--(axis cs:63,0.0201117929166071)
--(axis cs:64,0.0188347969472382)
--(axis cs:65,0.0185545077965639)
--(axis cs:66,0.0184014335647416)
--(axis cs:67,0.0181553495650464)
--(axis cs:68,0.0155335592488179)
--(axis cs:69,0.0152706140108343)
--(axis cs:70,0.0150555349402592)
--(axis cs:71,0.0146229039201597)
--(axis cs:72,0.0160213943089002)
--(axis cs:73,0.0162775485315721)
--(axis cs:74,0.0156985479741012)
--(axis cs:75,0.0152793871536135)
--(axis cs:76,0.0150317254205207)
--(axis cs:77,0.0142850684324029)
--(axis cs:78,0.0146395814430895)
--(axis cs:79,0.0156157571519478)
--(axis cs:80,0.0155429942773763)
--(axis cs:81,0.0147397120319007)
--(axis cs:82,0.0144431161583147)
--(axis cs:83,0.0146519780244113)
--(axis cs:84,0.0146037198823299)
--(axis cs:85,0.0143713021901674)
--(axis cs:86,0.0128798426572085)
--(axis cs:87,0.0132658609987697)
--(axis cs:88,0.0142270407208918)
--(axis cs:89,0.0141167872509611)
--(axis cs:90,0.0139616577207307)
--(axis cs:91,0.0141620631243597)
--(axis cs:92,0.0128950160632781)
--(axis cs:93,0.0112686341080353)
--(axis cs:94,0.0103241244259105)
--(axis cs:95,0.0101466299363229)
--(axis cs:96,0.00969071505307217)
--(axis cs:97,0.0106089768895316)
--(axis cs:98,0.0111266699468285)
--(axis cs:99,0.0118982339357338)
--(axis cs:99,0.0200229798813874)
--(axis cs:99,0.0200229798813874)
--(axis cs:98,0.0194033507008021)
--(axis cs:97,0.0195671093139824)
--(axis cs:96,0.0188474865502763)
--(axis cs:95,0.0192175593939223)
--(axis cs:94,0.019660155521361)
--(axis cs:93,0.0200243595743784)
--(axis cs:92,0.0211253183097045)
--(axis cs:91,0.0222579676192565)
--(axis cs:90,0.0221122913826342)
--(axis cs:89,0.022357983509108)
--(axis cs:88,0.0230560396195899)
--(axis cs:87,0.0223228343946128)
--(axis cs:86,0.021278317699999)
--(axis cs:85,0.0219211515239988)
--(axis cs:84,0.0220809933103127)
--(axis cs:83,0.0219135501892405)
--(axis cs:82,0.0218905667562039)
--(axis cs:81,0.0222513724690434)
--(axis cs:80,0.0229904134108906)
--(axis cs:79,0.0230210376543033)
--(axis cs:78,0.0224435984094271)
--(axis cs:77,0.0224088403383963)
--(axis cs:76,0.0236079953629796)
--(axis cs:75,0.0237657105698192)
--(axis cs:74,0.0242313115883888)
--(axis cs:73,0.0243150099636304)
--(axis cs:72,0.0238732663617188)
--(axis cs:71,0.0226478534130792)
--(axis cs:70,0.0228930497595711)
--(axis cs:69,0.023220093327565)
--(axis cs:68,0.0236216304903139)
--(axis cs:67,0.0254532074097221)
--(axis cs:66,0.0258984682041889)
--(axis cs:65,0.0261427753048275)
--(axis cs:64,0.026539090217832)
--(axis cs:63,0.0278187858009062)
--(axis cs:62,0.0303809902343191)
--(axis cs:61,0.0310365023235002)
--(axis cs:60,0.0313674050498189)
--(axis cs:59,0.0321443662797173)
--(axis cs:58,0.0320281159766988)
--(axis cs:57,0.0327498111460208)
--(axis cs:56,0.033218957764269)
--(axis cs:55,0.0345518584397067)
--(axis cs:54,0.035484006194451)
--(axis cs:53,0.0377347914845301)
--(axis cs:52,0.0372886418207055)
--(axis cs:51,0.0381629055749506)
--(axis cs:50,0.0391747058375545)
--(axis cs:49,0.0403618473909609)
--(axis cs:48,0.0383948639624049)
--(axis cs:47,0.0396867297183662)
--(axis cs:46,0.0379101192691999)
--(axis cs:45,0.0396773168403409)
--(axis cs:44,0.0405259680134294)
--(axis cs:43,0.0432219078350177)
--(axis cs:42,0.0468347036067632)
--(axis cs:41,0.0482255372278565)
--(axis cs:40,0.0486280236393952)
--(axis cs:39,0.0498893614775187)
--(axis cs:38,0.0508848535992577)
--(axis cs:37,0.0511623521201869)
--(axis cs:36,0.0511867233065854)
--(axis cs:35,0.0524219760469691)
--(axis cs:34,0.0539197467911682)
--(axis cs:33,0.0554332331048372)
--(axis cs:32,0.0585718875536054)
--(axis cs:31,0.0597131971009666)
--(axis cs:30,0.0610668469818776)
--(axis cs:29,0.0625994648608285)
--(axis cs:28,0.0655796674815041)
--(axis cs:27,0.0696613762970931)
--(axis cs:26,0.0728205756785432)
--(axis cs:25,0.0741502439738455)
--(axis cs:24,0.077104724740041)
--(axis cs:23,0.0813717460778327)
--(axis cs:22,0.083343867955029)
--(axis cs:21,0.0858029055992316)
--(axis cs:20,0.0855309318493988)
--(axis cs:19,0.0892594744168396)
--(axis cs:18,0.0879352813470968)
--(axis cs:17,0.0965679925743849)
--(axis cs:16,0.0957863638229162)
--(axis cs:15,0.101870421720302)
--(axis cs:14,0.10851655358765)
--(axis cs:13,0.120306351208739)
--(axis cs:12,0.137494470870836)
--(axis cs:11,0.147689222306937)
--(axis cs:10,0.16267395495122)
--(axis cs:9,0.178986007757951)
--(axis cs:8,0.192568059326242)
--(axis cs:7,0.211838205680572)
--(axis cs:6,0.227698199006857)
--(axis cs:5,0.245740897707682)
--(axis cs:4,0.293338397551661)
--(axis cs:3,0.294653996728314)
--(axis cs:2,0.32678997479442)
--(axis cs:1,0.327452158469283)
--(axis cs:0,0.477082061030364)
--cycle;

\path [fill=color2, fill opacity=0.2]
(axis cs:0,0.185557250536331)
--(axis cs:0,0.0720943065686868)
--(axis cs:1,0.153265535474425)
--(axis cs:2,0.16463342581603)
--(axis cs:3,0.250974243813687)
--(axis cs:4,0.225130696485204)
--(axis cs:5,0.20109126646188)
--(axis cs:6,0.231992616844438)
--(axis cs:7,0.213878667488926)
--(axis cs:8,0.214025576584099)
--(axis cs:9,0.203662254771302)
--(axis cs:10,0.190495209236166)
--(axis cs:11,0.177764289015172)
--(axis cs:12,0.175294505741888)
--(axis cs:13,0.163075454537152)
--(axis cs:14,0.161812238872957)
--(axis cs:15,0.151771188308526)
--(axis cs:16,0.148311305244606)
--(axis cs:17,0.144224307278989)
--(axis cs:18,0.146641699372946)
--(axis cs:19,0.14198302369256)
--(axis cs:20,0.134735726621612)
--(axis cs:21,0.134552621378684)
--(axis cs:22,0.128796727302434)
--(axis cs:23,0.127199072236696)
--(axis cs:24,0.123635789077443)
--(axis cs:25,0.118880566420618)
--(axis cs:26,0.118801508415335)
--(axis cs:27,0.114472972367579)
--(axis cs:28,0.111911147383173)
--(axis cs:29,0.104580217921808)
--(axis cs:30,0.101077069365978)
--(axis cs:31,0.103194993060278)
--(axis cs:32,0.100871253264722)
--(axis cs:33,0.0993197147142244)
--(axis cs:34,0.0969195308330277)
--(axis cs:35,0.0945628963068937)
--(axis cs:36,0.0911374910656906)
--(axis cs:37,0.0892185464036189)
--(axis cs:38,0.0860837764049537)
--(axis cs:39,0.0822423399334107)
--(axis cs:40,0.0804800590622353)
--(axis cs:41,0.0788482267282968)
--(axis cs:42,0.0766486083093047)
--(axis cs:43,0.0761948505697544)
--(axis cs:44,0.0737129801686055)
--(axis cs:45,0.0723521392678574)
--(axis cs:46,0.0705210839361511)
--(axis cs:47,0.0705727255886297)
--(axis cs:48,0.0709157731021308)
--(axis cs:49,0.0706023687794718)
--(axis cs:50,0.0693015215608664)
--(axis cs:51,0.0662564502104157)
--(axis cs:52,0.0631304418566619)
--(axis cs:53,0.0637937577845298)
--(axis cs:54,0.0625989884403423)
--(axis cs:55,0.061588894006476)
--(axis cs:56,0.0602105844658728)
--(axis cs:57,0.0588449549872746)
--(axis cs:58,0.0570053886852268)
--(axis cs:59,0.0555181097809403)
--(axis cs:60,0.054797084937831)
--(axis cs:61,0.0543490138008951)
--(axis cs:62,0.0521826259875649)
--(axis cs:63,0.0527871351197528)
--(axis cs:64,0.0518516691665746)
--(axis cs:65,0.0512668182849561)
--(axis cs:66,0.0497406342359601)
--(axis cs:67,0.0487948123354681)
--(axis cs:68,0.0482305700523408)
--(axis cs:69,0.0475909013994543)
--(axis cs:70,0.0468141978643033)
--(axis cs:71,0.0456185435637766)
--(axis cs:72,0.0455053202929399)
--(axis cs:73,0.0449351455075242)
--(axis cs:74,0.0452375746929434)
--(axis cs:75,0.0446557564694727)
--(axis cs:76,0.0438408270119692)
--(axis cs:77,0.0424024502207687)
--(axis cs:78,0.0433478196961296)
--(axis cs:79,0.043959294603796)
--(axis cs:80,0.0434462693014829)
--(axis cs:81,0.0421677209827351)
--(axis cs:82,0.0415655834309784)
--(axis cs:83,0.0397354626355861)
--(axis cs:84,0.040044713862051)
--(axis cs:85,0.0394787699389667)
--(axis cs:86,0.0378393403821432)
--(axis cs:87,0.0376011750497189)
--(axis cs:88,0.0369403368060007)
--(axis cs:89,0.0375885537401205)
--(axis cs:90,0.0373168218919265)
--(axis cs:91,0.0371691701004338)
--(axis cs:92,0.0381257348950438)
--(axis cs:93,0.0381362899521197)
--(axis cs:94,0.0366200763050453)
--(axis cs:95,0.0361531871479688)
--(axis cs:96,0.0363963412534569)
--(axis cs:97,0.0350041223795756)
--(axis cs:98,0.0353552254712451)
--(axis cs:99,0.0356928175682992)
--(axis cs:99,0.0477520354951968)
--(axis cs:99,0.0477520354951968)
--(axis cs:98,0.0476146419290657)
--(axis cs:97,0.046517310973275)
--(axis cs:96,0.0473649349398531)
--(axis cs:95,0.0472479851133486)
--(axis cs:94,0.0477072045635736)
--(axis cs:93,0.0488026907534797)
--(axis cs:92,0.0488905660442074)
--(axis cs:91,0.0477696339873206)
--(axis cs:90,0.0479474572434372)
--(axis cs:89,0.0483891174025225)
--(axis cs:88,0.0478142162363323)
--(axis cs:87,0.0481510152391393)
--(axis cs:86,0.0487689779286141)
--(axis cs:85,0.0500826727474776)
--(axis cs:84,0.0506589757204219)
--(axis cs:83,0.0502753624939751)
--(axis cs:82,0.0519953946619445)
--(axis cs:81,0.0531132393331863)
--(axis cs:80,0.0541870532587997)
--(axis cs:79,0.0548098853947137)
--(axis cs:78,0.0549692532076241)
--(axis cs:77,0.0536488549865759)
--(axis cs:76,0.0550864637682478)
--(axis cs:75,0.0559455801975574)
--(axis cs:74,0.0566542554035345)
--(axis cs:73,0.0559448881349733)
--(axis cs:72,0.0558796773455445)
--(axis cs:71,0.0563832728348153)
--(axis cs:70,0.0578426449417676)
--(axis cs:69,0.0586686040677891)
--(axis cs:68,0.0602030139062567)
--(axis cs:67,0.0614931687069408)
--(axis cs:66,0.062333518805354)
--(axis cs:65,0.0638039957138552)
--(axis cs:64,0.064620089716466)
--(axis cs:63,0.0654706270107054)
--(axis cs:62,0.0649552710577536)
--(axis cs:61,0.0656283879440642)
--(axis cs:60,0.0664088819608915)
--(axis cs:59,0.0673569732580902)
--(axis cs:58,0.0680786070602096)
--(axis cs:57,0.0694934849311553)
--(axis cs:56,0.0709729941235694)
--(axis cs:55,0.0722223369805439)
--(axis cs:54,0.0734685353728819)
--(axis cs:53,0.0755999442511358)
--(axis cs:52,0.076836727135128)
--(axis cs:51,0.0807527509443947)
--(axis cs:50,0.0830863491735787)
--(axis cs:49,0.0846979648196198)
--(axis cs:48,0.0855318731902102)
--(axis cs:47,0.0856823866295558)
--(axis cs:46,0.0866028020138652)
--(axis cs:45,0.0875046030588166)
--(axis cs:44,0.0892937208667377)
--(axis cs:43,0.0920124422872369)
--(axis cs:42,0.0918508597232557)
--(axis cs:41,0.0939541535477525)
--(axis cs:40,0.0964053386531853)
--(axis cs:39,0.0986124777139813)
--(axis cs:38,0.104126690654999)
--(axis cs:37,0.110307987140612)
--(axis cs:36,0.114568596565215)
--(axis cs:35,0.117262119830497)
--(axis cs:34,0.12035594295171)
--(axis cs:33,0.124334896490785)
--(axis cs:32,0.128284663125638)
--(axis cs:31,0.129371860476362)
--(axis cs:30,0.131293350696469)
--(axis cs:29,0.135786867591921)
--(axis cs:28,0.143765529733009)
--(axis cs:27,0.146826238764943)
--(axis cs:26,0.150098588383469)
--(axis cs:25,0.149799655807369)
--(axis cs:24,0.155791642039664)
--(axis cs:23,0.158386160115841)
--(axis cs:22,0.163862186933389)
--(axis cs:21,0.169287359728484)
--(axis cs:20,0.172014908159869)
--(axis cs:19,0.180819856376901)
--(axis cs:18,0.188425329584596)
--(axis cs:17,0.182929502430255)
--(axis cs:16,0.1895805606235)
--(axis cs:15,0.19802453613137)
--(axis cs:14,0.211106507326688)
--(axis cs:13,0.209149133340236)
--(axis cs:12,0.225666295033079)
--(axis cs:11,0.231750967533198)
--(axis cs:10,0.239386020418954)
--(axis cs:9,0.25845736554258)
--(axis cs:8,0.274570361891859)
--(axis cs:7,0.273466336016237)
--(axis cs:6,0.295071999744379)
--(axis cs:5,0.26492239938588)
--(axis cs:4,0.304548336476369)
--(axis cs:3,0.325298139324695)
--(axis cs:2,0.278570274259322)
--(axis cs:1,0.247673157479232)
--(axis cs:0,0.185557250536331)
--cycle;

\path [fill=color3, fill opacity=0.2]
(axis cs:0,0.380812381517549)
--(axis cs:0,0.177642704539464)
--(axis cs:1,0.219881356759653)
--(axis cs:2,0.220056109534208)
--(axis cs:3,0.228789018236299)
--(axis cs:4,0.241723435417573)
--(axis cs:5,0.231644824830055)
--(axis cs:6,0.23576837404722)
--(axis cs:7,0.233863853601934)
--(axis cs:8,0.229575441958127)
--(axis cs:9,0.218827973705147)
--(axis cs:10,0.210357523360619)
--(axis cs:11,0.215583435240232)
--(axis cs:12,0.206538197721714)
--(axis cs:13,0.215074063107639)
--(axis cs:14,0.20331994623237)
--(axis cs:15,0.19943950891529)
--(axis cs:16,0.195355437593509)
--(axis cs:17,0.17702222624019)
--(axis cs:18,0.17713679957081)
--(axis cs:19,0.17762585117662)
--(axis cs:20,0.17264130686006)
--(axis cs:21,0.169221905402457)
--(axis cs:22,0.166841162151851)
--(axis cs:23,0.159118138288399)
--(axis cs:24,0.154420879635827)
--(axis cs:25,0.151757830555832)
--(axis cs:26,0.149607684006852)
--(axis cs:27,0.148117118829934)
--(axis cs:28,0.146625963107497)
--(axis cs:29,0.140702127694674)
--(axis cs:30,0.13706181982104)
--(axis cs:31,0.136282576783016)
--(axis cs:32,0.133238425107708)
--(axis cs:33,0.13217418071679)
--(axis cs:34,0.128545448829352)
--(axis cs:35,0.125000968144981)
--(axis cs:36,0.121775446704138)
--(axis cs:37,0.114937275088387)
--(axis cs:38,0.113717971436847)
--(axis cs:39,0.109711261004663)
--(axis cs:40,0.108945973317754)
--(axis cs:41,0.104769762749407)
--(axis cs:42,0.104366957383989)
--(axis cs:43,0.101844689920133)
--(axis cs:44,0.101460760082268)
--(axis cs:45,0.0993586098469962)
--(axis cs:46,0.0998875422901759)
--(axis cs:47,0.0981726528703537)
--(axis cs:48,0.0955944516914354)
--(axis cs:49,0.0968075382182963)
--(axis cs:50,0.0956531745585965)
--(axis cs:51,0.0929781272375689)
--(axis cs:52,0.0894909132976144)
--(axis cs:53,0.0853976243648707)
--(axis cs:54,0.0846411401908872)
--(axis cs:55,0.0831775133383649)
--(axis cs:56,0.0817448216846049)
--(axis cs:57,0.0813249565105207)
--(axis cs:58,0.0800474637128114)
--(axis cs:59,0.0768359222391175)
--(axis cs:60,0.0754910240664606)
--(axis cs:61,0.0740943216243802)
--(axis cs:62,0.0732006744397632)
--(axis cs:63,0.0705376620373683)
--(axis cs:64,0.0691754856455489)
--(axis cs:65,0.0671192852320056)
--(axis cs:66,0.0658817149255272)
--(axis cs:67,0.064900006826559)
--(axis cs:68,0.0638012802726139)
--(axis cs:69,0.0639386131679758)
--(axis cs:70,0.0637143161215977)
--(axis cs:71,0.0622112213516185)
--(axis cs:72,0.0607513402396954)
--(axis cs:73,0.0608353528124214)
--(axis cs:74,0.059047941917549)
--(axis cs:75,0.0591615868362442)
--(axis cs:76,0.0578166265656207)
--(axis cs:77,0.0558461878663417)
--(axis cs:78,0.0557758526363006)
--(axis cs:79,0.0546277642126609)
--(axis cs:80,0.0536361402136534)
--(axis cs:81,0.0529459763635584)
--(axis cs:82,0.0520533354865017)
--(axis cs:83,0.0510136245250571)
--(axis cs:84,0.0509232993563871)
--(axis cs:85,0.0499025471299029)
--(axis cs:86,0.0485525941997946)
--(axis cs:87,0.0473522672313559)
--(axis cs:88,0.0465068378046112)
--(axis cs:89,0.0456656084487959)
--(axis cs:90,0.0458944015742474)
--(axis cs:91,0.0461815455700679)
--(axis cs:92,0.0458821663277476)
--(axis cs:93,0.045970506190494)
--(axis cs:94,0.0459784030252039)
--(axis cs:95,0.0452921141386737)
--(axis cs:96,0.0447848843457703)
--(axis cs:97,0.0443995792522948)
--(axis cs:98,0.043904750685998)
--(axis cs:99,0.0435722980239818)
--(axis cs:99,0.0647405402046886)
--(axis cs:99,0.0647405402046886)
--(axis cs:98,0.0654854926419169)
--(axis cs:97,0.0660630941661392)
--(axis cs:96,0.0673318769442988)
--(axis cs:95,0.0681525923622202)
--(axis cs:94,0.0695446156495109)
--(axis cs:93,0.0701118459643344)
--(axis cs:92,0.0703109989067055)
--(axis cs:91,0.0711659564792929)
--(axis cs:90,0.0710052635637724)
--(axis cs:89,0.0706456833879846)
--(axis cs:88,0.0718164260368635)
--(axis cs:87,0.0726386175572231)
--(axis cs:86,0.0733912845675582)
--(axis cs:85,0.0754941729598573)
--(axis cs:84,0.0770204015468592)
--(axis cs:83,0.0777336883760649)
--(axis cs:82,0.0795283480409057)
--(axis cs:81,0.0805064964138736)
--(axis cs:80,0.0810336763488364)
--(axis cs:79,0.0828641331700671)
--(axis cs:78,0.0860352550653166)
--(axis cs:77,0.085915069593162)
--(axis cs:76,0.0890287126480951)
--(axis cs:75,0.0916989572182592)
--(axis cs:74,0.0916540552149466)
--(axis cs:73,0.0931995726376846)
--(axis cs:72,0.0930697707474247)
--(axis cs:71,0.0931511270189492)
--(axis cs:70,0.0956052545468531)
--(axis cs:69,0.0958944762727588)
--(axis cs:68,0.0972865492156143)
--(axis cs:67,0.0981262503238791)
--(axis cs:66,0.0999256541644041)
--(axis cs:65,0.101658532709304)
--(axis cs:64,0.1048009947762)
--(axis cs:63,0.106611831625817)
--(axis cs:62,0.108887295986425)
--(axis cs:61,0.110929217690902)
--(axis cs:60,0.112899182317457)
--(axis cs:59,0.11470608790057)
--(axis cs:58,0.118467432953927)
--(axis cs:57,0.120343938866723)
--(axis cs:56,0.120597888630651)
--(axis cs:55,0.122219699341205)
--(axis cs:54,0.123948579350924)
--(axis cs:53,0.125804817117631)
--(axis cs:52,0.130413166353583)
--(axis cs:51,0.135120777032467)
--(axis cs:50,0.139960912180167)
--(axis cs:49,0.141236480241696)
--(axis cs:48,0.1407770873409)
--(axis cs:47,0.144476079755118)
--(axis cs:46,0.147224786879731)
--(axis cs:45,0.147813632928242)
--(axis cs:44,0.151030427687826)
--(axis cs:43,0.152778864818826)
--(axis cs:42,0.156934121448408)
--(axis cs:41,0.160720485792778)
--(axis cs:40,0.167675552388819)
--(axis cs:39,0.168625079884999)
--(axis cs:38,0.174373660969862)
--(axis cs:37,0.17618416672441)
--(axis cs:36,0.181930582651004)
--(axis cs:35,0.185664787609191)
--(axis cs:34,0.189386946108711)
--(axis cs:33,0.194668184129144)
--(axis cs:32,0.195648669248408)
--(axis cs:31,0.20073042336299)
--(axis cs:30,0.200813771538413)
--(axis cs:29,0.203431457073881)
--(axis cs:28,0.211405987899235)
--(axis cs:27,0.215212760289823)
--(axis cs:26,0.216840921696874)
--(axis cs:25,0.215136025996209)
--(axis cs:24,0.215742959958475)
--(axis cs:23,0.224266474244023)
--(axis cs:22,0.233870536914136)
--(axis cs:21,0.235465322695421)
--(axis cs:20,0.238904500949926)
--(axis cs:19,0.242121485662243)
--(axis cs:18,0.246347187481065)
--(axis cs:17,0.251379447872479)
--(axis cs:16,0.267576906578043)
--(axis cs:15,0.270549052274337)
--(axis cs:14,0.271591127989627)
--(axis cs:13,0.281297700519469)
--(axis cs:12,0.275046896101264)
--(axis cs:11,0.287318273852128)
--(axis cs:10,0.28116722934025)
--(axis cs:9,0.284803870784699)
--(axis cs:8,0.295779118794687)
--(axis cs:7,0.304931614658245)
--(axis cs:6,0.306994285028141)
--(axis cs:5,0.300125594992318)
--(axis cs:4,0.294592825397287)
--(axis cs:3,0.297592490494712)
--(axis cs:2,0.301349014403421)
--(axis cs:1,0.356633046423352)
--(axis cs:0,0.380812381517549)
--cycle;

\addplot [semithick, color3]
table {%
0 0.279227543028507
1 0.288257201591502
2 0.260702561968814
3 0.263190754365505
4 0.26815813040743
5 0.265885209911187
6 0.27138132953768
7 0.269397734130089
8 0.262677280376407
9 0.251815922244923
10 0.245762376350435
11 0.25145085454618
12 0.240792546911489
13 0.248185881813554
14 0.237455537110998
15 0.234994280594814
16 0.231466172085776
17 0.214200837056335
18 0.211741993525938
19 0.209873668419431
20 0.205772903904993
21 0.202343614048939
22 0.200355849532994
23 0.191692306266211
24 0.185081919797151
25 0.183446928276021
26 0.183224302851863
27 0.181664939559878
28 0.179015975503366
29 0.172066792384278
30 0.168937795679726
31 0.168506500073003
32 0.164443547178058
33 0.163421182422967
34 0.158966197469031
35 0.155332877877086
36 0.151853014677571
37 0.145560720906399
38 0.144045816203355
39 0.139168170444831
40 0.138310762853286
41 0.132745124271092
42 0.130650539416198
43 0.127311777369479
44 0.126245593885047
45 0.123586121387619
46 0.123556164584953
47 0.121324366312736
48 0.118185769516168
49 0.119022009229996
50 0.117807043369382
51 0.114049452135018
52 0.109952039825599
53 0.105601220741251
54 0.104294859770905
55 0.102698606339785
56 0.101171355157628
57 0.100834447688622
58 0.0992574483333694
59 0.0957710050698435
60 0.0941951031919586
61 0.0925117696576409
62 0.0910439852130941
63 0.0885747468315924
64 0.0869882402108747
65 0.0843889089706547
66 0.0829036845449657
67 0.0815131285752191
68 0.0805439147441141
69 0.0799165447203673
70 0.0796597853342254
71 0.0776811741852838
72 0.07691055549356
73 0.077017462725053
74 0.0753509985662478
75 0.0754302720272517
76 0.0734226696068579
77 0.0708806287297518
78 0.0709055538508086
79 0.068745948691364
80 0.0673349082812449
81 0.066726236388716
82 0.0657908417637037
83 0.064373656450561
84 0.0639718504516232
85 0.0626983600448801
86 0.0609719393836764
87 0.0599954423942895
88 0.0591616319207373
89 0.0581556459183903
90 0.0584498325690099
91 0.0586737510246804
92 0.0580965826172265
93 0.0580411760774142
94 0.0577615093373574
95 0.0567223532504469
96 0.0560583806450345
97 0.055231336709217
98 0.0546951216639574
99 0.0541564191143352
};
\addlegendentry{$p = 4$}
\addplot [semithick, color2]
table {%
0 0.128825778552509
1 0.200469346476828
2 0.221601850037676
3 0.288136191569191
4 0.264839516480787
5 0.23300683292388
6 0.263532308294409
7 0.243672501752581
8 0.244297969237979
9 0.231059810156941
10 0.21494061482756
11 0.204757628274185
12 0.200480400387483
13 0.186112293938694
14 0.186459373099822
15 0.174897862219948
16 0.168945932934053
17 0.163576904854622
18 0.167533514478771
19 0.16140144003473
20 0.15337531739074
21 0.151919990553584
22 0.146329457117911
23 0.142792616176268
24 0.139713715558553
25 0.134340111113994
26 0.134450048399402
27 0.130649605566261
28 0.127838338558091
29 0.120183542756864
30 0.116185210031223
31 0.11628342676832
32 0.11457795819518
33 0.111827305602505
34 0.108637736892369
35 0.105912508068695
36 0.102853043815453
37 0.0997632667721154
38 0.0951052335299764
39 0.090427408823696
40 0.0884426988577103
41 0.0864011901380246
42 0.0842497340162802
43 0.0841036464284956
44 0.0815033505176716
45 0.079928371163337
46 0.0785619429750082
47 0.0781275561090927
48 0.0782238231461705
49 0.0776501667995458
50 0.0761939353672225
51 0.0735046005774052
52 0.0699835844958949
53 0.0696968510178328
54 0.0680337619066121
55 0.0669056154935099
56 0.0655917892947211
57 0.064169219959215
58 0.0625419978727182
59 0.0614375415195153
60 0.0606029834493612
61 0.0599887008724796
62 0.0585689485226593
63 0.0591288810652291
64 0.0582358794415203
65 0.0575354069994056
66 0.056037076520657
67 0.0551439905212044
68 0.0542167919792988
69 0.0531297527336217
70 0.0523284214030355
71 0.0510009081992959
72 0.0506924988192422
73 0.0504400168212488
74 0.050945915048239
75 0.050300668333515
76 0.0494636453901085
77 0.0480256526036723
78 0.0491585364518769
79 0.0493845899992549
80 0.0488166612801413
81 0.0476404801579607
82 0.0467804890464615
83 0.0450054125647806
84 0.0453518447912365
85 0.0447807213432221
86 0.0433041591553787
87 0.0428760951444291
88 0.0423772765211665
89 0.0429888355713215
90 0.0426321395676819
91 0.0424694020438772
92 0.0435081504696256
93 0.0434694903527997
94 0.0421636404343095
95 0.0417005861306587
96 0.041880638096655
97 0.0407607166764253
98 0.0414849337001554
99 0.041722426531748
};
\addlegendentry{$p = 3$}

\addplot [semithick, color1]
table {%
0 0.378890406148627
1 0.261134085902973
2 0.280760315055935
3 0.249026664467256
4 0.2596779780353
5 0.216644181533056
6 0.204290006443267
7 0.193281173523317
8 0.177004995477544
9 0.162760322413284
10 0.147931292241064
11 0.135882662539428
12 0.126316264137935
13 0.108384001201484
14 0.0966163609026974
15 0.0906924763777145
16 0.0863075288878555
17 0.085045928348137
18 0.0767150121309199
19 0.0777458411379548
20 0.0747470967810822
21 0.0761177323486517
22 0.0741397749838627
23 0.0714767543004923
24 0.0672230004770622
25 0.0646513085156804
26 0.0642298304433168
27 0.061161410200823
28 0.0569721241751229
29 0.0540917147085303
30 0.0529416796098543
31 0.0518646611441533
32 0.0506450493760567
33 0.0482318347600413
34 0.0467430716457722
35 0.0454446529889452
36 0.0440225113442214
37 0.0452434169556265
38 0.044590991036351
39 0.0436751607911571
40 0.0426099129669825
41 0.0426918582055258
42 0.0416100239038774
43 0.0379023385545622
44 0.0351506743878023
45 0.0344185988363682
46 0.032587433898627
47 0.033994054695046
48 0.0330389777956196
49 0.0349257746362888
50 0.0339442649546995
51 0.0331033398276367
52 0.0323997401115165
53 0.0326067466721376
54 0.0304587479184661
55 0.0296182985244471
56 0.0282701249007266
57 0.0281033859052771
58 0.027357968009395
59 0.0280049630696422
60 0.0272795143526732
61 0.0267418233516472
62 0.0261848945879512
63 0.0239652893587567
64 0.0226869435825351
65 0.0223486415506957
66 0.0221499508844652
67 0.0218042784873842
68 0.0195775948695659
69 0.0192453536691997
70 0.0189742923499152
71 0.0186353786666195
72 0.0199473303353095
73 0.0202962792476013
74 0.019964929781245
75 0.0195225488617164
76 0.0193198603917501
77 0.0183469543853996
78 0.0185415899262583
79 0.0193183974031255
80 0.0192667038441334
81 0.018495542250472
82 0.0181668414572593
83 0.0182827641068259
84 0.0183423565963213
85 0.0181462268570831
86 0.0170790801786038
87 0.0177943476966913
88 0.0186415401702409
89 0.0182373853800345
90 0.0180369745516825
91 0.0182100153718081
92 0.0170101671864913
93 0.0156464968412069
94 0.0149921399736357
95 0.0146820946651226
96 0.0142691008016742
97 0.015088043101757
98 0.0152650103238153
99 0.0159606069085606
};
\addlegendentry{$p = 2$}

\addplot [semithick, color0]
table {%
0 0.0300484876743265
1 0.157652485319842
2 0.0711073014932819
3 0.137502413059041
4 0.117660070507421
5 0.0980500587561841
6 0.0829186907325224
7 0.0685046048417924
8 0.0568189054758248
9 0.0469482831115966
10 0.0426802573741788
11 0.0389878983918099
12 0.0326799918155093
13 0.0358436076800768
14 0.033482806997717
15 0.0313901315603597
16 0.026300911577897
17 0.0212828225870001
18 0.0220920583893271
19 0.0242229722549748
20 0.0230694973856903
21 0.0209488941649052
22 0.015785163136244
23 0.0140259603111429
24 0.0148611638050301
25 0.0142895805817597
26 0.0157824783178879
27 0.0178396081989123
28 0.0170560917128901
29 0.0159265186897416
30 0.0154127600223306
31 0.0188246694173788
32 0.015772962867516
33 0.0164428484545397
34 0.0176583838040307
35 0.0171678731428076
36 0.018977500162272
37 0.0156681671797873
38 0.0171442062010452
39 0.0158629213964371
40 0.0154760208745728
41 0.0168053639185692
42 0.0161417286603697
43 0.0172976263190483
44 0.0170029089027308
45 0.0166332804483236
46 0.015943764618421
47 0.0159698542749846
48 0.0149306659804109
49 0.0158005359845754
50 0.0154907215535053
51 0.0138976611140212
52 0.0142318801258452
53 0.0164605543433814
54 0.0157250211251021
55 0.0154442171764396
56 0.0160317968998964
57 0.0156572758533897
58 0.016881705603494
59 0.0176561311360401
60 0.0173666863633182
61 0.0168971515449598
62 0.0162759671032586
63 0.0163607047389933
64 0.0161322516526922
65 0.0158878235973484
66 0.0150490618973013
67 0.0147596491588746
68 0.0159176561813591
69 0.01490112014069
70 0.014691245209131
71 0.0151686090977404
72 0.014487309197195
73 0.0145751951086661
74 0.0143292504156616
75 0.0141407076470345
76 0.0127148858370305
77 0.0126306905953992
78 0.0128537412874978
79 0.0124477199459219
80 0.012294044391034
81 0.0123096251239306
82 0.0112879292826491
83 0.0110604247421209
84 0.00942532106613183
85 0.00931572430954889
86 0.00923254029737434
87 0.0078848713837963
88 0.00693826714699666
89 0.00670228455952
90 0.00662863308084395
91 0.0063136758908973
92 0.00647357274592652
93 0.00687801862552879
94 0.00746239483126142
95 0.00738466155176911
96 0.00826591696941064
97 0.00894016025593857
98 0.00941966718965758
99 0.00905226675820013
};
\addlegendentry{$p = 1$}

\end{axis}

\end{tikzpicture}

%% file: avg_regret_s_size_error.tex
% This file was created by tikzplotlib v0.9.2.
\begin{tikzpicture}

\definecolor{color0}{rgb}{0.12156862745098,0.466666666666667,0.705882352941177}
\definecolor{color1}{rgb}{1,0.498039215686275,0.0549019607843137}
\definecolor{color2}{rgb}{0.172549019607843,0.627450980392157,0.172549019607843}
\definecolor{color3}{rgb}{0.83921568627451,0.152941176470588,0.156862745098039}

\begin{axis}[
legend cell align={left},
legend style={fill opacity=0.8, draw opacity=1, text opacity=1, draw=white!80!black},
tick align=outside,
tick pos=both,
x grid style={white!69.0196078431373!black},
xlabel={Episode},
xmin=-4.95, xmax=103.95,
xtick style={color=black},
xtick={0,10,20,30,40,50,60,70,80,90,100},
xticklabels={0,2,4,6,8,10,12,14,16,18,20},
y grid style={white!69.0196078431373!black},
ylabel={Average Regret},
ymin=-0.011685941460694, ymax=0.445032584303179,
ytick style={color=black}
]
\path [fill=color0, fill opacity=0.1]
(axis cs:0,0.387226866595019)
--(axis cs:0,0.21066883248364)
--(axis cs:1,0.25475045916261)
--(axis cs:2,0.241860574275349)
--(axis cs:3,0.208868629112383)
--(axis cs:4,0.252663791442962)
--(axis cs:5,0.20404406142953)
--(axis cs:6,0.179659777556492)
--(axis cs:7,0.174230847453958)
--(axis cs:8,0.190597351333864)
--(axis cs:9,0.216066522552632)
--(axis cs:10,0.19814504103549)
--(axis cs:11,0.185025993725346)
--(axis cs:12,0.165914303524917)
--(axis cs:13,0.171214159746289)
--(axis cs:14,0.173379322846658)
--(axis cs:15,0.160864792793251)
--(axis cs:16,0.148739071713777)
--(axis cs:17,0.146083214117135)
--(axis cs:18,0.135350880154289)
--(axis cs:19,0.134983770208304)
--(axis cs:20,0.129436569941202)
--(axis cs:21,0.125961755215415)
--(axis cs:22,0.122855454536229)
--(axis cs:23,0.118938949580155)
--(axis cs:24,0.117044943924897)
--(axis cs:25,0.111696546462888)
--(axis cs:26,0.109285989334533)
--(axis cs:27,0.109371158258347)
--(axis cs:28,0.107339145942375)
--(axis cs:29,0.103895914205731)
--(axis cs:30,0.100340753319287)
--(axis cs:31,0.0950677591823213)
--(axis cs:32,0.0933521309599996)
--(axis cs:33,0.0904865371394211)
--(axis cs:34,0.0875810377618664)
--(axis cs:35,0.0852971838972473)
--(axis cs:36,0.0828642631836939)
--(axis cs:37,0.0790952569912351)
--(axis cs:38,0.0773479186041607)
--(axis cs:39,0.0750770117971955)
--(axis cs:40,0.0734024119348891)
--(axis cs:41,0.0711164944775047)
--(axis cs:42,0.067514038266876)
--(axis cs:43,0.0641283028991992)
--(axis cs:44,0.0619955035623786)
--(axis cs:45,0.0600638332924975)
--(axis cs:46,0.0606572559279793)
--(axis cs:47,0.0575921907816261)
--(axis cs:48,0.0581871118672299)
--(axis cs:49,0.0573163506570476)
--(axis cs:50,0.0553202176439586)
--(axis cs:51,0.0533806631612527)
--(axis cs:52,0.0526887697191969)
--(axis cs:53,0.0504890913382293)
--(axis cs:54,0.050190665474783)
--(axis cs:55,0.0484436120090032)
--(axis cs:56,0.0467486237911079)
--(axis cs:57,0.0436291742927693)
--(axis cs:58,0.0423985815751067)
--(axis cs:59,0.0415856053068555)
--(axis cs:60,0.0412735065112961)
--(axis cs:61,0.0407710064768033)
--(axis cs:62,0.0429204275946854)
--(axis cs:63,0.0428237839613334)
--(axis cs:64,0.0430241194877047)
--(axis cs:65,0.042235964364834)
--(axis cs:66,0.0430107392984797)
--(axis cs:67,0.0423965491312809)
--(axis cs:68,0.0415716519866606)
--(axis cs:69,0.0422925899046656)
--(axis cs:70,0.0420403344103434)
--(axis cs:71,0.0422251130662033)
--(axis cs:72,0.0406909670162696)
--(axis cs:73,0.0384532370014656)
--(axis cs:74,0.0367434430939295)
--(axis cs:75,0.0364789315237903)
--(axis cs:76,0.037221212374577)
--(axis cs:77,0.037261525834155)
--(axis cs:78,0.036491927930546)
--(axis cs:79,0.0346894023009977)
--(axis cs:80,0.0348727138307283)
--(axis cs:81,0.0342672638244861)
--(axis cs:82,0.0346250184255685)
--(axis cs:83,0.0352705936560971)
--(axis cs:84,0.0350247732133726)
--(axis cs:85,0.0348280927019257)
--(axis cs:86,0.0350834314646183)
--(axis cs:87,0.0348370654325456)
--(axis cs:88,0.0339228488350498)
--(axis cs:89,0.0340673473956115)
--(axis cs:90,0.033447438428096)
--(axis cs:91,0.0323868308478114)
--(axis cs:92,0.0324952093362622)
--(axis cs:93,0.0315583088446815)
--(axis cs:94,0.0316400173273088)
--(axis cs:95,0.0320001959898443)
--(axis cs:96,0.0314071556959312)
--(axis cs:97,0.0313835161353524)
--(axis cs:98,0.0324842472593387)
--(axis cs:99,0.0318848963976823)
--(axis cs:99,0.0502868725551929)
--(axis cs:99,0.0502868725551929)
--(axis cs:98,0.0507512367660534)
--(axis cs:97,0.0506030460150438)
--(axis cs:96,0.050665937672718)
--(axis cs:95,0.0513967447293693)
--(axis cs:94,0.051552150709518)
--(axis cs:93,0.0515748216751949)
--(axis cs:92,0.0527856553730476)
--(axis cs:91,0.0529782585704216)
--(axis cs:90,0.053563188800872)
--(axis cs:89,0.054492090414493)
--(axis cs:88,0.052916006957434)
--(axis cs:87,0.0543162096459939)
--(axis cs:86,0.0546556121182742)
--(axis cs:85,0.0551407021054349)
--(axis cs:84,0.0556913713731789)
--(axis cs:83,0.0559853950582434)
--(axis cs:82,0.054849303268641)
--(axis cs:81,0.0549231503629113)
--(axis cs:80,0.0559504815843224)
--(axis cs:79,0.0559285968187949)
--(axis cs:78,0.0576233383272746)
--(axis cs:77,0.0577631929727323)
--(axis cs:76,0.0578759318724291)
--(axis cs:75,0.0568289156493382)
--(axis cs:74,0.0574728789476817)
--(axis cs:73,0.0589484192885367)
--(axis cs:72,0.0603192547175554)
--(axis cs:71,0.0618969804562474)
--(axis cs:70,0.0632444612466471)
--(axis cs:69,0.064094400796527)
--(axis cs:68,0.0639559985682653)
--(axis cs:67,0.0645368406298564)
--(axis cs:66,0.0651037817571982)
--(axis cs:65,0.0643862549364024)
--(axis cs:64,0.0662581067506773)
--(axis cs:63,0.0656437911265715)
--(axis cs:62,0.0654569049496733)
--(axis cs:61,0.0650268179557786)
--(axis cs:60,0.0660124238334606)
--(axis cs:59,0.0670949453861637)
--(axis cs:58,0.067903922111817)
--(axis cs:57,0.0684718372117555)
--(axis cs:56,0.0701128075834792)
--(axis cs:55,0.071059596525547)
--(axis cs:54,0.0734089064172267)
--(axis cs:53,0.0748788303998379)
--(axis cs:52,0.078006549529511)
--(axis cs:51,0.0786594485498108)
--(axis cs:50,0.0819993327670956)
--(axis cs:49,0.0833311528088064)
--(axis cs:48,0.0843261090007032)
--(axis cs:47,0.0834153283838822)
--(axis cs:46,0.0865192963357794)
--(axis cs:45,0.0845739796361383)
--(axis cs:44,0.0872755524298622)
--(axis cs:43,0.0903083386377762)
--(axis cs:42,0.0931808160857526)
--(axis cs:41,0.0961135480233252)
--(axis cs:40,0.0999764557470096)
--(axis cs:39,0.10165962411233)
--(axis cs:38,0.104530280622407)
--(axis cs:37,0.107689243532413)
--(axis cs:36,0.109987126710979)
--(axis cs:35,0.111519658475507)
--(axis cs:34,0.115217305166969)
--(axis cs:33,0.116450862189997)
--(axis cs:32,0.120959717741736)
--(axis cs:31,0.124109460404976)
--(axis cs:30,0.126998613485335)
--(axis cs:29,0.131373672501711)
--(axis cs:28,0.137214476777431)
--(axis cs:27,0.14121596460309)
--(axis cs:26,0.141799387453718)
--(axis cs:25,0.142979477353342)
--(axis cs:24,0.148174985865528)
--(axis cs:23,0.151263009049154)
--(axis cs:22,0.155583408174399)
--(axis cs:21,0.160673841047289)
--(axis cs:20,0.1682154967272)
--(axis cs:19,0.176349701948692)
--(axis cs:18,0.174487674777071)
--(axis cs:17,0.186717430943125)
--(axis cs:16,0.192428296699212)
--(axis cs:15,0.207749569912979)
--(axis cs:14,0.222911396242155)
--(axis cs:13,0.213108617542411)
--(axis cs:12,0.212461514050427)
--(axis cs:11,0.229620261114211)
--(axis cs:10,0.250189784893474)
--(axis cs:9,0.276396798900401)
--(axis cs:8,0.25509927873561)
--(axis cs:7,0.230404491843802)
--(axis cs:6,0.234021656987871)
--(axis cs:5,0.245807636815255)
--(axis cs:4,0.308695977308808)
--(axis cs:3,0.285563510067735)
--(axis cs:2,0.339312407038243)
--(axis cs:1,0.408174660471144)
--(axis cs:0,0.387226866595019)
--cycle;

\path [fill=color1, fill opacity=0.1]
(axis cs:0,0.343071442391953)
--(axis cs:0,0.122500971986797)
--(axis cs:1,0.222007911459845)
--(axis cs:2,0.326639724371732)
--(axis cs:3,0.298330744922557)
--(axis cs:4,0.28862124445827)
--(axis cs:5,0.251780442177515)
--(axis cs:6,0.243262135837256)
--(axis cs:7,0.24138541542709)
--(axis cs:8,0.231504257988644)
--(axis cs:9,0.211609548335312)
--(axis cs:10,0.193155131070351)
--(axis cs:11,0.196892928853298)
--(axis cs:12,0.195824163253485)
--(axis cs:13,0.181574611725467)
--(axis cs:14,0.179883072292499)
--(axis cs:15,0.167111030588824)
--(axis cs:16,0.163686874290114)
--(axis cs:17,0.160471543337006)
--(axis cs:18,0.162035660636213)
--(axis cs:19,0.165200318354711)
--(axis cs:20,0.157007870734378)
--(axis cs:21,0.156595533456507)
--(axis cs:22,0.147819718694419)
--(axis cs:23,0.138662036114287)
--(axis cs:24,0.134817599015933)
--(axis cs:25,0.129746109655881)
--(axis cs:26,0.124975030513545)
--(axis cs:27,0.118904169285725)
--(axis cs:28,0.114397635244)
--(axis cs:29,0.114525005972139)
--(axis cs:30,0.110127696791839)
--(axis cs:31,0.10814661535188)
--(axis cs:32,0.103287135042309)
--(axis cs:33,0.103412181226487)
--(axis cs:34,0.103142120883089)
--(axis cs:35,0.101780138389198)
--(axis cs:36,0.0992585919761101)
--(axis cs:37,0.0938225140631419)
--(axis cs:38,0.0921053614494615)
--(axis cs:39,0.0940644527266924)
--(axis cs:40,0.0904589722486852)
--(axis cs:41,0.0885674855309577)
--(axis cs:42,0.086417839570488)
--(axis cs:43,0.083692997629163)
--(axis cs:44,0.0861527943649993)
--(axis cs:45,0.0837202655562694)
--(axis cs:46,0.0818155539845115)
--(axis cs:47,0.0793881642328827)
--(axis cs:48,0.0784643578907413)
--(axis cs:49,0.0774017202866449)
--(axis cs:50,0.076014867839809)
--(axis cs:51,0.0736423044692642)
--(axis cs:52,0.0730479191859239)
--(axis cs:53,0.072655321087312)
--(axis cs:54,0.0709711370183746)
--(axis cs:55,0.0697037952859036)
--(axis cs:56,0.0686189648286561)
--(axis cs:57,0.0666115986507321)
--(axis cs:58,0.0647456146958573)
--(axis cs:59,0.0650285150570453)
--(axis cs:60,0.063962473826602)
--(axis cs:61,0.0637907815781273)
--(axis cs:62,0.0618682360650292)
--(axis cs:63,0.0610437100119775)
--(axis cs:64,0.0611105239597103)
--(axis cs:65,0.0604985222539821)
--(axis cs:66,0.058771164886868)
--(axis cs:67,0.0579108117899673)
--(axis cs:68,0.0558216663115856)
--(axis cs:69,0.0547850582631013)
--(axis cs:70,0.0547775674181977)
--(axis cs:71,0.0540761965632537)
--(axis cs:72,0.0531880504451042)
--(axis cs:73,0.0538391750310469)
--(axis cs:74,0.0536012206252673)
--(axis cs:75,0.0521291311060619)
--(axis cs:76,0.050537697048006)
--(axis cs:77,0.0492671338870905)
--(axis cs:78,0.0487574059285953)
--(axis cs:79,0.0483041797899655)
--(axis cs:80,0.0484573987140342)
--(axis cs:81,0.0479555608784609)
--(axis cs:82,0.0473277090374857)
--(axis cs:83,0.046702871200513)
--(axis cs:84,0.0450300806027435)
--(axis cs:85,0.0441208493034434)
--(axis cs:86,0.0434379623493736)
--(axis cs:87,0.0425544132738989)
--(axis cs:88,0.041461657339644)
--(axis cs:89,0.0412237484733517)
--(axis cs:90,0.0403928164128738)
--(axis cs:91,0.0393565476551415)
--(axis cs:92,0.0396658599509374)
--(axis cs:93,0.0398399642675614)
--(axis cs:94,0.0401774396520649)
--(axis cs:95,0.0396314856362219)
--(axis cs:96,0.0385228435679664)
--(axis cs:97,0.0378871586161226)
--(axis cs:98,0.0375769630507153)
--(axis cs:99,0.0366329388274507)
--(axis cs:99,0.0601363849461394)
--(axis cs:99,0.0601363849461394)
--(axis cs:98,0.0606090831450863)
--(axis cs:97,0.0611081386330947)
--(axis cs:96,0.0623443225849168)
--(axis cs:95,0.0632016413899359)
--(axis cs:94,0.063846587034355)
--(axis cs:93,0.0635882873303618)
--(axis cs:92,0.0626813706724444)
--(axis cs:91,0.0628848729000091)
--(axis cs:90,0.0640013919530235)
--(axis cs:89,0.0650099494743072)
--(axis cs:88,0.0651677571241181)
--(axis cs:87,0.0656890952636828)
--(axis cs:86,0.0668039062314089)
--(axis cs:85,0.0680927982348719)
--(axis cs:84,0.0692232551888773)
--(axis cs:83,0.0711865995017044)
--(axis cs:82,0.0719231106073795)
--(axis cs:81,0.07258452552415)
--(axis cs:80,0.0730594925844014)
--(axis cs:79,0.0734557008156517)
--(axis cs:78,0.0748426795812728)
--(axis cs:77,0.0756251821836173)
--(axis cs:76,0.0767878972855886)
--(axis cs:75,0.0787250593569157)
--(axis cs:74,0.080320417586287)
--(axis cs:73,0.0810227798606017)
--(axis cs:72,0.0808074291853342)
--(axis cs:71,0.0818614998915522)
--(axis cs:70,0.0828082624492582)
--(axis cs:69,0.0834562240888502)
--(axis cs:68,0.0846412375383302)
--(axis cs:67,0.0866545623140679)
--(axis cs:66,0.0888938842010284)
--(axis cs:65,0.0904563008010074)
--(axis cs:64,0.0920081834392215)
--(axis cs:63,0.0912068503869337)
--(axis cs:62,0.0917591216770731)
--(axis cs:61,0.0943214668692649)
--(axis cs:60,0.0946992152528477)
--(axis cs:59,0.0962775355070619)
--(axis cs:58,0.0958623745986347)
--(axis cs:57,0.097028726602734)
--(axis cs:56,0.0989421340296291)
--(axis cs:55,0.100625752467973)
--(axis cs:54,0.102455311603755)
--(axis cs:53,0.104263515560672)
--(axis cs:52,0.104897988782391)
--(axis cs:51,0.10545322430773)
--(axis cs:50,0.108011326548635)
--(axis cs:49,0.109842691643468)
--(axis cs:48,0.112275021499099)
--(axis cs:47,0.113974063538216)
--(axis cs:46,0.116414923554629)
--(axis cs:45,0.118720085759301)
--(axis cs:44,0.12169551792939)
--(axis cs:43,0.116576131676327)
--(axis cs:42,0.12021270329618)
--(axis cs:41,0.123677469977056)
--(axis cs:40,0.126575377472289)
--(axis cs:39,0.130947699405404)
--(axis cs:38,0.131420273876625)
--(axis cs:37,0.133302975081328)
--(axis cs:36,0.137834856561421)
--(axis cs:35,0.141600113979735)
--(axis cs:34,0.145150075990917)
--(axis cs:33,0.144483478858314)
--(axis cs:32,0.143427611635067)
--(axis cs:31,0.149921736566195)
--(axis cs:30,0.153022306921024)
--(axis cs:29,0.158760594288862)
--(axis cs:28,0.160353126994841)
--(axis cs:27,0.165002763379734)
--(axis cs:26,0.173228205989194)
--(axis cs:25,0.179830226434393)
--(axis cs:24,0.186332112067497)
--(axis cs:23,0.186654559858691)
--(axis cs:22,0.198025851964643)
--(axis cs:21,0.204315136553959)
--(axis cs:20,0.204649301424146)
--(axis cs:19,0.214492379878806)
--(axis cs:18,0.211193597439541)
--(axis cs:17,0.205816382912876)
--(axis cs:16,0.203653622290296)
--(axis cs:15,0.20940151227168)
--(axis cs:14,0.226936971194891)
--(axis cs:13,0.236682394818345)
--(axis cs:12,0.254150699919992)
--(axis cs:11,0.254083389384173)
--(axis cs:10,0.243243260071546)
--(axis cs:9,0.263806955212556)
--(axis cs:8,0.293668684635777)
--(axis cs:7,0.310273851990594)
--(axis cs:6,0.30603426213992)
--(axis cs:5,0.314886004431085)
--(axis cs:4,0.360841796260957)
--(axis cs:3,0.380054291565417)
--(axis cs:2,0.425335009628521)
--(axis cs:1,0.352316704523459)
--(axis cs:0,0.343071442391953)
--cycle;

\path [fill=color2, fill opacity=0.1]
(axis cs:0,0.379750939606219)
--(axis cs:0,0.180521532906051)
--(axis cs:1,0.326019418073315)
--(axis cs:2,0.320699590931124)
--(axis cs:3,0.303978348385001)
--(axis cs:4,0.322817429877973)
--(axis cs:5,0.283752860144512)
--(axis cs:6,0.272667742210707)
--(axis cs:7,0.250033746883156)
--(axis cs:8,0.233894080266418)
--(axis cs:9,0.220309000083878)
--(axis cs:10,0.204311360676692)
--(axis cs:11,0.1983877647722)
--(axis cs:12,0.207855427083966)
--(axis cs:13,0.207976069578024)
--(axis cs:14,0.188625971330743)
--(axis cs:15,0.183613094237627)
--(axis cs:16,0.179754939489101)
--(axis cs:17,0.169556670978904)
--(axis cs:18,0.163614859841563)
--(axis cs:19,0.160183619147884)
--(axis cs:20,0.153859158604879)
--(axis cs:21,0.144369785129518)
--(axis cs:22,0.141661531372396)
--(axis cs:23,0.138555638136566)
--(axis cs:24,0.133323697729956)
--(axis cs:25,0.128916642176036)
--(axis cs:26,0.120825900844903)
--(axis cs:27,0.117480861500386)
--(axis cs:28,0.114251053209467)
--(axis cs:29,0.110350764106648)
--(axis cs:30,0.107265792866628)
--(axis cs:31,0.105344259896882)
--(axis cs:32,0.103764924668215)
--(axis cs:33,0.103584698610149)
--(axis cs:34,0.10168727569962)
--(axis cs:35,0.0987789338757727)
--(axis cs:36,0.0976583843071949)
--(axis cs:37,0.0951739909007266)
--(axis cs:38,0.0941367104401182)
--(axis cs:39,0.0894204281188422)
--(axis cs:40,0.0867304818995553)
--(axis cs:41,0.0844737668687479)
--(axis cs:42,0.0835252447871195)
--(axis cs:43,0.0794384553834832)
--(axis cs:44,0.078857206899607)
--(axis cs:45,0.0771997106970814)
--(axis cs:46,0.0754193334131463)
--(axis cs:47,0.0749389439691833)
--(axis cs:48,0.0747144841772826)
--(axis cs:49,0.0714085817740672)
--(axis cs:50,0.0695971535272006)
--(axis cs:51,0.0688216150572718)
--(axis cs:52,0.0678516772465142)
--(axis cs:53,0.067825569190097)
--(axis cs:54,0.0674356486251607)
--(axis cs:55,0.0662314406139971)
--(axis cs:56,0.0623876356744649)
--(axis cs:57,0.0600495338519829)
--(axis cs:58,0.059725276959466)
--(axis cs:59,0.0599782564282413)
--(axis cs:60,0.0592796013583979)
--(axis cs:61,0.0583182279479918)
--(axis cs:62,0.058326412635659)
--(axis cs:63,0.0587920123574677)
--(axis cs:64,0.0569178085983169)
--(axis cs:65,0.0564113483084514)
--(axis cs:66,0.0560228333908916)
--(axis cs:67,0.0563309224210977)
--(axis cs:68,0.056559868475079)
--(axis cs:69,0.0561544833426479)
--(axis cs:70,0.0553635751265543)
--(axis cs:71,0.0557624476625739)
--(axis cs:72,0.0539637071484029)
--(axis cs:73,0.0524150927245126)
--(axis cs:74,0.0518330668361165)
--(axis cs:75,0.0510794844225984)
--(axis cs:76,0.0499135846411573)
--(axis cs:77,0.0491915631996025)
--(axis cs:78,0.0485554269119806)
--(axis cs:79,0.0470027356215603)
--(axis cs:80,0.0465973654706054)
--(axis cs:81,0.0459082985431355)
--(axis cs:82,0.0452240307314395)
--(axis cs:83,0.0452698612336119)
--(axis cs:84,0.0427324190882661)
--(axis cs:85,0.0422355304942165)
--(axis cs:86,0.0402490084644926)
--(axis cs:87,0.0401845892760668)
--(axis cs:88,0.0400484797587366)
--(axis cs:89,0.0405645080507637)
--(axis cs:90,0.0401187442260301)
--(axis cs:91,0.0386711986597038)
--(axis cs:92,0.0372000534978854)
--(axis cs:93,0.0362653550939179)
--(axis cs:94,0.0355849201649239)
--(axis cs:95,0.0352142439132059)
--(axis cs:96,0.0342373232778697)
--(axis cs:97,0.0341885264568161)
--(axis cs:98,0.0332802008944911)
--(axis cs:99,0.0327593202551123)
--(axis cs:99,0.0513178408222314)
--(axis cs:99,0.0513178408222314)
--(axis cs:98,0.0519280100700491)
--(axis cs:97,0.0531317316714156)
--(axis cs:96,0.053249460212759)
--(axis cs:95,0.0543052677871881)
--(axis cs:94,0.054876902184948)
--(axis cs:93,0.0555008100318856)
--(axis cs:92,0.0565029975554504)
--(axis cs:91,0.0575447020040956)
--(axis cs:90,0.0590076527354497)
--(axis cs:89,0.0596632933213992)
--(axis cs:88,0.0591884822556351)
--(axis cs:87,0.0597648502862975)
--(axis cs:86,0.0604689725921977)
--(axis cs:85,0.0626190637421819)
--(axis cs:84,0.0633557586097369)
--(axis cs:83,0.0663303986269035)
--(axis cs:82,0.0669633274190339)
--(axis cs:81,0.0673951207287772)
--(axis cs:80,0.0686959295361632)
--(axis cs:79,0.0693903850597221)
--(axis cs:78,0.0707221553284765)
--(axis cs:77,0.0714523883689145)
--(axis cs:76,0.0723511613219402)
--(axis cs:75,0.0735318521852272)
--(axis cs:74,0.0745674942119689)
--(axis cs:73,0.0751299493811505)
--(axis cs:72,0.0768160574718752)
--(axis cs:71,0.0784146710711861)
--(axis cs:70,0.0784293733193319)
--(axis cs:69,0.0795497929381795)
--(axis cs:68,0.0798287411899013)
--(axis cs:67,0.0800223712232592)
--(axis cs:66,0.0796212586643661)
--(axis cs:65,0.0804333614690976)
--(axis cs:64,0.0814026983645595)
--(axis cs:63,0.0821166781255811)
--(axis cs:62,0.0820634938329662)
--(axis cs:61,0.0826323020739859)
--(axis cs:60,0.0841891873204358)
--(axis cs:59,0.0852668831042844)
--(axis cs:58,0.0861984020469201)
--(axis cs:57,0.0864318011159799)
--(axis cs:56,0.0888637239875344)
--(axis cs:55,0.0929536914344253)
--(axis cs:54,0.0946437585514148)
--(axis cs:53,0.0942937964394275)
--(axis cs:52,0.0950202450196961)
--(axis cs:51,0.0965505564423423)
--(axis cs:50,0.0976234218037256)
--(axis cs:49,0.100025289452327)
--(axis cs:48,0.102706661301279)
--(axis cs:47,0.103457704542691)
--(axis cs:46,0.10465851865812)
--(axis cs:45,0.107232185749966)
--(axis cs:44,0.109605975814564)
--(axis cs:43,0.111402897404647)
--(axis cs:42,0.11584507959849)
--(axis cs:41,0.11793480850707)
--(axis cs:40,0.121096504323672)
--(axis cs:39,0.124688338246028)
--(axis cs:38,0.130654684169345)
--(axis cs:37,0.131902597036169)
--(axis cs:36,0.134964121089469)
--(axis cs:35,0.137540452449366)
--(axis cs:34,0.14157142114867)
--(axis cs:33,0.142824886715689)
--(axis cs:32,0.145471216395433)
--(axis cs:31,0.148647062697068)
--(axis cs:30,0.149728684706011)
--(axis cs:29,0.154199154202674)
--(axis cs:28,0.160079160015616)
--(axis cs:27,0.164446976109466)
--(axis cs:26,0.170544114699294)
--(axis cs:25,0.180146187324691)
--(axis cs:24,0.18191679525692)
--(axis cs:23,0.18987753505106)
--(axis cs:22,0.196257908874026)
--(axis cs:21,0.20015381254655)
--(axis cs:20,0.209727212876718)
--(axis cs:19,0.213523184691183)
--(axis cs:18,0.214860126006098)
--(axis cs:17,0.220869761651432)
--(axis cs:16,0.232213843435242)
--(axis cs:15,0.23451051752011)
--(axis cs:14,0.236345019589856)
--(axis cs:13,0.249216797206197)
--(axis cs:12,0.251388505046402)
--(axis cs:11,0.24494777833603)
--(axis cs:10,0.244380893780107)
--(axis cs:9,0.252732086470005)
--(axis cs:8,0.264858242744469)
--(axis cs:7,0.293488575472561)
--(axis cs:6,0.316561673795374)
--(axis cs:5,0.336703928910582)
--(axis cs:4,0.376857731726285)
--(axis cs:3,0.40477764495177)
--(axis cs:2,0.409806437856377)
--(axis cs:1,0.421825968726098)
--(axis cs:0,0.379750939606219)
--cycle;

\path [fill=color3, fill opacity=0.1]
(axis cs:0,0.27960849031505)
--(axis cs:0,0.104725378397307)
--(axis cs:1,0.211065665100803)
--(axis cs:2,0.246066314759719)
--(axis cs:3,0.222127429926496)
--(axis cs:4,0.221495504518665)
--(axis cs:5,0.22623542996535)
--(axis cs:6,0.218440600049854)
--(axis cs:7,0.209968882047772)
--(axis cs:8,0.201809851497833)
--(axis cs:9,0.187929023851481)
--(axis cs:10,0.188370694045482)
--(axis cs:11,0.174042224637257)
--(axis cs:12,0.165891677905311)
--(axis cs:13,0.152453451186774)
--(axis cs:14,0.150225406778957)
--(axis cs:15,0.137193911242275)
--(axis cs:16,0.130821219348992)
--(axis cs:17,0.132601299745294)
--(axis cs:18,0.126161558076147)
--(axis cs:19,0.131210544689563)
--(axis cs:20,0.121644609124613)
--(axis cs:21,0.118851384460278)
--(axis cs:22,0.113294887709619)
--(axis cs:23,0.110639158865498)
--(axis cs:24,0.108635054678617)
--(axis cs:25,0.106773073236731)
--(axis cs:26,0.102715946580726)
--(axis cs:27,0.0985387474439222)
--(axis cs:28,0.0959961207487717)
--(axis cs:29,0.0949529082069686)
--(axis cs:30,0.0923349531087866)
--(axis cs:31,0.0921792568972471)
--(axis cs:32,0.0880349372393963)
--(axis cs:33,0.0849116608780414)
--(axis cs:34,0.0849451223491803)
--(axis cs:35,0.0859376322181798)
--(axis cs:36,0.0827784639459784)
--(axis cs:37,0.0817155263994766)
--(axis cs:38,0.0790780175897115)
--(axis cs:39,0.0770386135568675)
--(axis cs:40,0.0733349658108474)
--(axis cs:41,0.0725142762165727)
--(axis cs:42,0.0708480712201662)
--(axis cs:43,0.06744912065115)
--(axis cs:44,0.0669806180597193)
--(axis cs:45,0.0644609323281871)
--(axis cs:46,0.0623989648791074)
--(axis cs:47,0.0609100189924281)
--(axis cs:48,0.0601692536660585)
--(axis cs:49,0.061793274847095)
--(axis cs:50,0.0583812809421514)
--(axis cs:51,0.0584355107797883)
--(axis cs:52,0.0564674872917908)
--(axis cs:53,0.0536313876511719)
--(axis cs:54,0.0536211702301112)
--(axis cs:55,0.0523729330814087)
--(axis cs:56,0.0509355061495671)
--(axis cs:57,0.0497484159211278)
--(axis cs:58,0.0505959135951036)
--(axis cs:59,0.0490055610374603)
--(axis cs:60,0.0470102833464264)
--(axis cs:61,0.0462162315572136)
--(axis cs:62,0.0462544808571478)
--(axis cs:63,0.0469891344886052)
--(axis cs:64,0.0467485830784445)
--(axis cs:65,0.0463173554083277)
--(axis cs:66,0.0457098202492275)
--(axis cs:67,0.0448989389747683)
--(axis cs:68,0.0436140360039993)
--(axis cs:69,0.043629751249208)
--(axis cs:70,0.0435381895392847)
--(axis cs:71,0.0431857070918642)
--(axis cs:72,0.0438571401001332)
--(axis cs:73,0.0434002603644992)
--(axis cs:74,0.0449462753659207)
--(axis cs:75,0.0451496058789071)
--(axis cs:76,0.0445175818570048)
--(axis cs:77,0.0431716291319438)
--(axis cs:78,0.0428263037161367)
--(axis cs:79,0.0411382389898027)
--(axis cs:80,0.0414031162646493)
--(axis cs:81,0.0408121207440802)
--(axis cs:82,0.0405129747277301)
--(axis cs:83,0.0391526084595641)
--(axis cs:84,0.0391069101789311)
--(axis cs:85,0.0375373613062828)
--(axis cs:86,0.0371433111655178)
--(axis cs:87,0.0369695266875399)
--(axis cs:88,0.0369176519019842)
--(axis cs:89,0.0359622865825467)
--(axis cs:90,0.0357644843223548)
--(axis cs:91,0.03557401750806)
--(axis cs:92,0.0352672456224056)
--(axis cs:93,0.0353309070546475)
--(axis cs:94,0.0349044018107387)
--(axis cs:95,0.0355958330833452)
--(axis cs:96,0.0357780769866198)
--(axis cs:97,0.0360208966989413)
--(axis cs:98,0.0353578318622772)
--(axis cs:99,0.0355572926224325)
--(axis cs:99,0.0512562652709275)
--(axis cs:99,0.0512562652709275)
--(axis cs:98,0.05089213280306)
--(axis cs:97,0.052194817356466)
--(axis cs:96,0.0517652303929084)
--(axis cs:95,0.0516657285063102)
--(axis cs:94,0.0512235571167948)
--(axis cs:93,0.0514315889938376)
--(axis cs:92,0.0514724855789203)
--(axis cs:91,0.0512799003421112)
--(axis cs:90,0.0519331582600419)
--(axis cs:89,0.0523189430792938)
--(axis cs:88,0.0535562475722249)
--(axis cs:87,0.053280627487998)
--(axis cs:86,0.0535216759411695)
--(axis cs:85,0.0540772005589983)
--(axis cs:84,0.0548198310624212)
--(axis cs:83,0.0550888177023499)
--(axis cs:82,0.0569045385922971)
--(axis cs:81,0.0573280001689383)
--(axis cs:80,0.0576999398506379)
--(axis cs:79,0.0582493316135074)
--(axis cs:78,0.0600063747460024)
--(axis cs:77,0.0599069017481127)
--(axis cs:76,0.0612118225171277)
--(axis cs:75,0.063082316813833)
--(axis cs:74,0.0627977967043654)
--(axis cs:73,0.0622170646723266)
--(axis cs:72,0.0626093327664642)
--(axis cs:71,0.0616985234311136)
--(axis cs:70,0.0616154577448081)
--(axis cs:69,0.062558498615539)
--(axis cs:68,0.0628857535956738)
--(axis cs:67,0.0648100137854603)
--(axis cs:66,0.0651818574487788)
--(axis cs:65,0.0653210335930819)
--(axis cs:64,0.0662404758551475)
--(axis cs:63,0.0657892800295932)
--(axis cs:62,0.0650049962002295)
--(axis cs:61,0.0648893562252373)
--(axis cs:60,0.0659449729166823)
--(axis cs:59,0.0676463287870236)
--(axis cs:58,0.0691814999804618)
--(axis cs:57,0.0701801673423707)
--(axis cs:56,0.0733513218487401)
--(axis cs:55,0.0753422415171563)
--(axis cs:54,0.0771175606856264)
--(axis cs:53,0.0775408979118094)
--(axis cs:52,0.079694201147703)
--(axis cs:51,0.0807874637968372)
--(axis cs:50,0.0814675688964021)
--(axis cs:49,0.0843932666104283)
--(axis cs:48,0.0850635741559912)
--(axis cs:47,0.0866319671767495)
--(axis cs:46,0.0889491464011645)
--(axis cs:45,0.0911884769266274)
--(axis cs:44,0.0933174161236006)
--(axis cs:43,0.0941932836595415)
--(axis cs:42,0.096858395663632)
--(axis cs:41,0.0992444682305953)
--(axis cs:40,0.102067703621219)
--(axis cs:39,0.105107182496122)
--(axis cs:38,0.108342285461802)
--(axis cs:37,0.111350723359185)
--(axis cs:36,0.111756890454391)
--(axis cs:35,0.115501192422952)
--(axis cs:34,0.115242812831791)
--(axis cs:33,0.115924513966195)
--(axis cs:32,0.120483726806744)
--(axis cs:31,0.125047878369433)
--(axis cs:30,0.127842212569806)
--(axis cs:29,0.130887958285726)
--(axis cs:28,0.132421383707195)
--(axis cs:27,0.13692603805801)
--(axis cs:26,0.142307644096232)
--(axis cs:25,0.145784070173138)
--(axis cs:24,0.148367770389115)
--(axis cs:23,0.15310486032932)
--(axis cs:22,0.158423451359182)
--(axis cs:21,0.164508605054817)
--(axis cs:20,0.17065321285674)
--(axis cs:19,0.180196888509921)
--(axis cs:18,0.178943843439284)
--(axis cs:17,0.189072022179161)
--(axis cs:16,0.191165773449503)
--(axis cs:15,0.202044684236915)
--(axis cs:14,0.214531924225355)
--(axis cs:13,0.208585802015446)
--(axis cs:12,0.22335593296673)
--(axis cs:11,0.235042487972245)
--(axis cs:10,0.248385961786984)
--(axis cs:9,0.252598029712034)
--(axis cs:8,0.271465675304484)
--(axis cs:7,0.278181113200847)
--(axis cs:6,0.286396616270199)
--(axis cs:5,0.302602936387144)
--(axis cs:4,0.306925429702093)
--(axis cs:3,0.32609758187449)
--(axis cs:2,0.336310701831507)
--(axis cs:1,0.312805955896035)
--(axis cs:0,0.27960849031505)
--cycle;

\addplot [semithick, color0]
table {%
0 0.298947849539329
1 0.331462559816877
2 0.290586490656796
3 0.247216069590059
4 0.280679884375885
5 0.224925849122392
6 0.206840717272181
7 0.20231766964888
8 0.222848315034737
9 0.246231660726517
10 0.224167412964482
11 0.207323127419778
12 0.189187908787672
13 0.19216138864435
14 0.198145359544406
15 0.184307181353115
16 0.170583684206495
17 0.16640032253013
18 0.15491927746568
19 0.155666736078498
20 0.148826033334201
21 0.143317798131352
22 0.139219431355314
23 0.135100979314654
24 0.132609964895213
25 0.127338011908115
26 0.125542688394125
27 0.125293561430719
28 0.122276811359903
29 0.117634793353721
30 0.113669683402311
31 0.109588609793649
32 0.107155924350868
33 0.103468699664709
34 0.101399171464418
35 0.0984084211863772
36 0.0964256949473363
37 0.0933922502618242
38 0.0909390996132838
39 0.0883683179547627
40 0.0866894338409494
41 0.0836150212504149
42 0.0803474271763143
43 0.0772183207684877
44 0.0746355279961204
45 0.0723189064643179
46 0.0735882761318794
47 0.0705037595827541
48 0.0712566104339666
49 0.070323751732927
50 0.0686597752055271
51 0.0660200558555317
52 0.065347659624354
53 0.0626839608690336
54 0.0617997859460048
55 0.0597516042672751
56 0.0584307156872936
57 0.0560505057522624
58 0.0551512518434619
59 0.0543402753465096
60 0.0536429651723783
61 0.052898912216291
62 0.0541886662721794
63 0.0542337875439525
64 0.054641113119191
65 0.0533111096506182
66 0.054057260527839
67 0.0534666948805686
68 0.0527638252774629
69 0.0531934953505963
70 0.0526423978284953
71 0.0520610467612253
72 0.0505051108669125
73 0.0487008281450012
74 0.0471081610208056
75 0.0466539235865642
76 0.0475485721235031
77 0.0475123594034437
78 0.0470576331289103
79 0.0453089995598963
80 0.0454115977075253
81 0.0445952070936987
82 0.0447371608471047
83 0.0456279943571702
84 0.0453580722932758
85 0.0449843974036803
86 0.0448695217914462
87 0.0445766375392698
88 0.0434194278962419
89 0.0442797189050523
90 0.043505313614484
91 0.0426825447091165
92 0.0426404323546549
93 0.0415665652599382
94 0.0415960840184134
95 0.0416984703596068
96 0.0410365466843246
97 0.0409932810751981
98 0.0416177420126961
99 0.0410858844764376
};
\addlegendentry{$|\states| = 9$}
\addplot [semithick, color1]
table {%
0 0.232786207189375
1 0.287162307991652
2 0.375987367000126
3 0.339192518243987
4 0.324731520359614
5 0.2833332233043
6 0.274648198988588
7 0.275829633708842
8 0.26258647131221
9 0.237708251773934
10 0.218199195570948
11 0.225488159118736
12 0.224987431586738
13 0.209128503271906
14 0.203410021743695
15 0.188256271430252
16 0.183670248290205
17 0.183143963124941
18 0.186614629037877
19 0.189846349116759
20 0.180828586079262
21 0.180455335005233
22 0.172922785329531
23 0.162658297986489
24 0.160574855541715
25 0.154788168045137
26 0.14910161825137
27 0.14195346633273
28 0.13737538111942
29 0.1366428001305
30 0.131575001856432
31 0.129034175959038
32 0.123357373338688
33 0.1239478300424
34 0.124146098437003
35 0.121690126184467
36 0.118546724268765
37 0.113562744572235
38 0.111762817663043
39 0.112506076066048
40 0.108517174860487
41 0.106122477754007
42 0.103315271433334
43 0.100134564652745
44 0.103924156147195
45 0.101220175657785
46 0.0991152387695705
47 0.0966811138855495
48 0.0953696896949203
49 0.0936222059650565
50 0.0920130971942222
51 0.0895477643884972
52 0.0889729539841574
53 0.0884594183239922
54 0.0867132243110647
55 0.0851647738769385
56 0.0837805494291426
57 0.0818201626267331
58 0.080303994647246
59 0.0806530252820536
60 0.0793308445397249
61 0.0790561242236961
62 0.0768136788710511
63 0.0761252801994556
64 0.0765593536994659
65 0.0754774115274948
66 0.0738325245439482
67 0.0722826870520176
68 0.0702314519249579
69 0.0691206411759757
70 0.068792914933728
71 0.067968848227403
72 0.0669977398152192
73 0.0674309774458243
74 0.0669608191057771
75 0.0654270952314888
76 0.0636627971667973
77 0.0624461580353539
78 0.061800042754934
79 0.0608799403028086
80 0.0607584456492178
81 0.0602700432013054
82 0.0596254098224326
83 0.0589447353511087
84 0.0571266678958104
85 0.0561068237691577
86 0.0551209342903913
87 0.0541217542687909
88 0.053314707231881
89 0.0531168489738294
90 0.0521971041829486
91 0.0511207102775753
92 0.0511736153116909
93 0.0517141257989616
94 0.0520120133432099
95 0.0514165635130789
96 0.0504335830764416
97 0.0494976486246086
98 0.0490930230979008
99 0.0483846618867951
};
\addlegendentry{$|\states| = 15$}
\addplot [semithick, color2]
table {%
0 0.280136236256135
1 0.373922693399707
2 0.36525301439375
3 0.354377996668385
4 0.349837580802129
5 0.310228394527547
6 0.29461470800304
7 0.271761161177858
8 0.249376161505443
9 0.236520543276941
10 0.2243461272284
11 0.221667771554115
12 0.229621966065184
13 0.228596433392111
14 0.2124854954603
15 0.209061805878869
16 0.205984391462172
17 0.195213216315168
18 0.18923749292383
19 0.186853401919533
20 0.181793185740798
21 0.172261798838034
22 0.168959720123211
23 0.164216586593813
24 0.157620246493438
25 0.154531414750364
26 0.145685007772098
27 0.140963918804926
28 0.137165106612542
29 0.132274959154661
30 0.128497238786319
31 0.126995661296975
32 0.124618070531824
33 0.123204792662919
34 0.121629348424145
35 0.118159693162569
36 0.116311252698332
37 0.113538293968448
38 0.112395697304731
39 0.107054383182435
40 0.103913493111614
41 0.101204287687909
42 0.0996851621928048
43 0.0954206763940649
44 0.0942315913570853
45 0.0922159482235235
46 0.090038926035633
47 0.0891983242559371
48 0.0887105727392807
49 0.0857169356131972
50 0.0836102876654631
51 0.082686085749807
52 0.0814359611331052
53 0.0810596828147623
54 0.0810397035882878
55 0.0795925660242112
56 0.0756256798309996
57 0.0732406674839814
58 0.0729618395031931
59 0.0726225697662628
60 0.0717343943394169
61 0.0704752650109889
62 0.0701949532343126
63 0.0704543452415244
64 0.0691602534814382
65 0.0684223548887745
66 0.0678220460276289
67 0.0681766468221784
68 0.0681943048324901
69 0.0678521381404137
70 0.0668964742229431
71 0.06708855936688
72 0.065389882310139
73 0.0637725210528316
74 0.0632002805240427
75 0.0623056683039128
76 0.0611323729815488
77 0.0603219757842585
78 0.0596387911202286
79 0.0581965603406412
80 0.0576466475033843
81 0.0566517096359563
82 0.0560936790752367
83 0.0558001299302577
84 0.0530440888490015
85 0.0524272971181992
86 0.0503589905283452
87 0.0499747197811822
88 0.0496184810071858
89 0.0501139006860814
90 0.0495631984807399
91 0.0481079503318997
92 0.0468515255266679
93 0.0458830825629018
94 0.0452309111749359
95 0.044759755850197
96 0.0437433917453143
97 0.0436601290641158
98 0.0426041054822701
99 0.0420385805386718
};
\addlegendentry{$|\states| = 23$}
\addplot [semithick, color3]
table {%
0 0.192166934356178
1 0.261935810498419
2 0.291188508295613
3 0.274112505900493
4 0.264210467110379
5 0.264419183176247
6 0.252418608160027
7 0.24407499762431
8 0.236637763401158
9 0.220263526781757
10 0.218378327916233
11 0.204542356304751
12 0.19462380543602
13 0.18051962660111
14 0.182378665502156
15 0.169619297739595
16 0.160993496399248
17 0.160836660962228
18 0.152552700757716
19 0.155703716599742
20 0.146148910990676
21 0.141679994757547
22 0.1358591695344
23 0.131872009597409
24 0.128501412533866
25 0.126278571704935
26 0.122511795338479
27 0.117732392750966
28 0.114208752227984
29 0.112920433246347
30 0.110088582839296
31 0.10861356763334
32 0.10425933202307
33 0.100418087422118
34 0.100093967590486
35 0.100719412320566
36 0.0972676772001844
37 0.0965331248793307
38 0.0937101515257566
39 0.0910728980264949
40 0.0877013347160334
41 0.085879372223584
42 0.0838532334418991
43 0.0808212021553457
44 0.08014901709166
45 0.0778247046274073
46 0.075674055640136
47 0.0737709930845888
48 0.0726164139110248
49 0.0730932707287617
50 0.0699244249192767
51 0.0696114872883128
52 0.0680808442197469
53 0.0655861427814906
54 0.0653693654578688
55 0.0638575872992825
56 0.0621434139991536
57 0.0599642916317492
58 0.0598887067877827
59 0.058325944912242
60 0.0564776281315544
61 0.0555527938912255
62 0.0556297385286886
63 0.0563892072590992
64 0.056494529466796
65 0.0558191945007048
66 0.0554458388490031
67 0.0548544763801143
68 0.0532498947998366
69 0.0530941249323735
70 0.0525768236420464
71 0.0524421152614889
72 0.0532332364332987
73 0.0528086625184129
74 0.0538720360351431
75 0.05411596134637
76 0.0528647021870662
77 0.0515392654400283
78 0.0514163392310696
79 0.0496937853016551
80 0.0495515280576436
81 0.0490700604565092
82 0.0487087566600136
83 0.047120713080957
84 0.0469633706206762
85 0.0458072809326405
86 0.0453324935533436
87 0.045125077087769
88 0.0452369497371046
89 0.0441406148309203
90 0.0438488212911983
91 0.0434269589250856
92 0.0433698656006629
93 0.0433812480242426
94 0.0430639794637668
95 0.0436307807948277
96 0.0437716536897641
97 0.0441078570277036
98 0.0431249823326686
99 0.04340677894668
};
\addlegendentry{$|\states| = 31$}
\end{axis}

\end{tikzpicture}